\documentclass[11pt]{article}
\usepackage{braket}
\usepackage{graphicx} 
\usepackage[utf8]{inputenc}
\usepackage{fullpage, amsmath, amsfonts, amsthm, color, graphicx, amssymb, nicefrac}
\usepackage{algpseudocode}
\usepackage{empheq}
\usepackage{textcomp}
\usepackage{pgfplots}
\usepackage{mleftright}
\pgfplotsset{compat=1.18}

\usepackage{mathtools}

\usepackage[T1]{fontenc}

\usepackage{hyperref}

\usepackage{bbm}
\usepackage{multirow}
\usepackage[sort,numbers]{natbib}
\usepackage{thmtools} 
\usepackage{thm-restate}

\usepackage{mdframed}

\usepackage{enumitem}
\usepackage{comment}
 \usepackage[normalem]{ulem}

\newlength\caseLen
\newlist{mycases}{enumerate}{1}
\setlist[mycases,1]{label=\textbf{Case~\arabic*.}, 
  labelwidth=\dimexpr-\caseLen-\labelsep\relax,leftmargin=0pt,align=right}
  
\numberwithin{table}{section}
\numberwithin{figure}{section}
\numberwithin{equation}{section}
\newtheorem{theorem}{Theorem}[section]
\newtheorem{lemma}[theorem]{Lemma}
\newtheorem{corollary}[theorem]{Corollary}
\theoremstyle{definition}
\newtheorem{algorithm}{Algorithm}[section]
\newtheorem{definition}[theorem]{Definition}
\newtheorem{remark}[theorem]{Remark}
\AtBeginEnvironment{remark}{
  \pushQED{\qed}
}
\AtEndEnvironment{remark}{\popQED\endremark}

\newcommand{\genGraphTable}{General graphs}
\newcommand{\eye}{I}

\newcommand{\degr}[1]{\mathsf{deg}\left({#1}\right)}

\newcommand{\setFunct}[2]{ \left\{ {#1} \, : \, {#2}\right\}}

\newcommand{\feasLass}[2]{F^{#1}_{#2}}
\newcommand{\lambdaMax}[1]{\lambda_{\mathrm{max}}\mleft({#1}\mright)}
\newcommand{\hyperGeom}{{}_2 \mathrm{F}_1}
\newcommand{\modSubscript}[2]{\left({#1}\right)_{\! \! \! \! \! \mod {#2}}}

\newcommand{\R}{\mathbb{R}}
\newcommand{\edge}[2]{\left\{{#1},{#2}\right\}}
\newcommand{\optTheta}{1-e^{-x/20}}
\newcommand{\qmcFullName}{Quantum Max-Cut}
\newcommand{\C}{\mathbb{C}}

\newcommand{\ratioGeneral}{0.603}
\newcommand{\ratioBipart}{0.8162}
\newcommand{\ratioBipartUnrounded}{0.8162}
\newcommand{\ratioNoTri}{0.61383}
\newcommand{\unroundedApproxOpt}{0.61383}
\newcommand{\ratioBipartUB}{0.8339}
\newcommand{\Tr}[1]{\mathrm{tr}\left({#1}\right)}

\newcommand{\edgeSubset}[1]{E_{#1}}
\newcommand{\imagUnit}{\mathbf{i}}

\newcommand{\argFunc}[1]{\mathrm{Arg}\left({#1}\right)}
\newcommand{\signFunction}[1]{\mathrm{sgn}\!\left({#1} \right)}

\newcommand{\ratioUB}{0.61392}

\usepackage{amsmath}

\usepackage[nameinlink]{cleveref}
\Crefname{enumi}{Step}{Steps}
\crefname{enumi}{step}{steps}
\renewcommand{\i}{\mathbf{i}}

\newcommand{\stkout}[1]{\ifmmode\text{\sout{\ensuremath{#1}}}\else\sout{#1}\fi}
\newenvironment{proofsketch}{
  \proof}{\endproof}

\title{Improved approximation ratios for the \qmcFullName{} problem on general, triangle-free and bipartite graphs} 
\author{Sander Gribling\thanks{CentER,  Department of Econometrics and OR, Tilburg University, The Netherlands \\
\hphantom{asd} \texttt{\{s.j.gribling,l.m.sinjorgo,r.sotirov\}@tilburguniversity.edu}} \\  \and Lennart Sinjorgo\footnotemark[1]
\and Renata Sotirov\footnotemark[1]}

\begin{document}

\date{}
\maketitle

\vspace{-3em}
\begin{abstract}
     We study polynomial-time approximation  algorithms for the \qmcFullName{} (QMC) problem. Given an edge-weighted graph $G$ on $n$ vertices, the QMC problem is to determine the largest eigenvalue of a particular $2^n \times 2^n$ matrix that corresponds to $G$. We provide a sharpened analysis of the currently best-known QMC approximation algorithm for general graphs. This algorithm achieves an approximation ratio of $0.599$, which our analysis improves to \ratioGeneral{}. Additionally, we propose two new approximation algorithms for the QMC problem on triangle-free and bipartite graphs, that achieve approximation ratios of \ratioNoTri{} and $\ratioBipart{}$, respectively. These are the best-known approximation ratios for their respective graph classes. 
\end{abstract}

\newcommand{\githubLink}{\href{https://github.com/LMSinjorgo/QMC_proofVerification}{\texttt{https://github.com/LMSinjorgo/QMC\_proofVerification}}}
\section{Introduction}
\label{section_introduction}
The \qmcFullName{} (QMC) problem is to determine the largest energy eigenvalue and corresponding eigenstate of the Hamiltonian
\begin{align}
    \label{eqn_hamiltonDef}
    H_G := \sum_{\edge{i}{j} \in E} w_{ij} H_{ij}, \text{ for } H_{ij} :=   \eye - X_i X_j - Y_i Y_j - Z_i Z_j  \text{ and } w_{ij} > 0 
\end{align}
where $E$ denotes the edge set of a positive edge-weighted graph $G = (V,E,w)$ on $n$ vertices, and $X_i :=  \eye^{\otimes (i-1)}_2 \otimes X \otimes  \eye^{\otimes(n-i)}_2.$
Here $X$ is one of the well-known Pauli matrices and $\eye_2$ the identity matrix of order two. The matrices $Y_i$ and $Z_i$ are similarly defined. We provide further details in \Cref{section_notation}.

The QMC problem is an example of a $k$-local Hamiltonian problem with $k = 2$, since each local Hamiltonian $H_{ij}$ acts on 2 qubits. The QMC problem is one of the simplest \textsf{QMA}-hard $k$-local Hamiltonian problems \cite[Thm.~2]{PM17}, as the general $k$-local Hamiltonian problem is \textsf{QMA}-hard if and only if $k \geq 2$ \cite{kempe2006complexity}. This motivates the search for polynomial-time (in the number of qubits $n$) approximation algorithms. 
The QMC problem is the $k$-local Hamiltonian problem that has received the most attention when it comes to developing such approximation algorithms, see \cite{GP19,PT21-lev2,Lee22,PT22,king2023improved,lee2024improved,jorquera2024algorithms,kannan2024quantum,huber2024second}. This interest is partially due to the fact that the QMC problem is equivalent to the anti-ferromagnetic Heisenberg model from physics.

The study of QMC approximation algorithms is also motivated by the similarities and differences between the QMC problem and the classical \textsf{NP}-hard Max-Cut problem. Indeed, taking as local Hamiltonian $H_{ij} = \eye - Z_i Z_j$ reduces problem \eqref{eqn_hamiltonDef} to the Max-Cut problem. Famously, the Max-Cut problem admits a polynomial-time $0.878$-approximation algorithm \cite{GW95}, which was later shown to be optimal \cite{KKEO07} under the Unique Games Conjecture (UGC) \cite{khot2002power}. In contrast, the current best-known approximation ratio for the QMC problem is far from its upper bound: the current best-known ratio (proven in this paper) equals 0.603, while the upper bound equals 0.956 (under UGC and a related conjecture, see \cite{HNPTW21}). \Cref{table_approxLBs} provides an overview of QMC approximation algorithms for different graph classes, and their approximation ratios. Note that these algorithms are designed to run on classical computers.

A natural question is whether one can design better approximation algorithms if we allow quantum algorithms. Quantum algorithms, based on the well-known Quantum Approximate Optimization Algorithm \cite{farhi2014quantum}, have recently been proposed in \cite{kannan2024quantum,marwaha2024performance}. While these algorithms are empirically shown to achieve strong approximation ratios on some unweighted QMC instances, their approximation ratios are currently unknown. Instead, several asymptotic results have been established when restricting the QMC problem to certain graph classes. For example, let $\ket{v_p}$ be the output state of the the algorithm from \cite{kannan2024quantum}, with circuit depth $p \in \mathbb{N}$. On unweighted bipartite graphs $G$, the state $\ket{v_p}$ converges to an eigenvector corresponding to the largest eigenvalue of $H_G$, as $p \to \infty$ (the convergence rate is currently not known). 

Most QMC approximation algorithms use a semidefinite programming (SDP) relaxation of the QMC problem. SDP relaxations of the QMC problem belong to the field of noncommutative polynomial optimization, and are based on the NPA hierarchy~\cite{navascues2008convergent,pironio2010convergent}, which is the noncommutative variant of the Lasserre hierarchy~\cite{lasserre2001global}.  The connection to noncommutative polynomials follows from considering the local Hamiltonians in \eqref{eqn_hamiltonDef} as degree-2 polynomials in noncommutative variables $x_i, y_i$ and $z_i$ that represent the Pauli matrices (we provide further details in \Cref{section_SDP_defs}). SDP relaxations of the QMC problem, based on the variables $x_i, y_i$ and $z_i$, were first considered in \cite{GP19}. Later research \cite{KPTTZ23,CEHKW23} investigated SDP relaxations of the QMC problem based on the noncommutative SWAP operators $S_{ij} := (\eye + X_i X_j + Y_i Y_j + Z_i Z_j)/2$.

\paragraph{Contributions.}
In this work, we improve the analysis of the QMC approximation algorithm \cite[Alg.~5]{lee2024improved} for general graphs, and propose two new QMC approximation algorithms, one for triangle-free graphs, and the other for bipartite graphs. This yields the best-known approximation ratios  for these three graph classes: \ratioGeneral{}, \ratioNoTri{} and $\ratioBipart{}$ respectively, see \Cref{table_approxLBs}. 

The improved analysis of \cite[Alg.~5]{lee2024improved} follows by showing that a particular vector induced by an SDP relaxation of the QMC problem is contained in the matching polytope. To prove the containment,  we compute a bound on the optimal QMC value for all (up to isomorphism) unweighted graphs on $\leq 13$ vertices. We derive properties of the used SDP relaxation that reduce the required computation time to approximately 16 hours.

Our QMC approximation algorithm for triangle-free graphs is inspired by~\cite[Alg.~17]{king2023improved}, which achieves an approximation ratio of 0.582. Our improvements here are due to the use of the matching-based QMC approximation algorithm from \cite{lee2024improved}, and improved parameters. One such parameter is a real-valued function $\Theta$ that is required to satisfy non-trivial constraints, which we capture using a set of functions $\mathcal{A}$. Despite these non-trivial constraints, we prove that our choice of $\Theta \in \mathcal{A}$ yields an approximation ratio that is at most 0.00009 below the ratio obtained by an optimal $\Theta \in \mathcal{A}$.

Our QMC approximation algorithm for bipartite graphs is inspired by \cite[Alg.~1]{Lee22}, which is suited for general graphs and achieves an approximation ratio of 0.562. One step of \cite[Alg.~1]{Lee22} is rounding the used SDP relaxation of the QMC to a cut of the input graph. For bipartite graphs, we instead take the cut as the bipartition. The resulting algorithm requires a real-valued function $\Theta \in \mathcal{A}$ as parameter, similar to our algorithm for triangle-free graphs. The function $\Theta$ maps an optimal solution of the used QMC SDP relaxation, to angles $\theta_e \in [0,\pi/2]$ used in the outputted state, for all edges $e$ of the bipartite input graph. To further improve the algorithm, we carefully adjust the value of $\theta_e$ when the SDP relaxation assigns a sufficiently large objective value to edge~$e$.

\paragraph{Outline.} This paper is organized as follows. We provide our notation regarding graphs and noncommutative polynomials in \Cref{section_notation}. \Cref{section_graphTheoryPrelim} provides definitions of all the concepts from graph theory that we use. In \Cref{section_SDP_defs}, we state the SDP relaxation that is used for the approximation algorithms, and consider some simple properties of it. \Cref{section_sharpenedAnalysis} provides the improved analysis of \cite[Alg.~5]{lee2024improved}. 
In \Cref{section_triangleFreeGraphs} we provide our QMC approximation algorithm for triangle-free graphs, and prove that it achieves an approximation ratio of \ratioNoTri{}. In \Cref{section_approxBipart}, we provide our QMC approximation algorithm for bipartite graphs, and prove that it achieves an approximation ratio of $\ratioBipart{}$. Concluding remarks and future research directions are given in~\Cref{section_conclusions}. Some of our results are based on computations; our code is available in an online repository available at \githubLink.

\begin{table}[ht]
    \centering
\begin{tabular}{l|l|l}
\hline
Reference & Ratio & Remark \\ \hline
\cite{GP19} & $0.498$ & \genGraphTable{}, outputs product state \\
\cite{PT22} & $1/2$ & \genGraphTable{}, outputs product state \\
\cite{huber2024second} & $0.526$ & \genGraphTable{}, uses SOC instead of SDP \\
\cite{AGM20} & $0.531$ & \genGraphTable{} \\
\citep{PT21-lev2} & $0.533$ & \genGraphTable{} \\
\cite{Lee22} & $0.562$ & \genGraphTable{} \\
\cite{lee2024improved} & $0.595$ & \genGraphTable{} \\
\cite{jorquera2024algorithms} & $0.599$ & \genGraphTable{}, improved analysis of \cite{lee2024improved} \\
Thm. \ref{thm_boundN5} & $\ratioGeneral{}$ & \genGraphTable{}, improved analysis of \cite{lee2024improved} \\
\cite{king2023improved} & $0.582$ & Triangle-free graphs \\
Thm. \ref{lemma_UB_triangleFree} & $\ratioNoTri{}$ & Triangle-free graphs \\
\cite{king2023improved} & $\frac{1}{\sqrt{2}} \approx 0.707$ & Bipartite graphs \\
Thm. \ref{thm_LBbipartAlg} & $\ratioBipart{}$ & Bipartite graphs \\ \hline
\end{tabular}
    \caption{Lower bounds on the approximation ratios of classical QMC approximation algorithms.}
    \label{table_approxLBs}
\end{table}

\subsection{Notation}
\label{section_notation}

The Pauli matrices are defined as
\begin{align*}
    X := \begin{bmatrix}
        0 & 1 \\
        1 & 0
    \end{bmatrix}, \, Y := \begin{bmatrix}
        0 & -\imagUnit \\
        \imagUnit & 0
    \end{bmatrix}, \text{ and } Z := \begin{bmatrix}
        1 & 0 \\
        0 & -1
    \end{bmatrix},
\end{align*}
where $\imagUnit$ is the imaginary unit. Recall that the matrix $X_i$ is defined as
$X_i :=  \eye^{\otimes (i-1)}_2 \otimes X \otimes  \eye^{\otimes(n-i)}_2$  where $n \in \mathbb{N}$. Matrices  $Y_i$ and $Z_i$ are similarly defined. The $3n$ matrices in $\mathbf{P} := \{ X_i, Y_i, Z_i \, : \, i \in [n] \}$, where $[n]:=\{1,\dots,n\}$, are also referred to as Pauli matrices. 

We denote by  $G = (V,E,w)$ an edge-weighted graph with vertex set $V$, edge set $E$, and positive edge weights $w$. 
If a graph is unweighted,  we write $G=(V,E)$. All graphs we consider are simple and undirected. We also denote the vertex and edge sets of a graph $G$ as $V(G)$ and $E(G)$, respectively. The neighborhood of a vertex $i$ is defined as $N(i) := \setFunct{ j \in V}{ \edge{i}{j} \in E}$, and its degree as $\degr{i} := |N(i)|$. The set of all unweighted graphs on $n$ vertices is denoted $\mathcal{G}_n$.

A Hermitian matrix is a matrix that is equal to its conjugate transpose. The largest eigenvalue of a Hermitian matrix $A$ is denoted $\lambdaMax{A}$. We write $A \succeq 0$ to indicate that $A$ is positive semidefinite (PSD). For any $x \in \mathbb{R}$, $x^+ := \max\{x,0\}$ and $\lfloor x \rfloor$ is the floor of $x$. The expectation operator is denoted $\mathbb{E}(\cdot)$. The function $\argFunc{\cdot} : \mathbb{C} \to [0, 2\pi)$  returns the angle of a complex number, with $\argFunc{0} = 0$. The function $\signFunction{\cdot}: \mathbb{R} \to \{ \pm 1 \}$ returns the sign of the input, with $\signFunction{0} = 1$. The Euclidean norm of a vector  $x \in \mathbb{R}^k$ is denoted by  $\| x \|$.

In this work we use noncommutative polynomials in the $3n$ variables $x_1,\ldots,x_n,$ $ y_1,\ldots,y_n,$ $z_1,\ldots,z_n$. We denote the set of variables by $\mathbf p$, that is, $\mathbf{p} := \setFunct{ x_i, y_i, z_i}{i \in [n]}$. A product of noncommutative variables is referred to as a word, or monomial. The length, or degree, of a word is defined as the sum of the exponents in the word. We define $\langle \mathbf{p}^n \rangle_k$ as the set of all words of degree at most~$k$, consisting of the $3n$ variables from $\mathbf{p}$. The degree of a polynomial $p$, denoted $\deg(p)$, equals the largest degree among its words with nonzero coefficients. 
We denote the ring of noncommutative polynomials with complex coefficients by $\C \langle \mathbf{p} \rangle$, and we let $\C\langle \mathbf{p} \rangle_k$ be the set of all such polynomials of degree at most $k$, where $k \in \mathbb{N}$. The sets $\R \langle \mathbf{p} \rangle$ and $\R\langle \mathbf{p} \rangle_k$ are defined analogously, using real coefficients instead of complex. The involution $w \mapsto w^*$ returns the word $w$ with the order of the symbols reversed. This involution is extended to $\C \langle \mathbf{p} \rangle$ and $\R \langle \mathbf{p} \rangle$ by conjugate linearity. 

The dual spaces of $\C\langle \mathbf{p} \rangle, \, \C\langle \mathbf{p} \rangle_k, \, \R\langle \mathbf{p} \rangle$ and $\R\langle \mathbf{p} \rangle_k$ are denoted by $\C\langle \mathbf{p} \rangle^*, \, \C\langle \mathbf{p} \rangle^*_k, \, \R\langle \mathbf{p} \rangle^*$ and $\R\langle \mathbf{p} \rangle^*_k$. Let $f_1, \dots, f_m \in \C \langle \mathbf p \rangle$. 
We define $\langle f_1, \dots, f_m \rangle := \setFunct{ \sum_{i=1}^m g_i f_i}{ g_i \in \C \langle \mathbf p \rangle \, \forall i \in [m]}$ as the ideal generated by the polynomials $f_1, \dots, f_m$. We define the ideal truncated at degree $k$ as $\langle f_1, \dots, f_m \rangle_k := \setFunct{ \sum_{i=1}^m g_i f_i}{ g_i \in \C\langle \mathbf p \rangle, \deg(g_i f_i) \leq k  \, \forall i \in [m]}$. For $f \in \C\langle \mathbf p\rangle$, we write $f(\mathbf{P}) \in \mathbb{C}^{2^n \times 2^n}$ to denote the polynomial $p$ evaluated at the Pauli matrices.

\subsection{Graph theory preliminaries}
\label{section_graphTheoryPrelim}
We  present some concepts from graph theory that will be used in the following sections.
\begin{definition}[Graph matching]
    \label{def_graphMatching}
    A matching $\mathcal{M}$ of $G = (V,E,w)$ is a mapping $\mathcal{M}: E \to \{0,1\}$, satisfying $\mathcal{M}_e \mathcal{M}_{e'} = 0$ if edges $e$ and $e'$ are adjacent. We say that the edges satisfying $\mathcal{M}_e = 1$ form a matching. The weight of a matching is given by $\sum_{e \in E} w_e \mathcal{M}_e$. We say that a vertex $i \in V$ is unmatched by $\mathcal{M}$ if $\mathcal{M}_e = 0$ for all edges $e$ incident to $i$.
\end{definition}
Given a weighted graph $G= (V, E, w)$, Edmonds \cite{Ed65} showed that the following linear program (LP), gives the weight of a maximum weight matching in $G$:
\begin{align}
& \max &&  \sum_{e \in E} w_e x_e \tag{MW-LP} \label{eq:matching:lp} \\
& \mathrm{s.t.} &&  \sum_{j \in N(i)} x_{ij} \leq 1 \quad \forall i \in V, \quad x_e \geq 0 \quad \forall e \in E \label{eq:matching:vertex} \hspace{5em}  \\ 
& && \sum_{e \in \edgeSubset{S}} x_e \leq \frac{|S|-1}{2} \quad \forall S \subseteq V, \, |S| \text{ odd}  \label{eq:matching:set}
\end{align}
where $\edgeSubset{S} := \setFunct{ \edge{i}{j} \in E }{ i,j \in S }$ for all $S \subseteq V$. 

The feasible set of \ref{eq:matching:lp} is referred to as the \textit{matching polytope}. It is common to refer to the constraints \eqref{eq:matching:set} as \textit{odd set inequalities}. \ref{eq:matching:lp} is an LP with an exponential number of constraints, and thus cannot be efficiently solved.  
 However, there exist polynomial time algorithms for finding a maximum matching in a general graph, see e.g., \cite{Edmonds1965PathsTA,Blum1990ANA}.

\begin{definition}[Edge-induced subgraph]
    \label{def_edgeInduceSub}
    Let $G = (V,E)$, and $\widetilde{E} 
    \subseteq E$. The edge-induced subgraph of $G$ is denoted by $G[\widetilde{E}]$ and defined as $G[\widetilde{E}] := (\widetilde{V}, \widetilde{E})$,  for
        $\widetilde{V} = \setFunct{i 
        \in V}{ \exists j \in V \text{ such that } \edge{i}{j} \in \widetilde{E} }$.
\end{definition}

\begin{definition}[Vertex cover number $\tau(G)$]
    \label{def_graphCover}
    For $G = (V,E)$, we say that $S \subseteq V$ is a vertex cover  of $G$ if every edge of $E$ is adjacent to at least one vertex in $S$. The vertex cover number of $G$, denoted $\tau(G)$, is defined as the  cardinality of the smallest  vertex cover of $G$. Equivalently,
    \begin{align*}
        \tau(G) = \min_{S \subseteq V(G)} \setFunct{ |S| }{ \edge{i}{j} \cap S \neq \emptyset \, \, \, \forall \edge{i}{j} \in E}.
    \end{align*}
\end{definition}

\section{SDP bounds}
\label{section_SDP_defs}
The matrix $H_G$, see \eqref{eqn_hamiltonDef}, is of order $2^n$ which prohibits computing $\lambdaMax{H_G}$ directly. We present two SDP hierarchies, based on noncommutative polynomial optimization \cite{navascues2008convergent,pironio2010convergent}, see also~\cite{GP19}, that yield upper bounds on $\lambdaMax{H_G}$ in time polynomial in $n$.

The problem of computing $\lambdaMax{H_G}$
can be considered as a noncommutative polynomial optimization problem. 
Indeed, for each of the Pauli matrices $X_i, Y_i, Z_i$ $i\in[n]$ from  $\mathbf{P}$ we introduce a noncommuting variable $x_i, y_i,z_i$, respectively. We then consider the ring of noncommutative polynomials $\C\langle \mathbf{p} \rangle$. 
The Pauli (anti)commutation relations
can be expressed as $p(\mathbf{P}) = 0$ for all $p \in \mathcal{I}$, 
where   $\mathcal{I}$ is the ideal defined as: 
\begin{equation*}
\begin{aligned}
    \mathcal{I} = \langle &x_i^2 - 1, y_i^2-1, z_i^2-1, x_i y_j - y_j x_i, x_i z_j - z_j x_i, y_i z_j - z_j y_i, \\
    &x_i y_i- \imagUnit z_i, x_i z_i + \imagUnit y_i, y_i z_i - \imagUnit x_i,  y_i x_i + \imagUnit z_i, z_i x_i - \imagUnit y_i, z_i y_i + \imagUnit x_i, : i,j \in V, i \neq j \rangle.
\end{aligned}    
\end{equation*}
Let us also write $H_G(\mathbf{p})$ as the polynomial obtained after replacing the Pauli matrices in \eqref{eqn_hamiltonDef} with their corresponding variables $x_i$, $y_i$ and $z_i$. We have the following upper bound on $\lambdaMax{H_G}$, parametrized by some $k \in \mathbb{N}$: 
\begin{equation}
\label{eqn_firstSDPattempt}
\begin{aligned}
    \lambdaMax{H_G}  \leq  \max_{L \in \C\langle \mathbf{p} \rangle_{2k}^*}  & L(H_G(\mathbf{p})) \\
     \mathrm{s.t.}  &\,  L(p) = 0 \quad  \forall p \in \mathcal{I}_{2k}, \,  L(p^*p) \geq 0  \ \forall p \in \C\langle \mathbf{p} \rangle_{k}, \, L(1) = 1.
\end{aligned}
\end{equation}
Here, $L : \C\langle  \mathbf{p} \rangle_{2k} \to \C$ is a linear functional. The condition $L(p^*p) \geq 0$ for all degree-$k$ polynomials is equivalent to the \textit{moment matrix} $M(L)$ being PSD. The moment matrix $M(L)$ is indexed by the elements of 
\begin{align}
    \label{eqn_matsInP}
    \langle \mathbf{p}^n \rangle^{\mathcal{I}}_k := \setFunct{ u }{u = \prod_{i = 1}^n u_i, \, u_i \in \{1,x_i,y_i,z_i \} \, \, \forall i \in [n], \, \deg(u) \leq k},
\end{align}
and  $M(L)_{u,v} = L(u^*v) = L(uv)$. Observe that $\langle \mathbf{p}^n \rangle^{\mathcal{I}}_k \equiv \langle \mathbf{p}^n \rangle_k \mod \mathcal{I}$. Program \eqref{eqn_firstSDPattempt} is the $k$th level of the NPA hierarchy~\cite{navascues2008convergent,pironio2010convergent} for the QMC problem on the graph $G$.

To see that \eqref{eqn_firstSDPattempt} defines an upper bound on $\lambdaMax{H_G}$, consider
\begin{align}
    \label{eqn_finiteRealization}
    \widetilde{L}(p) := \Tr{ \left( p(\mathbf P) +  p^*(\mathbf P) \right) 
  \rho } /2,
\end{align}
where $\rho = \ket{\psi}\bra{\psi}$, for $\psi$ a unit length eigenvector of $H_G$ corresponding to $\lambdaMax{H_G}$. It can be shown that $\widetilde{L}$ is feasible for \eqref{eqn_firstSDPattempt} and achieves an objective value of $\lambdaMax{H_G}$ \cite[App.~A.1]{Lee22}. When a functional $L$ can be written in the form \eqref{eqn_finiteRealization}, for some unit-trace matrix $\rho$, we say that $L$ corresponds to an evaluation at some finite-dimensional realization of the Pauli algebra.

Note that $\widetilde{L}(p)$ is real for all $p \in \C \langle \mathbf{p} \rangle_{2k}$, because $p(\mathbf P) +  p^*(\mathbf P)$ is a Hermitian matrix for all $p \in \C \langle \mathbf{p} \rangle$. Moreover, $\widetilde{L}(p) = 0$ if $p+p^* \in \mathcal{I}_{2k}$. Thus, if we restrict in \eqref{eqn_firstSDPattempt} $L \in \R \langle \mathbf{p} \rangle_{2k}^*$, and  $L(p) = 0$ whenever $p+p^* \in \mathcal{I}_{2k}$, $\widetilde{L}$ remains  feasible. Therefore, the following real SDP, introduced in \cite[Sect.~4]{GP19}, also provides an upper bound on $\lambdaMax{H_G}$:
\begin{equation}
    \label{eqn_sdpRelax}
    \tag{SDP$^k$}
    \begin{aligned}
        & \max_{L \in \R\langle \mathbf{p}\rangle^*_{2k}} && \sum_{\edge{i}{j} \in E(G)} w_{ij} \left(1- L(x_i x_j) - L(y_i y_j) - L(z_i z_j)\right) \\
        & \quad \, \, \mathrm{s.t.} &&L(p) = 0 \quad \forall p \in \R\langle \mathbf p \rangle_{2k} \text{ satisfying }  p + p^* \in \mathcal{I}_{2k} \\
        &   &&  M(L) \succeq 0, \, \, L(1) = 1.
    \end{aligned}
    \noeqref{eqn_sdpRelax}
\end{equation}
Observe that the polynomial $p + p^*$ has only real coefficients, and that the positivity condition $L(p^*p) \geq 0$ has been replaced by the (equivalent) PSD constraint $M(L) \succeq 0$. We denote the feasible set of \ref{eqn_sdpRelax} by 
\begin{align}
    \label{eqn_feasLass}
    \feasLass{k}{n} :=  \setFunct{L \in \R\langle \mathbf{p} \rangle^*_{2k}}{L(p) = 0  \text{ for all }  p \in \R\langle \mathbf p \rangle_{2k} \text{ satisfying }  p + p^* \in \mathcal{I}_{2k}, M(L) \succeq 0}. 
\end{align}
The order of the PSD variable $M(L)$ in \ref{eqn_sdpRelax} is given by $\left| \langle \mathbf{p}^n \rangle^{\mathcal{I}}_k \right|= \sum_{i=0}^k 3^i \binom{n}{i} \in \mathcal{O}(n^k)$.
Hence, larger values of $k$ induce larger SDPs that are harder to solve, but offer tighter bounds. When $k=n$, we have the following exactness result. For a similar statement about an SDP hierarchy based on the SWAP operators, see \cite[Thm.~4.10]{CEHKW23}.
\begin{lemma}  \label{lemma_exactnessResult}
    Let $G$ be a graph on $n$ vertices. The optimal value of $\textup{SDP}^n$ equals $\lambdaMax{H_G}$.
\end{lemma}
\begin{proof}
    The monomials in $\langle \mathbf{p}^n \rangle^\mathcal{I}_n$ form a basis of the quotient ring $\mathbb{R}\langle \mathbf{p} \rangle / \mathcal{I}$. In particular, for $k \geq n$ we have $\langle \mathbf{p}^n \rangle^\mathcal{I}_n = \langle \mathbf{p}^n \rangle^\mathcal{I}_k$ and therefore the optimal value of $\mathrm{SDP}^n$ equals that of \ref{eqn_sdpRelax}. It remains to argue that the optimal value of $\mathrm{SDP}^n$ equals $\lambdaMax{H_G}$. To that end, consider an optimal solution $L$ to $\mathrm{SDP}^{n+2}$. The associated moment matrix satisfies $\mathrm{rank}(M_{n+2}(L)) = \mathrm{rank}(M_{n}(L))$ since $\langle \mathbf{p}^n \rangle^\mathcal{I}_n = \langle \mathbf{p}^n \rangle^\mathcal{I}_{n+2}$. Therefore, by \cite[Thm.~2]{pironio2010convergent}, $L$ corresponds to an evaluation at some finite-dimensional realization of the algebra described by the relations in $\mathcal I$. 
    The finite-dimensional irreducible representations of this algebra are $*$-isomorphic to the algebra generated by the Pauli strings (cf.~\cite[Thm.~2.3]{chao2017overlapping}). Hence, for some $K \leq \mathrm{rank}(M_n(L))$, we can write $L(p) = \sum_{i=1}^{K} w_i \bra{\psi_i} \left( p(\mathbf P) +  p^*(\mathbf P) \right) \ket{\psi_i}/2$ for some states $\{\ket{\psi_i}\}_{i \in [K]}$ and weights $w_1, \ldots, 
    w_K \geq 0$ with $\sum_{i=1}^K w_i=1$. This shows that $L(H_G(\mathbf p)) \leq \lambda_{\max}(H_G)$, which concludes the proof. 
\end{proof}

\section{Improved analysis of a QMC approximation algorithm}
\label{section_sharpenedAnalysis}

In this section, we consider the polynomial-time QMC approximation algorithm from \cite{lee2024improved}, and provide a sharper analysis of this algorithm.

 To prove the 0.595 approximation ratio of \cite[Alg.~5]{lee2024improved}, Lee and Parekh showed how to convert a solution of \ref{eqn_sdpRelax} for relaxation level $k=2$, to a vector $h^+ = (h^+_e)_{e \in E}$ 
 (to be defined later) that has the property that $(\frac45 h^+_e)_{e \in E}$ is feasible for \ref{eq:matching:lp}. They did so by showing that $h^+$ satisfies the odd set inequalities \eqref{eq:matching:set} for $|S| \leq 3$. Recently, the authors of \cite{jorquera2024algorithms} showed that $h^+$ satisfies the odd set inequalities for $|S| \leq 5$. 
 Here, we show that $h^+$, for $k \geq 2$, satisfies the odd set inequalities \eqref{eq:matching:set} for all $S \subseteq V$ such that $|S| \leq 9$. 
 This implies that $(\frac{10}{11} h^+_e)_{e \in E}$ is feasible for \ref{eq:matching:lp} (see \Cref{lemma_generalShifting}). The improved prefactor $10/11$ (instead of $4/5$) directly translates to an improved approximation ratio of 0.602 for the algorithm. In case $k \geq 13$, we can even show that $(\frac{14}{15} h^+_e)_{e \in E}$ is feasible for \ref{eq:matching:lp}, which improves the approximation ratio further to \ratioGeneral{}. Our proof of this relies on the exactness of \ref{eqn_sdpRelax}, with $k \geq 13$, for graphs that have at most 13 vertices (see \Cref{lemma_exactnessResult}).  Let us first present \cite[Alg.~5]{lee2024improved}.
\begin{algorithm}[QMC approximation algorithm~\cite{lee2024improved}]
\label{alg:approxLee}
\hphantom{.}

Input: weighted graph  $G = (V,E,w)$ with $n$ vertices, $k \in \mathbb{N}$.
\begin{enumerate}
    \item \label{step_product_state} Find a product state as follows:
    \begin{enumerate}
        \item Solve \ref{eqn_sdpRelax} to obtain an optimal moment matrix $M(L)$. From this $M(L)$ obtain  unit length vectors $\mathbf{v}(u) \in \mathbb{R}^{\left| \langle \mathbf{p}^n \rangle^{\mathcal{I}}_k \right|}$, $u \in  \langle\mathbf{p}^n \rangle^{\mathcal{I}}_k$, satisfying $M(L)_{u,u'}  = \mathbf{v}(u)^{\top} \mathbf{v}(u')$ for any $u, u' \in \langle \mathbf{p}^n \rangle^{\mathcal{I}}_k$. Set
    \begin{align}
    \label{eqn_gramVectors}
    \mathbf{v}_i : = (\mathbf{v}(x_i)^\top, \mathbf{v}(y_i)^\top, \mathbf{v}(z_i)^\top)^\top /\sqrt 3, \quad i \in V.
\end{align}
        \item Sample a random matrix $R \in \mathbb{R}^{3\times 3 \left| \langle \mathbf{p}^n \rangle^{\mathcal{I}}_k \right|}$, 
        whose elements are independently drawn from the standard normal distribution.
        \item \label{item_uComputation} Compute $u_i:= R \mathbf{v}_i/ \| R \mathbf{v}_i \| \in \mathbb{R}^3$ for all $i \in V$.
               
        \item Set $\rho_1:= \prod_{i\in V}\frac{1}{2}( \eye + u_{i,1} X_i +u_{i,2} Y_i +u_{i,3} Z_i ) $. 

    \end{enumerate}
    \item \label{step_matching} Find a matching state as follows:
    \begin{enumerate}
        \item Find a maximum weight matching $\mathcal{M}$ of $G$, see \Cref{def_graphMatching}.  
        \item \label{step_rhoTwoDef} Set  $\rho_2 := 2^{-n} \prod_{\{i,j\}:  \mathcal{M}_{\{i,j\}} =1} (\eye-X_i X_j - Y_i Y_j - Z_i Z_j)$.
        \end{enumerate}
    \item Output $\rho_{i^*}$, for ${i^*} := \arg \max_{i \in \{1,2\}} \Tr{ \rho_i H_G}$.
\end{enumerate}
\end{algorithm} 

Note that \Cref{alg:approxLee} is a polynomial-time algorithm for any fixed $k \in \mathbb{N}$. In particular, solving \ref{eqn_sdpRelax} up to fixed precision, and computing a maximum weight matching \cite{Ed65} are both possible in time polynomial in $n = |V(G)|$.

For $k = 2$, it is shown \cite[Thm.~11]{lee2024improved} that a lower bound on the approximation ratio of \Cref{alg:approxLee} is given by the function value $\alpha(4/5)$ $(\geq 0.595)$, for $\alpha$ given by
\begin{align}
    \label{eqn_approxRatioCompute}
     \alpha(\mu) &:=  \max_{p \in [0,1]} \min_{x \in (-1, 1]}  \frac{ p q(x) + (1-p)\left(   1+3\mu x^+ \right)}{2+2x}, \, \mu \in [0,1],
\end{align}
where
\begin{align}
    \label{eqn_hyperGeoQ}
    q(x) := 1 +  \left( \frac{8+16x}{9 \pi}\right) \, \,   \hyperGeom\!\left(1/2,1/2,5/2,\left( \frac{1+2x}{3}\right)^2 \right),
\end{align}
for $\hyperGeom$ the hypergeometric function
\begin{align*}
\hyperGeom(a,b,c,z) := \sum\limits_{n=0}^\infty \frac{(a)_n (b)_n}{(c)_n}\frac{z^n}{n!}, \quad (t)_n := \frac{\Gamma(t+n)}{\Gamma(t)}=t(t+1)\cdots (t+n-1).
\end{align*}
Note that $q(x)$ is well-defined for all $x \in [-1,1]$,  in particular for $x=1$ due to Gauss's hypergeometric theorem, see e.g., \cite[Sect.~1.3]{bailey1935generalized}. 

\begin{remark}
    \label{remark_divByZero}
     The minimization over $x \in (-1,1]$ in  \eqref{eqn_approxRatioCompute} is well-defined, since the minimum does not occur near $x = -1$. Indeed, one can verify that $q(-1) \approx 0.71$, so that 
        $\lim_{x \downarrow -1} \frac{ p q(x) +(1-p) \left(  1+3\mu x^+ \right)}{2+2x} = + \infty$,
    for any $\mu, p \in [0,1]$.
\end{remark}
The variable $x$ in \eqref{eqn_approxRatioCompute} corresponds to $h_{ij}$, which is defined as follows:
\begin{definition}
    \label{def_sdpValues}
    Let $L \in \feasLass{k}{n}$, see \eqref{eqn_feasLass}, with $n \geq 2$ and $k \geq 1$. For $i,j \in [n]$, $i \neq j$, define the values
    \begin{align*}
        g_{ij} :=  1 - L(x_i x_j) - L(y_i y_j) - L(z_i z_j), \, \, \,  h_{ij} := g_{ij}/2 - 1, \, \, \, h_{ij}^+ := \max\{ h_{ij}, 0 \}.
    \end{align*}
    If $k \geq 2$, then $g_{ij} \in [0,4]$, $h_{ij} \in [-1,1]$ and $h^+_{ij} \in [0,1]$. Note that $g_{ij}$, $h_{ij}$ and $h^+_{ij}$ are functions of $L \in \feasLass{k}{n}$. Throughout the rest of this paper, we implicitly assume this dependence.
\end{definition}
\Cref{def_sdpValues} is similar to \cite[Def.~3]{lee2024improved} with the following difference:  $h:=(h_{ij})$ here differs by a factor of 2 compared to the definition of $h$ in \cite{lee2024improved}. We note that the lower and upper bounds on the variables  are tight  for all $k \geq 2$. The SDP relaxation \ref{eqn_sdpRelax} can be expressed in terms of the above values:
\begin{align}
    \label{eqn_sdpUBtoHMax}
    \max_{L \in \feasLass{k}{n}} \sum_{e \in E(G)}  w_e g_e = \max_{L \in \feasLass{k}{n}} \sum_{e \in E(G)}  w_e (2+2h_e).
\end{align}
The approximation ratio of \Cref{alg:approxLee} is related to the values $h^+$, see \cite[Thm.~11]{lee2024improved}.
\begin{theorem}[\cite{lee2024improved}]
    \label{lemma_approxRatioLee}
    Let $k \in \mathbb{N}$ and let $G$ be the input graph on $n$ vertices to \Cref{alg:approxLee}. If $(\mu h^+_e)_{e \in E(G)}$, for some $\mu \in [0,1]$, is contained in the matching polytope for all $L \in \feasLass{k}{n}$, then the approximation ratio of \Cref{alg:approxLee} with input $k$ is at least $\alpha(\mu)$, where the function $\alpha$ is defined in~\eqref{eqn_approxRatioCompute}.
\end{theorem}
\Cref{table_alphaMuVals} presents lower bounds on $\alpha(\mu)$, for different values of $\mu$. It can be observed that $\alpha(\mu)$ is increasing in $\mu$, and we are therefore interested in proving that $(\mu h^+_e)_{e \in E}$ is contained in the matching polytope for $\mu$ as close to $1$ as possible. To do so, we require the following well-known property of the matching polytope, see e.g., \cite[Ex.~6]{braun2015matching} or \cite[App.~A]{sinha2018lower}.

\begin{table}
    \centering
    \begin{tabular}{l|lllllll}
    \hline
        $\mu$ & $4/5$ & $6/7$ & $8/9$ & $10/11$ & $12/13$ & $14/15$ & $1$ \\ \hline 
        $\alpha(\mu) \geq$ & $0.595$ & $0.599$ & $0.601$ & $0.602$ & $0.602$ & $\ratioGeneral{}$ & $0.606$ \\
        $p^*$ & $0.672$ & $0.697$ & $0.709$ & $0.716$ & $0.721$ & $0.724$ & $0.744$ \\ 
        $x^*$ & $0.152$ & $0.153$ & $0.146$ & $0.139$ & $0.142$ & $0.131$ & $0.115$ \\ \hline
    \end{tabular}
    \caption{Lower bounds on $\alpha(\mu)$ 
obtained by rounding down $\alpha(\mu)$ to 3 digits. The rows $p^*$ and $x^*$ present a corresponding approximate maximizer/minimizer for the optimization problem in the definition of $\alpha(\mu)$.}
\label{table_alphaMuVals}
\end{table}

\begin{lemma}
\label{lemma_generalShifting}
    Let $x\in \mathbb{R}^{|E|}$ be a vector that satisfies the constraints \eqref{eq:matching:vertex}. If 
     $x$ also satisfies the odd set inequalities \eqref{eq:matching:set} for all $S \subseteq V$ with $|S| \leq s$, $s$ odd, then $\frac{s+1}{s+2}x$ is contained in the matching polytope.
\end{lemma}
We will use \Cref{lemma_generalShifting} for $x = (h^+_e)_{e \in E}$, see \Cref{def_sdpValues}. Therefore, we need to verify whether the odd set inequalities $\sum_{e \in E_S} h_e^+ \leq \lfloor s /2 \rfloor$, with $s := |S|$ and $E_S$ defined below \ref{eq:matching:lp}, hold. To do so, we compute an upper bound on $\sum_{e \in E_S} h_e^+$ for all feasible solutions to \ref{eqn_sdpRelax}. That is, we compute
\begin{align}
    \label{eqn_nonLinearObj}
    \max_{L \in \feasLass{k}{s}} \sum_{ 1 \leq i <j \leq s} h^+_{ij},
\end{align}
for odd $s > 1$ and $\feasLass{k}{s}$ the feasible set of 
\ref{eqn_sdpRelax}, see \eqref{eqn_feasLass}. Observe that \eqref{eqn_nonLinearObj} equals 
\begin{align*}
    \max_{z \in \{0,1\}^{s(s-1)/2}, \, L \in \feasLass{k}{s}} \sum_{ 1 \leq i <j \leq s} z_{ij} h_{ij}.
\end{align*}
Any fixed $z \in \{0,1\}^{s(s-1)/2}$ defines a graph $G \in \mathcal{G}_s$ with edge set $\setFunct{ \{i,j\}}{z_{ij} =1}$, so that \\ $\max_{L \in \feasLass{k}{s}} \sum_{ 1 \leq i <j \leq s} z_{ij} h_{ij} = c(G,k)$, where 
\begin{equation}
\label{eqn_cFunctionDef}
     c(G,k) := \max_{L \in \feasLass{k}{s}} \sum_{e \in E(G) } h_e. 
 \end{equation}
 Note that $c(G,k)$ can be computed in polynomial time, up to finite precision, by solving the corresponding SDP. 
   The connection between $c(G,k)$ and $\sum_{e \in E} h^+_e$ is clarified by the following result. 
 \begin{lemma}
    \label{lemma_hUBinC}
     Let $s \geq 2$ and $k \geq 1$. For any $G = (V,E) \in \mathcal{G}_s$ and $h^+$ as in \Cref{def_sdpValues}, we have
     \begin{align*}
        \max_{L \in \feasLass{k}{s}} \sum_{e \in E(G)} h^+_{e} \leq  \max_{L \in \feasLass{k}{s}} \sum_{ 1 \leq i <j \leq s} h^+_{ij} = \max_{G \in \mathcal{G}_s} c(G,k).
 \end{align*}
 \end{lemma}
 \begin{proof}
     Since $h^+ \geq 0$, 
     we have $\max_{L \in \feasLass{k}{s}} \sum_{e \in E} h^+_{e} \leq \max_{L \in \feasLass{k}{s}} \sum_{1 \leq i < j \leq s} h^+_{ij}$.
     Let $\widetilde{L} \in \feasLass{k}{s}$, see \eqref{eqn_feasLass}, be the linear functional that maximizes this upper bound. Let $\widetilde{h}$ be the $h$ values corresponding to $\widetilde{L}$, as in \Cref{def_sdpValues}. Consider the graph $\widetilde{G} = (\widetilde{V},\widetilde{E})$, with $\widetilde{V} = [s]$ and $\widetilde{E} := \{ \{i,j\} :  \widetilde{h}_{ij} \geq 0 \}$. Then 
     \begin{align}
        \label{eqn_inequalityResultLemma}
         \max_{L \in \feasLass{k}{s}} \sum_{\edge{i}{j} \in E} h^+_{ij} \leq \max_{L \in \feasLass{k}{s}} \sum_{1 \leq i < j \leq s} h^+_{ij} = c\mleft( \widetilde{G},k \mright) \leq \max_{G \in \mathcal{G}_s } c(G,k).
     \end{align}
     It remains to show that the second inequality of \eqref{eqn_inequalityResultLemma} is an equality. To do so, let $G^\prime \in \arg \max_{G \in \mathcal{G}_s} c(G,k)$ with edge set $E^\prime$, and let $h^\prime$  be the $h$ values corresponding to the SDP defining $c(G^\prime,k)$. By optimality of $G^\prime$, $h^\prime_{ij} \geq 0$ for all $\edge{i}{j} \in E^\prime$, and $h^\prime_{ij} \leq 0$ for those $\{i,j\} \not\in E^\prime$. Hence,
     \begin{align*}
         \max_{G \in \mathcal{G}_s} c(G,k) = c(G^\prime,k) = \sum_{e \in E^\prime} h_e^\prime = \sum_{1 \leq i < j \leq s} \max\{ h_{ij}^\prime,0\} \leq \max_{L \in \feasLass{k}{s}} \sum_{1 \leq i < j \leq s} h^+_{ij}. \\[-45pt]
     \end{align*}
 \end{proof}

By combining \Cref{lemma_generalShifting,lemma_hUBinC}, we obtain the following corollary.
\begin{corollary}
    \label{corr_simpleCons}
    Let $k \in \mathbb{N}$ and let $s$ be the largest odd integer for which $\max_{G \in \mathcal{G}_s} c(G,k) \leq \lfloor s / 2 \rfloor$. Then, for any $G \in \mathcal{G}_n$ and $L \in \feasLass{k}{n}$, the vector $\left( \frac{s+1}{s+2}h^+_e \right)_{e \in E(G)}$ is contained in the matching polytope of $G$. 
\end{corollary}
\begin{proof}
Let $G$ be a graph. We consider the odd set inequalities \eqref{eq:matching:set} for $x = (h_e^+)_{e \in E(G)}$. Take a set $S \subseteq V$ with $|S| \leq s$ and~$|S|$ odd. We have, for $E_S$ as defined below \ref{eq:matching:lp},
\begin{align*}
    \sum_{e \in \edgeSubset{S}} h^+_e \leq \sum_{i,j \in S} h^+_{ij} \leq \max_{G^\prime \in \mathcal{G}_{|S|}} c(G^\prime,k) \leq  \frac{|S|-1}{2},
\end{align*}
by \Cref{lemma_hUBinC} and the assumption of the corollary. Hence, $(h_e^+)_{e \in E(G)}$ satisfies the odd set inequalities $\forall S \subseteq V$ with $|S|$ odd and $|S| \leq s$. 
The claim then follows from \Cref{lemma_generalShifting}.
\end{proof}
    
 Considering \Cref{corr_simpleCons}, we aim to compute $\max_{G \in \mathcal{G}_s} c(G,k)$ for $s$ as large as possible. The case $s = 3$ is computed in \cite[Cor.~10]{lee2024improved}, where it was stated without the use of $c(G,k)$.
\begin{lemma}[\cite{lee2024improved}]
\label{lemma_s3Case}
For $k = 2$ and $s=3$, $\displaystyle\max_{G \in \mathcal{G}_s} c(G,k) = \displaystyle\max_{L \in \feasLass{k}{s}} \displaystyle\sum_{ 1 \leq i <j \leq s} h^+_{ij} = \left\lfloor \frac{s}{2} \right\rfloor$.
\end{lemma}
The recent paper \cite{jorquera2024algorithms} proves that \Cref{lemma_s3Case} also holds for $s = 5$.
 Here, we compute $\max_{G \in \mathcal{G}_s} c(G,k)$ for $3\leq s \leq 13$ and $k \geq 2$ by exploiting properties of the function $c(G,k)$ that allow us to reduce the number of computations needed. To prove our results on $c(G,k)$, we require an important property of the $h^+$ and $h$ values derived from \textit{monogamy of entanglement}. Monogamy of entanglement is a physical property of quantum systems, restricting the maximum entanglement between different quantum states. The connection between monogamy of entanglement and \ref{eqn_sdpRelax}, $k \geq 2$, is due to \cite[Thm.~11]{PT21-lev2} and \cite[Lem.~4]{lee2024improved}, which state that 
\begin{align}
    \label{eqn_monogamyStar}
    \sum_{j \in S} h_{ij} \leq \sum_{j \in S} h^+_{ij}  \leq 1,
\end{align}
for any fixed $i$ and subset $S$ of the vertices. It is known that \eqref{eqn_monogamyStar} is not implied by \ref{eqn_sdpRelax} when $k = 1.5$, see \cite[Thm.~10]{PT21-lev2}. To compare \cite[Thm.~11]{PT21-lev2} with \eqref{eqn_monogamyStar}, observe that $x_e$ from \cite{PT21-lev2} satisfies $x_e = (1 + 2 h_e)/3$.
\begin{lemma}
    \label{lemma_vertexCover}
     For any $G \in \mathcal{G}_s$, $s \geq 2$, we have the following bounds on $c(G,k)$:
     \begin{enumerate}
        \item \label{item_subgraphDecomp} $c(G,k) \leq \sum_{i = 1}^p c\left(G[E^i],k \right)$, for $\left\{E^1,\dots, E^p\right\}$ a partition of $E$, see \Cref{def_edgeInduceSub}.
        \item \label{item_tauBound} $c(G,k) \leq \tau(G)$ for $\tau(G)$ as in \Cref{def_graphCover} and $k \geq 2$.
        \item  \label{item_vertexBound} $c(G,k) \leq s / 2$, for $k \geq 2$.
         \item \label{item_lasserImp} $c(G,k) \leq c(G,k-1)$ for $k \geq 2$.
        \item \label{item_vertex1B} $c(G,k) \leq 1+\max_{G \in \mathcal{G}_{s-2}} c(G,k)$ if $G$ contains a vertex of degree $1$, $s \geq 4$ and $k \geq 2$.
        \item \label{item_exactEVB} $c(G,k) \geq \lambda_{\mathrm{max}}(H_{G})/2 - |E(G)|$. This inequality is tight if $k \geq s$.
     \end{enumerate}
 \end{lemma}
 \begin{proof} \hfill
 \begin{enumerate}
    \item For $\left\{ E^1, \dots, E^p \right\}$ a partition of $E$, we have
        $c(G,k) \leq \sum_{i=1}^p \left( \max_{L \in \feasLass{k}{s}} \sum_{e \in E^i} h_e  \right) $ \\ $= \sum_{i=1}^p c(G[E^i],k).$
    The inequality follows from the fact that the maximum of a sum is at most the sum of the maxima, while the equality follows from the definition of $c(G,k)$. 
     \item Let $V^* \subseteq V$ be a vertex cover of $G = (V,E)$ satisfying $|V^*| = \tau(G)$. For each $i \in V^*$, let $E^i \subseteq E$ be a subset of the edges adjacent to $i$, chosen in such a way that the subsets $E^i$ form a partition of $E$. Observe that each edge-induced subgraph $G[E^i]$ is a star graph. Then, by Item \ref{item_subgraphDecomp} and \eqref{eqn_monogamyStar}, we have
         $c(G,k) \leq \sum_{i \in V^*} c(G[E^i],k) \leq |V^*| = \tau(G)$.
    \item The statement follows from \eqref{eqn_monogamyStar}, since $\max_{L \in \feasLass{k}{s}} h^+_{ij} = \max_{L \in \feasLass{k}{s}}  \frac{1}{2} \sum_{i \in V} \, \, \sum_{j \in N(i)} h^+_{i j} \leq \frac{s}{2}$. The factor $1/2$ accounts for the fact that we count each edge twice.
     \item   Let $L \in \feasLass{k}{s}$. Let $\widetilde{L} \in \R\langle \mathbf{p}\rangle^*_{2(k-1)}$ be the restriction of $L$ to inputs from $\mathbb{R} \langle \mathbf p \rangle_{2(k-1)}$. It follows that $\widetilde{L} \in \feasLass{k-1}{s}$. Moreover, $\widetilde{L}$ and $L$ induce the same values of $g$, $h$ and $h^+$, see \Cref{def_sdpValues}.
     \item Without loss of generality, assume that vertex $1$ has degree $1$, and is connected to vertex $2$, i.e., $ \{1,2\} \in E$. Partition $E := E(G)$ into $E^1 := \setFunct{ \{2,i\}}{i \in N(2)}$ and $E^2 := E \setminus E^1$. By Item \ref{item_subgraphDecomp}, we have $c(G,k) \leq c(G[E^1],k) + c(G[E^2],k)$. Since $G[E^1]$ is a star graph, we have $c(G[E^1],k) \leq 1$, see \eqref{eqn_monogamyStar}. Thus, $c(G,k) \leq 1 + c(G[E^2],k) \leq 1+\max_{G \in \mathcal{G}_{s-2}} c(G,k)$, as $G[E^2]$ is a graph on $s-2$ vertices.
     
     \item Using that $h_e = -1+g_e/2$, see \Cref{def_sdpValues}, we have
     \begin{align*}
         c(G,k) =  \max_{L \in \feasLass{k}{s}}  \sum_{e \in E} h_{e} = -|E(G)| + \frac{1}{2} \max_{L \in \feasLass{k}{s}}  \sum_{e \in E} g_{e} \geq \frac{\lambdaMax{H_G}}{2} - |E(G)|.
     \end{align*}
     The inequality follows from the fact that \ref{eqn_sdpRelax} provides an upper bound on $\lambdaMax{H_G}$, which is tight if $k \geq s$, see \Cref{lemma_exactnessResult}. \qedhere
     \end{enumerate}
 \end{proof}

 \begin{restatable}{lemma}{Gclassification}
\label{lemma_gClassification}
    Let $s,k \in \mathbb{N}$, with $s,k \geq 2$ and $s$ odd. Suppose $\max_{G \in \mathcal{G}_{s^\prime}} c(G,k) = \lfloor s^\prime / 2 \rfloor$ for all $2 \leq s^\prime < s$. 
    Let $\mathbf{G} \subseteq \mathcal{G}_s$ be the set of graphs that satisfy the following properties:
    \begin{enumerate}
        \item   \label{eqn_graphClassific} The graph $G$ is biconnected, triangle-free, and not bipartite.
        \item \label{eqn_degBounds}  For all $i \in V(G), \, 2 \leq \degr{i} \leq \frac{s-1}{2}$.
        \item \label{eqn_edgeMin}   The graph $G$ has at least $s$ edges, i.e.,  $|E(G)| \geq s$.
        \item   \label{eqn_stableSetNeighb}     For any stable set  $S \subseteq V(G), \,  \left| \cup_{i \in S} N(i) \right| \geq |S|+1$.
    \end{enumerate}
    We have that
    \begin{align}
        \label{eqn_impliciationReduction}
         \max_{G \in \mathcal{G}_s} c(G,k) = \left\lfloor \frac{s}{2} \right\rfloor \iff  \max_{G \in \mathbf{G}} c(G,k) \leq \left\lfloor \frac{s}{2} \right\rfloor.
    \end{align}
\end{restatable}
\begin{proof}
\noeqref{eqn_edgeMin}\noeqref{eqn_degBounds}\noeqref{eqn_graphClassific} \noeqref{eqn_stableSetNeighb}
    See \Cref{section_extraResults}.
\end{proof}
Among the properties stated in \Cref{lemma_gClassification}, it can be shown that property \ref{eqn_stableSetNeighb}  implies property \ref{eqn_degBounds}. Finding a list of all graphs in $\mathbf{G}$ is intractable for large $s$, due to the difficulty of checking property~\ref{eqn_stableSetNeighb}. However, it is clear that \eqref{eqn_impliciationReduction} remains valid if we replace $\mathbf{G}$ by a superset of $\mathbf{G}$ that is simpler to construct. The set of all graphs satisfying properties \ref{eqn_graphClassific} to \ref{eqn_edgeMin} is such a superset. The cardinality of this superset is still much smaller compared to $|\mathcal{G}_s|$, as shown in \Cref{table_lemmaReductionG}. Note that the 5-cycle is the only triangle-free non-bipartite graph on 5 vertices, since non-bipartite graphs must contain an odd cycle. In \Cref{table_lemmaReductionG}, the number of graphs satisfying properties \ref{eqn_graphClassific} to \ref{eqn_edgeMin} for $s = 15$ is not computed due to limited computational resources. We will use a relaxed version of property \ref{eqn_stableSetNeighb} (obtained by constraining $|S| \leq 2$), to reduce the cardinality of the superset even further for the case $s = 13$ (see \Cref{section_compuDetails}).
\begin{table}[ht!]
\centering
\begin{tabular}{l|l|l}
\hline
$s$ & $|\mathcal{G}_s|$ & \# graphs satisfying properties \ref{eqn_graphClassific} to \ref{eqn_edgeMin}  \\ \hline
3 & 4 & 0 \\
5 & 34 & 1 (the 5-cycle) \\
7 & 1044 & 6 \\
9 & 274668 & 219 \\
11 & 1018997864 & 26360 \\
13 & 50502031367952 & 9035088 \\
15 & 31426485969804308768 & -  \\ \hline
\end{tabular}
\caption{Comparison of $|\mathcal{G}_s|$ and the number of graphs satisfying properties \ref{eqn_graphClassific} to \ref{eqn_edgeMin} of \Cref{lemma_gClassification}.}
\label{table_lemmaReductionG}
\end{table}
\begin{lemma}   
\label{lemma_5VertexBound}
For $2 \leq s \leq 10$ and $k \geq 2$, $\max_{G \in \mathcal{G}_s} c(G,k) = \lfloor s / 2 \rfloor$. For $11 \leq s \leq 14$, \\ \mbox{$\max_{G \in \mathcal{G}_s} c(G,s) = \lfloor s /2 \rfloor$}.
 \end{lemma}
 \begin{proofsketch}
We use \Cref{lemma_gClassification} as follows: let $\mathbf{G}'$ be the set of graphs satisfying properties \ref{eqn_graphClassific} to \ref{eqn_edgeMin}. Observe that $\mathbf{G}'$ is a superset of $\mathbf{G}$, and that we may replace $\mathbf{G}$ by $\mathbf{G}'$ in \eqref{eqn_impliciationReduction}. Hence, it remains to show that $\max_{G \in \mathbf{G}'} c(G,k) \leq \lfloor s /2 \rfloor$. If $s$ is even, $\max_{G \in \mathbf{G}'} c(G,k) \leq \lfloor s /2 \rfloor$, $k 
\geq 2$, follows by Item \ref{item_vertexBound} of \Cref{lemma_vertexCover}. If $s \in \{3,5,7,9\}$, we verify that $\max_{G \in \mathbf{G}'} c(G,k) \leq \lfloor s /2 \rfloor$, $k = 2$, by solving the SDPs that define $c(G,k)$. By Item \ref{item_lasserImp} of \Cref{lemma_vertexCover}, this also proves the case $k > 2$. If $s \in \{11,13\}$, we proceed similarly, except instead of computing $c(G,s)$ via solving the SDP, we use $c(G,s) =  \lambdaMax{H_G}/2 - |E(G)|$, which is valid due to Item \ref{item_exactEVB} of \Cref{lemma_vertexCover} (Computing $c(G,s)$ in this way requires significantly less time than solving the SDP). We provide more computational details of the proof in \Cref{section_compuDetails}.
  \qedhere
  \end{proofsketch}
 
Extending \Cref{lemma_5VertexBound} to $s = 15$ is intractable with current methods: the Hamiltonians corresponding to graphs in $\mathcal{G}_{15}$ are matrices of order $2^{15} = 32768$, and the number of graphs satisfying properties \ref{eqn_graphClassific} to \ref{eqn_edgeMin} is likely to be of the order $10^8$, see \Cref{table_lemmaReductionG}. However, \Cref{lemma_5VertexBound} already provides the following improved lower bounds on the approximation ratio of \Cref{alg:approxLee}. 

\begin{theorem}
    \label{thm_boundN5}
    For $k \geq 2$, the approximation ratio of \Cref{alg:approxLee} is at least $\alpha(10/11) \geq 0.602$, see \Cref{table_alphaMuVals}. If $k \geq 13$, the approximation ratio is at least $\alpha(14/15) \geq \ratioGeneral{}$.
\end{theorem}
\begin{proof}
By \Cref{corr_simpleCons} and \Cref{lemma_5VertexBound}, we have that $(\frac{10}{11}h^+_e)_{e \in E}$ is contained in the matching polytope when $k \geq 2$, and $(\frac{14}{15}h^+_e)_{e \in E}$ when $k \geq 13$. The claim then follows from \Cref{lemma_approxRatioLee}.
\end{proof}
Our results also imply the following non-linear inequalities for \ref{eqn_sdpRelax}.
\begin{lemma}
\label{lemma_connectionHandC}
    Let $L \in \feasLass{k}{n}$, see \eqref{eqn_feasLass}, and $k \geq 2$. Then, for the values $h^+$ derived from $L$ as in \Cref{def_sdpValues}, we have
        $\sum_{1 \leq i < j \leq s } h^+_{ij} \leq \lfloor s / 2 \rfloor \text{ for all } s \in \{2, \, 3, \dots, \, 10 \}$.
    These bounds are tight and extend to $s \in \{11,\dots,14\}$ if $k \geq s$. 
\end{lemma}

\section{New approximation algorithm on triangle-free graphs}
\label{section_triangleFreeGraphs}
In this section, we propose a QMC approximation algorithm for triangle-free graphs that achieves an approximation ratio of at least $\ratioNoTri{}$.
Our algorithm is inspired by \cite[Alg.~17]{king2023improved}, which is also designed for triangle-free graphs and achieves an approximation ratio of at least $0.582$.  

We present \Cref{alg:triangleFree} below. The input parameter $\Theta$ of \Cref{alg:triangleFree} is a real-valued function from the set $\mathcal{A}$ that we will define in \Cref{section_propTheta}. The restriction $\Theta \in \mathcal{A}$ is required for the computation of the approximation ratio of \Cref{alg:triangleFree} in \Cref{section_approxRatioTriFree}.

\begin{algorithm}[QMC approximation algorithm for triangle-free graphs]
\label{alg:triangleFree}
Input: triangle-free graph $G = (V,E, w)$, function $\Theta \in \mathcal{A}$, see \eqref{eqn_thetaProp}.
\begin{enumerate}
	\item \label{step_algTriangle1} Solve \ref{eqn_sdpRelax} for $k = 13$ to obtain the vectors $\mathbf{v}_i$ and values $(h_e^+)_{e \in E}$, see \eqref{eqn_gramVectors} and \Cref{def_sdpValues} respectively. Compute the vectors $u_i \in \mathbb{R}^3$, $i \in V$,
 as in \Cref{item_uComputation} of \Cref{alg:approxLee}. For each $i \in V$, let $\xi_i \in \C^2$ be a unit length vector satisfying $\ket{\xi_i} \bra{\xi_i} = \frac{1}{2}\left(\eye_2 + u_{i,1} X + u_{i,2} Y + u_{i,3} Z\right)$.
    \item \label{algStep_defTheta} Set $\theta_{e} := \arcsin{ \sqrt{\Theta\mleft(h^+_e\mright) } }$ for all $e \in E$.
    \item \label{step_triangleFree5}  For all $i \in V$, set $P_i := \begin{bmatrix}
    0 & 1 \\
        \exp{\left( \left( 2 \argFunc{\xi^*_{i,1} \xi_{i,2}} + \pi \right) \imagUnit    \right)}             & 0
        \end{bmatrix}$.
    \item \label{step_secondMatching} 
    Compute a maximum weight matching $\mathcal{M}$ on the modified graph $\widetilde{G} = (V,E,\widetilde{w})$, where $\widetilde{w}$ is defined as
    \begin{align}
        \label{eqn_wTildeDef}
        \widetilde{w}_e := w_e \, q(  h_e ) \sqrt{\Theta\mleft(h^+_e\mright) \left(1- \Theta\left(1-h^+_e\right)\right)},
    \end{align}
    for $q$ as in \eqref{eqn_hyperGeoQ}. Let $U \subseteq V$ be the set of vertices unmatched by $\mathcal{M}$.
	\item \label{step_choosingAngles}
    For all $\edge{i}{j} \in E$, let 
    \begin{align}
    \label{eqn_gammaDef}
    \gamma_{ij} := \pi - \argFunc{ \bra{\xi_i} e^{\varphi_j \imagUnit}{P}_j \ket{\xi_j} \bra{\xi_j} e^{\varphi_i \imagUnit} P_i \ket{\xi_i} }, 
    \end{align}
    where the values of $\varphi_i$, $i \in V$, are chosen as follows: for each $i \in U$, draw $\varphi_i \in [0,2\pi)$ uniformly at random. For $\edge{i}{j} \in E$ with $\mathcal{M}_{\{i,j\}} = 1$, draw $\varphi_i \in [0, 2\pi)$ uniformly at random, and choose $\varphi_j \in [0,2 \pi)$ such that $\gamma_{ij} = \pi/2$.
    \item For all $i \in V$, set $\widetilde{P}_i := e^{\varphi_i \imagUnit} \left( \eye_2^{\otimes(i-1)} \otimes  P_i \otimes \eye_2^{\otimes(n-i)} \right)$.
    
	\item \label{step_triRho2} Compute $\rho_2$ as in \Cref{step_matching} of \Cref{alg:approxLee}.
    \item \label{step_stateGeneration}
    Let
    \begin{align}
        \label{eqn_xiDefinition}
    	\ket{\xi} := \prod_{\edge{i}{j} \in E} \exp{\left(\frac{\imagUnit}{2} \signFunction{\gamma_{ij}} \theta_{ij} \widetilde{P}_i  \widetilde{P}_j \right)} \bigotimes_{i \in V} |\xi_i\rangle.       
    \end{align}
    \item Return the state $\ket{\xi} \bra{\xi}$ if $\Tr{ \ket{\xi}\bra{\xi} H_G} \geq \Tr{\rho_2 H_{G}}$, and state $\rho_2$ otherwise.
\end{enumerate}
\end{algorithm}
\Cref{alg:triangleFree} improves over \cite[Alg.~17]{king2023improved} in three ways. Firstly, we optimize the algorithm parameter $\Theta$ over a larger space. King chose $\Theta(x) = Rx^2$ and (numerically) optimized the value of $R \in \mathbb{R}$ to obtain the highest approximation ratio. In contrast, we consider functions $\Theta$ over a set $\mathcal{A}$, which we prove contains functions of the form $Rx^2$, but also $Rx$ and $1-e^{-Rx}$ (\Cref{lemma_linFunctionInOmega}). We determine a near-optimal $\Theta \in \mathcal{A}$ in \Cref{section_findingCandidateTheta}. Secondly, in \Cref{step_choosingAngles}, we choose the values of $\varphi_i$ based on a  maximum weight matching on a modified graph, which we later show improves over drawing all the $\varphi_i$ uniformly at random as in \cite[Alg.~17]{king2023improved}. Thirdly, we output the state $\rho_2$ from \Cref{step_matching} of \Cref{alg:approxLee} if the state $\ket{\xi} \bra{\xi}$ performs worse. 

Note that \Cref{alg:triangleFree} computes maximum weight matchings in \Cref{step_secondMatching,step_triRho2}. To compute the approximation ratio of \Cref{alg:triangleFree}, we relate the weight of these matchings to $h^+$ from \Cref{def_sdpValues}, as in \cite{lee2024improved}. Given a maximum weight matching $\mathcal{M}$ on $G=(V,E,w)$, this relation is the inequality $\sum_{e \in E(G)} w_e \mathcal{M}_e \geq \mu \sum_{e \in E(G)} w_e h^+_e$. Here, the value $\mu \in [0,1]$ is such that  $(\mu h^+_e)_{e \in E(G)}$ is contained in the matching polytope. Since \Cref{alg:triangleFree} uses an SDP relaxation level of $k = 13$, we may set $\mu = 14/15$, as explained in \Cref{section_sharpenedAnalysis}.

\subsection{Properties of \texorpdfstring{$\Theta$}{Theta}}
\label{section_propTheta}
We require that the function $\Theta$ used in \Cref{alg:triangleFree} is an element of the set
\begin{equation}
    \label{eqn_thetaProp}
    \mathcal{A} := \left\{ \Theta : \mathbb{R} \to \mathbb{R} : \, \,
    \begin{aligned}
        & \Theta(0) = 0, \, \Theta(1) \leq 1, \Theta \text{ an increasing function, and for all} \\[0.5ex]
        &  \text{fixed } c \in [0,1], \, \min_{x \in [0,c]}(1-\Theta(x))(1-\Theta(c-x)) = 1-\Theta(c)
    \end{aligned}
    \right\}.
\end{equation}

 The functions in $\mathcal A$ satisfy the following property, which is a generalization of \cite[Cor.~8]{king2023improved}. 
\begin{lemma}
\label{lemma_thetaProdLB}
For $\Theta \in \mathcal{A}$, and values $x_0, x_1,\dots, x_p \geq 0$ satisfying $x_0 + \sum_{s \in [p]} x_s \leq 1$, we have
\begin{align*}
    \prod_{s \in [p]} \left( 1-\Theta(x_s) \right) \geq 1-\Theta(1-x_0).
\end{align*}
\end{lemma}
\begin{proof}
    Since $\Theta \in \mathcal{A}$, $(1-\Theta(x_s))(1-\Theta(x_{s'})) \geq 1-\Theta(x_s+ x_{s'})$ for any distinct $s, s' \in [p]$. Iteratively 
    applying the inequality 
    shows that
        $\prod_{s \in [p]} \left( 1-\Theta(x_{s}) \right)  \geq 1-\Theta\left( \sum_{s \in [p]} x_s \right)$.
    Since $\Theta$ is an increasing function, and $\sum_{s \in [p]} x_{s} \leq 1 - x_0$, we finally find $1-\Theta\left( \sum_{s \in [p]} x_s \right) \geq 1-\Theta(1-x_0)$.
\end{proof}

\subsection{Computing the approximation ratio of \texorpdfstring{\Cref{alg:triangleFree}}{Algorithm 2}}
\label{section_approxRatioTriFree}
We compute a lower bound on the approximation ratio of \Cref{alg:triangleFree}. We use \cite[Lem.~12]{king2023improved}, which provides a lower bound on the expected energy of the state $\ket{\xi} \bra{\xi}$, see \eqref{eqn_xiDefinition}, in terms of $\theta_{ij}$, $\gamma'_{e}$ and $A_{ij}, B_{ij}$. Here, 
\begin{align}
\label{eqn_gammaPrimeDef}
    \gamma'_{e} := \gamma_{e} + \pi \frac{1-\signFunction{\gamma_{e}}}{2} \in [0, \pi],
\end{align}
where $\gamma_e$ is given by \eqref{eqn_gammaDef}, and
\begin{equation}
\label{eqn_edgeTermProds}
\begin{aligned}
    A_{ij} := \prod_{k \in N(i) \setminus \{j\}} \cos{\theta_{ik}}, \, \, \,  B_{ij} :=  \prod_{k \in N(j) \setminus \{i\}} \cos{\theta_{k j}}.
    \end{aligned}
\end{equation}

\begin{lemma}[\cite{king2023improved}]
   \label{lemma_expectedAdditionToObj}
   Let $G$ be a triangle-free graph used as input to \Cref{alg:triangleFree}. Let $\ket{\xi}$ be as in \eqref{eqn_xiDefinition}, and $A_{ij}$, $B_{ij}$ as in \eqref{eqn_edgeTermProds}. Then, for an edge $\edge{i}{j} \in E(G)$, we have
    \begin{align*}
    \mathbb{E}\bra{\xi} H_{ij} \ket{\xi} \geq \mathbb{E}\left(E_{ij} \right) \left( 1 +  A_{ij} B_{ij} + \mathbb{E}\left[ \sin \gamma'_{ij} \right] \sin{\mleft(\theta_{ij}\mright)} \left( A_{ij}+B_{ij} \right) \right),
    \end{align*}
    where $E_{ij} := (1-u_i^\top u_j)/2$, see \Cref{step_algTriangle1} of \Cref{alg:triangleFree}.
\end{lemma}
There are two differences in presentation between \cite[Lem.~12]{king2023improved} and \Cref{lemma_expectedAdditionToObj}. Firstly, we use a different scaling of $H_{ij}$ and $\theta_{ij}$ compared to \cite{king2023improved}. Secondly, in \cite[Alg.~17]{king2023improved}, the parameter $\gamma'_e$ is uniform random on $[0,\pi]$. Therefore, in \cite[Alg.~17]{king2023improved}, the expectation of $\sin{\gamma'_e}$ is given by $2/ \pi$. This is in contrast to
\Cref{alg:triangleFree}, where the distribution of $\gamma_e$ depends on the matching $\mathcal{M}$ computed in \Cref{step_secondMatching} of \Cref{alg:triangleFree}. Using that $\gamma_e$ can be written as
\begin{align*}
    \gamma_{ij} = \pi - \modSubscript{\argFunc{ \bra{\xi_i}  P_j \ket{\xi_j} \bra{\xi_j} P_i \ket{\xi_i}} + \varphi_i + \varphi_j }{2 \pi },
\end{align*}
it follows that $\gamma'_e$ is uniform random on $[0,\pi]$ if $\mathcal{M}_e = 0$, or equal to $\pi/2$, if $\mathcal{M}_e = 1$, see \Cref{step_choosingAngles} of \Cref{alg:triangleFree}. Thus, we have that 
\begin{align}
    \label{eqn_angleExpectedVal}
    \mathbb{E}\left[ \sin \gamma'_{e} \right] = (1-\mathcal{M}_e) \int_{0}^\pi \frac{ \sin{ x} }{\pi} \mathrm{d}x + \mathcal{M}_e \sin{\left( \frac{\pi}{2} \right)} = \frac{2}{\pi} + \left( 1 - \frac{2}{\pi} \right) \mathcal{M}_e.
\end{align}

We define the functions
\begin{align}
    \label{eqn_betaDef}
    \beta_{\Theta}(x,\mu) :=  q(x) \left(1 - \frac{\Theta(1-x^+) }{2} + \left( \frac{2}{\pi}  + \mu \frac{\pi-2}{\pi} x^+ \right) \sqrt{\Theta(x^+) (1- \Theta(1-x^+))}\right),
\end{align}
for $q$ as in \eqref{eqn_hyperGeoQ},
\begin{align}
        \label{eqn_zetaFunctDef}
        \zeta_{\Theta}(x,\mu,p) :=  \frac{ p \beta_{\Theta}(x, \mu) +(1-p) \left(1+3 \mu x^+ \right)}{2+2x}, 
\end{align}
and
\begin{align}
    \label{eqn_zetaStarDef}
    \zeta^*_{\Theta}(\mu,p) :=  \min_{x \in (-1,1]} \zeta_{\Theta}(x,\mu,p).
\end{align}
These functions will be used in the following theorem that establishes a lower bound on the approximation ratio of \Cref{alg:triangleFree}.

\begin{theorem}    \label{thm_ratioTriangleFree}
    The approximation ratio of \Cref{alg:triangleFree} for triangle-free graphs is at least 
    \begin{align*}
        \max_{p \in [0,1]} \zeta^*_{\Theta}(14/15,p).
    \end{align*}
\end{theorem}
\begin{proof}
Let $e = \{i,j\}$.
Using \Cref{lemma_thetaProdLB}, we find, for $A_{ij}$ as in \eqref{eqn_edgeTermProds},
\begin{align}
    \label{eqn_boundProdTheta2}
    A_{ij}=\prod_{k \in N(i) \setminus \{j\}} \cos{\theta_{i k}} = \prod_{k \in N(i) \setminus \{j\}} \sqrt{1- \Theta(h^+_{i k})} \geq \sqrt{1 - \Theta \mleft( 1-h^+_e \mright)},
\end{align}
and similarly $B_{ij} \geq \sqrt{1 - \Theta \mleft( 1-h^+_e \mright)}$. Indeed, \Cref{lemma_thetaProdLB} applies here since the values $(h^+_{ik})_{k \in N(i)}$ satisfy $h^+_{ik} \geq 0$ for all $k \in N(i)$, and $h^+_{ij}+\sum_{k \in N(i) \setminus \{j \}} h^+_{ik} \leq 1$, due to \eqref{eqn_monogamyStar}. 

We substitute \eqref{eqn_boundProdTheta2} into the expression of \Cref{lemma_expectedAdditionToObj} (note that $\sin{\theta_{e}} = \sqrt{\Theta \mleft(h^+_e \mright)}$). Additionally, we substitute $\mathbb{E}\left(E_{ij} \right) = q(h_e)/2$ \cite[Lem.~2.1]{BrietGrothendieck}, for $q$ as in \eqref{eqn_hyperGeoQ} (see also \cite[Lem.~13]{king2023improved}). This yields
\begin{align}
    \mathbb{E}\bra{\xi} H_e \ket{\xi} &\geq  \frac{q(h_e)}{2} \left( 1 + \left( 1 - \Theta \mleft( 1-h^+_e \mright) \right) + 2 \mathbb{E}\left[ \sin \gamma'_{e} \right] \sqrt{\Theta \mleft(h^+_e \mright)} \sqrt{1 - \Theta\left(1-h^+_{e}\right)} \right) \nonumber \\
    &= q(h_e) \left( 1 - \frac{\Theta\left(1-h^+_{e}\right)}{2} + \left( \frac{2}{\pi} + \frac{\pi - 2}{\pi } \mathcal{M}_e \right)  \sqrt{\Theta \mleft(h^+_e \mright) \left( 1 - \Theta \mleft( 1-h^+_e \mright) \right)} \right) \nonumber \\
    & = q(h_e) \left( 1 - \frac{\Theta \mleft( 1-h^+_e \mright)}{2} + \frac{2}{\pi}  \sqrt{\Theta \mleft(h^+_e \mright) \left( 1 - \Theta \mleft( 1-h^+_e \mright) \right)} \right) \nonumber \\
    & \hphantom{\, \, =} + \frac{\pi - 2}{\pi} \mathcal{M}_e \, q(h_e) \sqrt{\Theta \mleft(h^+_e \mright) \left( 1 - \Theta \mleft( 1-h^+_e \mright) \right)} \nonumber \\
    &= \beta_{\Theta}(h_e,0) + \frac{\pi - 2}{ \pi }  \frac{\widetilde{w}_e}{w_e} \mathcal{M}_e. \label{eqn_expressionEH}
\end{align}
The first equality is due to \eqref{eqn_angleExpectedVal}, for $\mathcal{M}$ the matching computed in \Cref{step_secondMatching} of \Cref{alg:triangleFree}. The third equality is due to the definitions of $\widetilde{w}$ and $\beta_\Theta$, see \eqref{eqn_wTildeDef} and \eqref{eqn_betaDef} respectively. By combining \eqref{eqn_hamiltonDef} and \eqref{eqn_expressionEH}, we have
\begin{align}
    \label{eqn_lemmaElowerBound}
    \mathbb{E}\bra{\xi} H_G \ket{\xi} &=  \sum_{e \in E } w_e \mathbb{E}\bra{\xi} H_e \ket{\xi} \geq \sum_{e \in E } \left( w_e \beta_{\Theta}(h_e,0) + \frac{\pi - 2}{ \pi }  \widetilde{w}_e \mathcal{M}_e \right).
\end{align}
Here, $\mathcal{M}$ corresponds to a maximum weight matching on the graph with positive edge weights $\widetilde{w}$. By \Cref{corr_simpleCons} and \Cref{lemma_5VertexBound}, $(\mu h^+_e)_{e \in E}$ is contained in the matching polytope for $\mu = 14/15$. Therefore, $\sum_{e \in E} \widetilde{w}_e \mathcal{M}_e \geq \mu \sum_{e \in E} \widetilde{w}_e h^+_e$, which we substitute in \eqref{eqn_lemmaElowerBound} and exploit \eqref{eqn_wTildeDef} to obtain
\begin{align}
    \mathbb{E}\bra{\xi} H_G \ket{\xi} & \geq \sum_{e \in E } \left( w_e \beta_{\Theta}(h_e,0) + \frac{\pi - 2}{ \pi }  \widetilde{w}_e \mu h_e^+\right) \nonumber \\
    &= \sum_{e \in E } \left( w_e \beta_{\Theta}(h_e,0) + \frac{\pi - 2}{ \pi }  w_e \, q(h_e) \sqrt{\Theta \mleft(h^+_e\mright) \left(1- \Theta\left(1-h^+_e\right)\right)} \mu h_e^+\right) \nonumber \\
    &= \sum_{e \in E } w_e \Bigg[ q(h_e) \left(1 - \frac{\Theta(1-h^+_e) }{2} +  \frac{2}{\pi}  \sqrt{\Theta(h^+_e) (1- \Theta(1-h^+_e))} \right) \nonumber \\
    & \phantom{= \sum_{e \in E} w_e \Bigg[} + \mu \frac{\pi - 2}{ \pi }  h^+_e \, q(h_e) \sqrt{\Theta \mleft(h^+_e\mright) \left(1- \Theta\left(1-h^+_e\right)\right)} \Bigg] \nonumber \\
    &= \sum_{e \in E} w_e q(h_e) \left(  1 - \frac{\Theta(1-h^+_e) }{2} + \left( \frac{2}{\pi}  + \mu \frac{\pi-2}{\pi} h^+_e \right) \sqrt{\Theta(h^+_e) (1- \Theta(1-h^+_e))}  \right) \nonumber \\
    &= \sum_{e \in E} w_e  \beta_{\Theta}\left(h_e, \mu \right).
    \label{eqn_boundExi}
\end{align}
Here, the second equality follows from substituting the definition of $\beta_{\Theta}(h_e,0)$, see \eqref{eqn_betaDef} (and factoring $w_e$). The fourth equality is by the definition of $\beta_{\Theta}(h_e,\mu)$.

As for $\rho_2$, one can show that $\Tr{ \rho_2 H_G}  \geq \sum_{e \in E} w_e \left(1 + 3 \mu h^+_e \right)$, see \cite[Eq.~9]{lee2024improved}. The expected energy attained by \Cref{alg:triangleFree} satisfies
\begin{align*}
    \mathbb{E} \max \left\{ \bra{\xi}H_G \ket{\xi}, \Tr{\rho_2 H_G} \right\} &\geq \max_{p \in [0,1]} \left[ p \mathbb{E} \left[ \bra{\xi}H_G \ket{\xi} \right] +(1-p) \Tr{\rho_2 H_G} \right]\\
    &\geq \max_{p \in [0,1]} \sum_{e \in E} w_e \left( p \beta_{\Theta}(h_e, \mu) +(1-p)\left(1+3 \mu h^+_e \right) \right) \\
    &\geq \max_{p \in [0,1]} \zeta^*_{\Theta}(\mu,p) \sum_{e \in E} w_e (2+2h_e) \geq \max_{p \in [0,1]} \zeta^*_{\Theta}(\mu,p) \lambdaMax{H_G}.
\end{align*}
We have used \eqref{eqn_boundExi} for the second inequality. The fourth inequality follows from the fact that \ref{eqn_sdpRelax} provides an upper bound on $\lambdaMax{H_G}$, as explained in \Cref{section_SDP_defs}. Note that $h_e \in [-1,1]$, see \Cref{def_sdpValues}, while the minimization in $\zeta^*_{\Theta}(\mu,p)$ is done over $x \in (-1,1]$. This distinction is made to avoid division by $0$ and can be done without loss of generality, see \Cref{remark_divByZero}. \qedhere

\end{proof}

It remains to find a function $\Theta \in \mathcal{A}$ for which $\max_{p \in [0,1]} \zeta^*_{\Theta}(14/15,p)$ is high.

\subsection{Finding candidate functions \texorpdfstring{$\Theta$}{Theta}}
\label{section_findingCandidateTheta}
The set $\mathcal{A}$ is difficult to characterize in general, but is easily shown to contain the following classes of functions.
\begin{lemma}
    \label{lemma_linFunctionInOmega}
    For $R \geq 0$, $1-e^{-Rx} \in \mathcal{A}$. For $R \in [0,1]$, $Rx \in \mathcal{A}$ and $Rx^2 \in \mathcal{A}$.
\end{lemma}
\begin{proof}
Observe first that if $\Theta$ is $x \mapsto 1-e^{-Rx}$ (for $R \geq 0)$, $x \mapsto Rx$, or $x \mapsto Rx^2$ (for $R \in [0,1]$), then $\Theta$ is increasing, satisfies $\Theta(0)=0$ and $\Theta(1) \leq 1$. It remains to show that any such function $\Theta$ satisfies $\min_{x \in [0,c]}(1-\Theta(x))(1-\Theta(c-x)) = 1-\Theta(c)$.

For the case $\Theta(x) = 1-e^{-Rx}$, we have
$\min_{x \in [0,c] } (1-\Theta(x))(1-\Theta(c-x)) = \min_{x \in [0,c] }  1-\Theta(c) = 1-\Theta(c)$. For the case $\Theta(x) = Rx$, we define $f_c(x) := (1-\Theta(x))(1-\Theta(c-x)) = (1-Rx)(1-R(c-x))$ and verify that $f_c^{''}(x)= -2R^2 \leq 0$. We conclude that $f_c(x)$ is concave on $[0,c]$. A concave function has its minimum attained at a boundary point, and since $f_c$ is symmetric around $c/2$, its minimum equals $f_c(0) = f_c(c) = 1 - \Theta(c)$. For the case $\Theta(x) = Rx^2$, we again define $f_c(x) := (1-\Theta(x))(1-\Theta(c-x)) = (1-Rx^2)(1-R(c-x)^2)$. It follows that
\begin{align*}
    f''_c(x) = 2R(c^2 R -6cRx+ 6Rx^2-2) = 2R\left( c^2R - 2 +6R x \left( x-c \right) \right).
\end{align*}
Since $0 \leq x \leq c \leq 1$ and $R \in [0,1]$, it follows that $6Rx (x-c) \leq 0$. Moreover, $c,R \in [0,1]  \implies c^2R -2 \leq 0$. Therefore, $f''_c(x) \leq 0$. Hence, $f_c(x)$ is concave, which implies that $\min_{x \in [0,c]} f_c(x) = f_c(0)$.
\qedhere
\end{proof}

We choose $\Theta(x) = \optTheta{}$. With this choice, \Cref{alg:triangleFree} achieves an approximation ratio of at least 
\begin{align}
    \label{eqn_bestBoundTriangleFree}
     \max_{p \in [0,1]}\zeta^*_{\optTheta{}}(14/15,p) \geq \zeta^*_{\optTheta{}}(14/15,p^*) =
     \unroundedApproxOpt{},
\end{align}
where $p^* = \left( 2 \cdot \unroundedApproxOpt{} - 1 \right) / \left( \beta_{\optTheta{}}(0,14/15)-1 \right)$. This value of $p^*$ is chosen such that $\zeta_{\Theta}(0,14/15,p^*) = \unroundedApproxOpt{}$, see \eqref{eqn_zetaFunctDef}. Let $\zeta(x) := \zeta_{\optTheta{}}(x,14/15,p^*)$, which is the function plotted in \Cref{fig:zetaFunction}. We briefly elaborate on how one can prove the statement $\zeta^*_{\optTheta{}}(14/15,p^*) = \min_{x \in (-1,1]} \zeta(x) =      \unroundedApproxOpt{}$ in \eqref{eqn_bestBoundTriangleFree}. It can be shown (details omitted) that $\zeta$ is decreasing for $x \in (-1,0]$ and increasing for $x \in [0,0.035]$, so that $\min_{x \in (-1,0.035]} \zeta(x) = \zeta(0) = \unroundedApproxOpt{}$. For $x \in [0.035,1]$, it is possible to derive a lower bound on $\zeta'(x)$, and evaluate $\zeta$ on a fine grid in the interval $[0.035,1]$. By combining the lower bound on $\zeta'$ with the Mean Value theorem, one can prove that $\zeta(x) \geq \unroundedApproxOpt{}$ for any $x \in [0.035,1]$.

\begin{figure}
    \centering
\begin{tikzpicture}[trim axis left, trim axis right]
\begin{axis}[
    axis lines=left,
    ylabel style={rotate=-90},
    xlabel={$x$},
    ylabel={$\zeta_\Theta(x,14/15,p^*) \quad$},
    xmin=-0.1, xmax=0.8,
    ymin=0.61, ymax=0.64,
    ytick={0.61,0.62,0.63,0.64}, 
    legend pos=north west,
]

\addplot[red, mark=*, mark options={scale=1, fill=red}] coordinates {(0, 0.61383)};

\addplot[
    color=black,
    thick
    ]
    coordinates {
   (-0.20000,0.70423)(-0.19960,0.70401)(-0.19920,0.70378)(-0.19880,0.70355)(-0.19840,0.70333)(-0.19800,0.70310)(-0.19760,0.70287)(-0.19719,0.70265)(-0.19679,0.70242)(-0.19639,0.70219)(-0.19599,0.70197)(-0.19559,0.70174)(-0.19519,0.70152)(-0.19479,0.70129)(-0.19439,0.70107)(-0.19399,0.70084)(-0.19359,0.70062)(-0.19319,0.70040)(-0.19279,0.70017)(-0.19238,0.69995)(-0.19198,0.69973)(-0.19158,0.69950)(-0.19118,0.69928)(-0.19078,0.69906)(-0.19038,0.69883)(-0.18998,0.69861)(-0.18958,0.69839)(-0.18918,0.69817)(-0.18878,0.69795)(-0.18838,0.69773)(-0.18798,0.69751)(-0.18758,0.69728)(-0.18717,0.69706)(-0.18677,0.69684)(-0.18637,0.69662)(-0.18597,0.69640)(-0.18557,0.69618)(-0.18517,0.69596)(-0.18477,0.69574)(-0.18437,0.69553)(-0.18397,0.69531)(-0.18357,0.69509)(-0.18317,0.69487)(-0.18277,0.69465)(-0.18236,0.69443)(-0.18196,0.69422)(-0.18156,0.69400)(-0.18116,0.69378)(-0.18076,0.69356)(-0.18036,0.69335)(-0.17996,0.69313)(-0.17956,0.69291)(-0.17916,0.69270)(-0.17876,0.69248)(-0.17836,0.69227)(-0.17796,0.69205)(-0.17756,0.69183)(-0.17715,0.69162)(-0.17675,0.69140)(-0.17635,0.69119)(-0.17595,0.69098)(-0.17555,0.69076)(-0.17515,0.69055)(-0.17475,0.69033)(-0.17435,0.69012)(-0.17395,0.68991)(-0.17355,0.68969)(-0.17315,0.68948)(-0.17275,0.68927)(-0.17234,0.68905)(-0.17194,0.68884)(-0.17154,0.68863)(-0.17114,0.68842)(-0.17074,0.68821)(-0.17034,0.68799)(-0.16994,0.68778)(-0.16954,0.68757)(-0.16914,0.68736)(-0.16874,0.68715)(-0.16834,0.68694)(-0.16794,0.68673)(-0.16754,0.68652)(-0.16713,0.68631)(-0.16673,0.68610)(-0.16633,0.68589)(-0.16593,0.68568)(-0.16553,0.68547)(-0.16513,0.68526)(-0.16473,0.68505)(-0.16433,0.68485)(-0.16393,0.68464)(-0.16353,0.68443)(-0.16313,0.68422)(-0.16273,0.68401)(-0.16232,0.68381)(-0.16192,0.68360)(-0.16152,0.68339)(-0.16112,0.68319)(-0.16072,0.68298)(-0.16032,0.68277)(-0.15992,0.68257)(-0.15952,0.68236)(-0.15912,0.68215)(-0.15872,0.68195)(-0.15832,0.68174)(-0.15792,0.68154)(-0.15752,0.68133)(-0.15711,0.68113)(-0.15671,0.68092)(-0.15631,0.68072)(-0.15591,0.68051)(-0.15551,0.68031)(-0.15511,0.68011)(-0.15471,0.67990)(-0.15431,0.67970)(-0.15391,0.67950)(-0.15351,0.67929)(-0.15311,0.67909)(-0.15271,0.67889)(-0.15230,0.67869)(-0.15190,0.67848)(-0.15150,0.67828)(-0.15110,0.67808)(-0.15070,0.67788)(-0.15030,0.67768)(-0.14990,0.67748)(-0.14950,0.67727)(-0.14910,0.67707)(-0.14870,0.67687)(-0.14830,0.67667)(-0.14790,0.67647)(-0.14749,0.67627)(-0.14709,0.67607)(-0.14669,0.67587)(-0.14629,0.67567)(-0.14589,0.67547)(-0.14549,0.67527)(-0.14509,0.67507)(-0.14469,0.67488)(-0.14429,0.67468)(-0.14389,0.67448)(-0.14349,0.67428)(-0.14309,0.67408)(-0.14269,0.67388)(-0.14228,0.67369)(-0.14188,0.67349)(-0.14148,0.67329)(-0.14108,0.67310)(-0.14068,0.67290)(-0.14028,0.67270)(-0.13988,0.67251)(-0.13948,0.67231)(-0.13908,0.67211)(-0.13868,0.67192)(-0.13828,0.67172)(-0.13788,0.67153)(-0.13747,0.67133)(-0.13707,0.67114)(-0.13667,0.67094)(-0.13627,0.67075)(-0.13587,0.67055)(-0.13547,0.67036)(-0.13507,0.67016)(-0.13467,0.66997)(-0.13427,0.66977)(-0.13387,0.66958)(-0.13347,0.66939)(-0.13307,0.66919)(-0.13267,0.66900)(-0.13226,0.66881)(-0.13186,0.66861)(-0.13146,0.66842)(-0.13106,0.66823)(-0.13066,0.66804)(-0.13026,0.66784)(-0.12986,0.66765)(-0.12946,0.66746)(-0.12906,0.66727)(-0.12866,0.66708)(-0.12826,0.66689)(-0.12786,0.66670)(-0.12745,0.66651)(-0.12705,0.66631)(-0.12665,0.66612)(-0.12625,0.66593)(-0.12585,0.66574)(-0.12545,0.66555)(-0.12505,0.66536)(-0.12465,0.66517)(-0.12425,0.66499)(-0.12385,0.66480)(-0.12345,0.66461)(-0.12305,0.66442)(-0.12265,0.66423)(-0.12224,0.66404)(-0.12184,0.66385)(-0.12144,0.66366)(-0.12104,0.66348)(-0.12064,0.66329)(-0.12024,0.66310)(-0.11984,0.66291)(-0.11944,0.66273)(-0.11904,0.66254)(-0.11864,0.66235)(-0.11824,0.66217)(-0.11784,0.66198)(-0.11743,0.66179)(-0.11703,0.66161)(-0.11663,0.66142)(-0.11623,0.66124)(-0.11583,0.66105)(-0.11543,0.66086)(-0.11503,0.66068)(-0.11463,0.66049)(-0.11423,0.66031)(-0.11383,0.66012)(-0.11343,0.65994)(-0.11303,0.65975)(-0.11263,0.65957)(-0.11222,0.65939)(-0.11182,0.65920)(-0.11142,0.65902)(-0.11102,0.65883)(-0.11062,0.65865)(-0.11022,0.65847)(-0.10982,0.65829)(-0.10942,0.65810)(-0.10902,0.65792)(-0.10862,0.65774)(-0.10822,0.65755)(-0.10782,0.65737)(-0.10741,0.65719)(-0.10701,0.65701)(-0.10661,0.65683)(-0.10621,0.65664)(-0.10581,0.65646)(-0.10541,0.65628)(-0.10501,0.65610)(-0.10461,0.65592)(-0.10421,0.65574)(-0.10381,0.65556)(-0.10341,0.65538)(-0.10301,0.65520)(-0.10261,0.65502)(-0.10220,0.65484)(-0.10180,0.65466)(-0.10140,0.65448)(-0.10100,0.65430)(-0.10060,0.65412)(-0.10020,0.65394)(-0.09980,0.65376)(-0.09940,0.65358)(-0.09900,0.65340)(-0.09860,0.65323)(-0.09820,0.65305)(-0.09780,0.65287)(-0.09739,0.65269)(-0.09699,0.65251)(-0.09659,0.65234)(-0.09619,0.65216)(-0.09579,0.65198)(-0.09539,0.65180)(-0.09499,0.65163)(-0.09459,0.65145)(-0.09419,0.65127)(-0.09379,0.65110)(-0.09339,0.65092)(-0.09299,0.65075)(-0.09259,0.65057)(-0.09218,0.65039)(-0.09178,0.65022)(-0.09138,0.65004)(-0.09098,0.64987)(-0.09058,0.64969)(-0.09018,0.64952)(-0.08978,0.64934)(-0.08938,0.64917)(-0.08898,0.64899)(-0.08858,0.64882)(-0.08818,0.64864)(-0.08778,0.64847)(-0.08737,0.64830)(-0.08697,0.64812)(-0.08657,0.64795)(-0.08617,0.64777)(-0.08577,0.64760)(-0.08537,0.64743)(-0.08497,0.64726)(-0.08457,0.64708)(-0.08417,0.64691)(-0.08377,0.64674)(-0.08337,0.64657)(-0.08297,0.64639)(-0.08257,0.64622)(-0.08216,0.64605)(-0.08176,0.64588)(-0.08136,0.64571)(-0.08096,0.64553)(-0.08056,0.64536)(-0.08016,0.64519)(-0.07976,0.64502)(-0.07936,0.64485)(-0.07896,0.64468)(-0.07856,0.64451)(-0.07816,0.64434)(-0.07776,0.64417)(-0.07735,0.64400)(-0.07695,0.64383)(-0.07655,0.64366)(-0.07615,0.64349)(-0.07575,0.64332)(-0.07535,0.64315)(-0.07495,0.64298)(-0.07455,0.64281)(-0.07415,0.64264)(-0.07375,0.64247)(-0.07335,0.64231)(-0.07295,0.64214)(-0.07255,0.64197)(-0.07214,0.64180)(-0.07174,0.64163)(-0.07134,0.64147)(-0.07094,0.64130)(-0.07054,0.64113)(-0.07014,0.64096)(-0.06974,0.64080)(-0.06934,0.64063)(-0.06894,0.64046)(-0.06854,0.64030)(-0.06814,0.64013)(-0.06774,0.63996)(-0.06733,0.63980)(-0.06693,0.63963)(-0.06653,0.63946)(-0.06613,0.63930)(-0.06573,0.63913)(-0.06533,0.63897)(-0.06493,0.63880)(-0.06453,0.63864)(-0.06413,0.63847)(-0.06373,0.63831)(-0.06333,0.63814)(-0.06293,0.63798)(-0.06253,0.63781)(-0.06212,0.63765)(-0.06172,0.63748)(-0.06132,0.63732)(-0.06092,0.63716)(-0.06052,0.63699)(-0.06012,0.63683)(-0.05972,0.63667)(-0.05932,0.63650)(-0.05892,0.63634)(-0.05852,0.63618)(-0.05812,0.63601)(-0.05772,0.63585)(-0.05731,0.63569)(-0.05691,0.63552)(-0.05651,0.63536)(-0.05611,0.63520)(-0.05571,0.63504)(-0.05531,0.63488)(-0.05491,0.63471)(-0.05451,0.63455)(-0.05411,0.63439)(-0.05371,0.63423)(-0.05331,0.63407)(-0.05291,0.63391)(-0.05251,0.63375)(-0.05210,0.63359)(-0.05170,0.63343)(-0.05130,0.63327)(-0.05090,0.63310)(-0.05050,0.63294)(-0.05010,0.63278)(-0.04970,0.63262)(-0.04930,0.63246)(-0.04890,0.63231)(-0.04850,0.63215)(-0.04810,0.63199)(-0.04770,0.63183)(-0.04729,0.63167)(-0.04689,0.63151)(-0.04649,0.63135)(-0.04609,0.63119)(-0.04569,0.63103)(-0.04529,0.63087)(-0.04489,0.63072)(-0.04449,0.63056)(-0.04409,0.63040)(-0.04369,0.63024)(-0.04329,0.63008)(-0.04289,0.62993)(-0.04248,0.62977)(-0.04208,0.62961)(-0.04168,0.62946)(-0.04128,0.62930)(-0.04088,0.62914)(-0.04048,0.62898)(-0.04008,0.62883)(-0.03968,0.62867)(-0.03928,0.62852)(-0.03888,0.62836)(-0.03848,0.62820)(-0.03808,0.62805)(-0.03768,0.62789)(-0.03727,0.62774)(-0.03687,0.62758)(-0.03647,0.62742)(-0.03607,0.62727)(-0.03567,0.62711)(-0.03527,0.62696)(-0.03487,0.62680)(-0.03447,0.62665)(-0.03407,0.62650)(-0.03367,0.62634)(-0.03327,0.62619)(-0.03287,0.62603)(-0.03246,0.62588)(-0.03206,0.62572)(-0.03166,0.62557)(-0.03126,0.62542)(-0.03086,0.62526)(-0.03046,0.62511)(-0.03006,0.62496)(-0.02966,0.62480)(-0.02926,0.62465)(-0.02886,0.62450)(-0.02846,0.62435)(-0.02806,0.62419)(-0.02766,0.62404)(-0.02725,0.62389)(-0.02685,0.62374)(-0.02645,0.62358)(-0.02605,0.62343)(-0.02565,0.62328)(-0.02525,0.62313)(-0.02485,0.62298)(-0.02445,0.62283)(-0.02405,0.62267)(-0.02365,0.62252)(-0.02325,0.62237)(-0.02285,0.62222)(-0.02244,0.62207)(-0.02204,0.62192)(-0.02164,0.62177)(-0.02124,0.62162)(-0.02084,0.62147)(-0.02044,0.62132)(-0.02004,0.62117)(-0.01964,0.62102)(-0.01924,0.62087)(-0.01884,0.62072)(-0.01844,0.62057)(-0.01804,0.62042)(-0.01764,0.62027)(-0.01723,0.62012)(-0.01683,0.61997)(-0.01643,0.61982)(-0.01603,0.61968)(-0.01563,0.61953)(-0.01523,0.61938)(-0.01483,0.61923)(-0.01443,0.61908)(-0.01403,0.61893)(-0.01363,0.61879)(-0.01323,0.61864)(-0.01283,0.61849)(-0.01242,0.61834)(-0.01202,0.61820)(-0.01162,0.61805)(-0.01122,0.61790)(-0.01082,0.61775)(-0.01042,0.61761)(-0.01002,0.61746)(-0.00962,0.61731)(-0.00922,0.61717)(-0.00882,0.61702)(-0.00842,0.61687)(-0.00802,0.61673)(-0.00762,0.61658)(-0.00721,0.61644)(-0.00681,0.61629)(-0.00641,0.61614)(-0.00601,0.61600)(-0.00561,0.61585)(-0.00521,0.61571)(-0.00481,0.61556)(-0.00441,0.61542)(-0.00401,0.61527)(-0.00361,0.61513)(-0.00321,0.61498)(-0.00281,0.61484)(-0.00240,0.61469)(-0.00200,0.61455)(-0.00160,0.61441)(-0.00120,0.61426)(-0.00080,0.61412)(-0.00040,0.61397)(0.00000,0.61383)(0.00000,0.61383)(0.00100,0.61616)(0.00200,0.61701)(0.00300,0.61762)(0.00401,0.61810)(0.00501,0.61850)(0.00601,0.61885)(0.00701,0.61915)(0.00801,0.61942)(0.00901,0.61966)(0.01001,0.61988)(0.01101,0.62007)(0.01202,0.62026)(0.01302,0.62042)(0.01402,0.62058)(0.01502,0.62072)(0.01602,0.62085)(0.01702,0.62097)(0.01802,0.62109)(0.01902,0.62119)(0.02003,0.62129)(0.02103,0.62138)(0.02203,0.62147)(0.02303,0.62155)(0.02403,0.62162)(0.02503,0.62169)(0.02603,0.62176)(0.02703,0.62182)(0.02804,0.62188)(0.02904,0.62193)(0.03004,0.62198)(0.03104,0.62203)(0.03204,0.62207)(0.03304,0.62211)(0.03404,0.62214)(0.03504,0.62218)(0.03605,0.62221)(0.03705,0.62224)(0.03805,0.62227)(0.03905,0.62229)(0.04005,0.62231)(0.04105,0.62233)(0.04205,0.62235)(0.04305,0.62237)(0.04406,0.62238)(0.04506,0.62239)(0.04606,0.62240)(0.04706,0.62241)(0.04806,0.62242)(0.04906,0.62243)(0.05006,0.62243)(0.05106,0.62243)(0.05207,0.62244)(0.05307,0.62244)(0.05407,0.62244)(0.05507,0.62244)(0.05607,0.62243)(0.05707,0.62243)(0.05807,0.62242)(0.05907,0.62242)(0.06008,0.62241)(0.06108,0.62240)(0.06208,0.62239)(0.06308,0.62238)(0.06408,0.62237)(0.06508,0.62236)(0.06608,0.62235)(0.06708,0.62234)(0.06809,0.62232)(0.06909,0.62231)(0.07009,0.62229)(0.07109,0.62228)(0.07209,0.62226)(0.07309,0.62224)(0.07409,0.62222)(0.07509,0.62221)(0.07610,0.62219)(0.07710,0.62217)(0.07810,0.62215)(0.07910,0.62213)(0.08010,0.62210)(0.08110,0.62208)(0.08210,0.62206)(0.08310,0.62204)(0.08411,0.62201)(0.08511,0.62199)(0.08611,0.62197)(0.08711,0.62194)(0.08811,0.62192)(0.08911,0.62189)(0.09011,0.62187)(0.09111,0.62184)(0.09212,0.62181)(0.09312,0.62179)(0.09412,0.62176)(0.09512,0.62173)(0.09612,0.62170)(0.09712,0.62168)(0.09812,0.62165)(0.09912,0.62162)(0.10013,0.62159)(0.10113,0.62156)(0.10213,0.62153)(0.10313,0.62150)(0.10413,0.62147)(0.10513,0.62144)(0.10613,0.62141)(0.10713,0.62138)(0.10814,0.62135)(0.10914,0.62132)(0.11014,0.62129)(0.11114,0.62126)(0.11214,0.62123)(0.11314,0.62119)(0.11414,0.62116)(0.11514,0.62113)(0.11615,0.62110)(0.11715,0.62107)(0.11815,0.62103)(0.11915,0.62100)(0.12015,0.62097)(0.12115,0.62093)(0.12215,0.62090)(0.12315,0.62087)(0.12416,0.62083)(0.12516,0.62080)(0.12616,0.62077)(0.12716,0.62073)(0.12816,0.62070)(0.12916,0.62067)(0.13016,0.62063)(0.13116,0.62060)(0.13217,0.62056)(0.13317,0.62053)(0.13417,0.62050)(0.13517,0.62046)(0.13617,0.62043)(0.13717,0.62039)(0.13817,0.62036)(0.13917,0.62032)(0.14018,0.62029)(0.14118,0.62025)(0.14218,0.62022)(0.14318,0.62019)(0.14418,0.62015)(0.14518,0.62012)(0.14618,0.62008)(0.14718,0.62005)(0.14819,0.62001)(0.14919,0.61998)(0.15019,0.61994)(0.15119,0.61991)(0.15219,0.61987)(0.15319,0.61984)(0.15419,0.61980)(0.15519,0.61977)(0.15620,0.61973)(0.15720,0.61970)(0.15820,0.61966)(0.15920,0.61963)(0.16020,0.61959)(0.16120,0.61956)(0.16220,0.61952)(0.16320,0.61949)(0.16421,0.61945)(0.16521,0.61942)(0.16621,0.61938)(0.16721,0.61935)(0.16821,0.61931)(0.16921,0.61928)(0.17021,0.61925)(0.17121,0.61921)(0.17222,0.61918)(0.17322,0.61914)(0.17422,0.61911)(0.17522,0.61907)(0.17622,0.61904)(0.17722,0.61900)(0.17822,0.61897)(0.17922,0.61894)(0.18023,0.61890)(0.18123,0.61887)(0.18223,0.61883)(0.18323,0.61880)(0.18423,0.61876)(0.18523,0.61873)(0.18623,0.61870)(0.18723,0.61866)(0.18824,0.61863)(0.18924,0.61860)(0.19024,0.61856)(0.19124,0.61853)(0.19224,0.61849)(0.19324,0.61846)(0.19424,0.61843)(0.19524,0.61839)(0.19625,0.61836)(0.19725,0.61833)(0.19825,0.61829)(0.19925,0.61826)(0.20025,0.61823)(0.20125,0.61820)(0.20225,0.61816)(0.20325,0.61813)(0.20426,0.61810)(0.20526,0.61807)(0.20626,0.61803)(0.20726,0.61800)(0.20826,0.61797)(0.20926,0.61794)(0.21026,0.61790)(0.21126,0.61787)(0.21227,0.61784)(0.21327,0.61781)(0.21427,0.61778)(0.21527,0.61774)(0.21627,0.61771)(0.21727,0.61768)(0.21827,0.61765)(0.21927,0.61762)(0.22028,0.61759)(0.22128,0.61756)(0.22228,0.61753)(0.22328,0.61749)(0.22428,0.61746)(0.22528,0.61743)(0.22628,0.61740)(0.22728,0.61737)(0.22829,0.61734)(0.22929,0.61731)(0.23029,0.61728)(0.23129,0.61725)(0.23229,0.61722)(0.23329,0.61719)(0.23429,0.61716)(0.23529,0.61713)(0.23630,0.61710)(0.23730,0.61707)(0.23830,0.61704)(0.23930,0.61701)(0.24030,0.61698)(0.24130,0.61696)(0.24230,0.61693)(0.24330,0.61690)(0.24431,0.61687)(0.24531,0.61684)(0.24631,0.61681)(0.24731,0.61678)(0.24831,0.61676)(0.24931,0.61673)(0.25031,0.61670)(0.25131,0.61667)(0.25232,0.61664)(0.25332,0.61662)(0.25432,0.61659)(0.25532,0.61656)(0.25632,0.61653)(0.25732,0.61651)(0.25832,0.61648)(0.25932,0.61645)(0.26033,0.61643)(0.26133,0.61640)(0.26233,0.61637)(0.26333,0.61635)(0.26433,0.61632)(0.26533,0.61629)(0.26633,0.61627)(0.26733,0.61624)(0.26834,0.61622)(0.26934,0.61619)(0.27034,0.61617)(0.27134,0.61614)(0.27234,0.61611)(0.27334,0.61609)(0.27434,0.61606)(0.27534,0.61604)(0.27635,0.61601)(0.27735,0.61599)(0.27835,0.61597)(0.27935,0.61594)(0.28035,0.61592)(0.28135,0.61589)(0.28235,0.61587)(0.28335,0.61585)(0.28436,0.61582)(0.28536,0.61580)(0.28636,0.61577)(0.28736,0.61575)(0.28836,0.61573)(0.28936,0.61570)(0.29036,0.61568)(0.29136,0.61566)(0.29237,0.61564)(0.29337,0.61561)(0.29437,0.61559)(0.29537,0.61557)(0.29637,0.61555)(0.29737,0.61552)(0.29837,0.61550)(0.29937,0.61548)(0.30038,0.61546)(0.30138,0.61544)(0.30238,0.61542)(0.30338,0.61540)(0.30438,0.61537)(0.30538,0.61535)(0.30638,0.61533)(0.30738,0.61531)(0.30839,0.61529)(0.30939,0.61527)(0.31039,0.61525)(0.31139,0.61523)(0.31239,0.61521)(0.31339,0.61519)(0.31439,0.61517)(0.31539,0.61515)(0.31640,0.61513)(0.31740,0.61511)(0.31840,0.61509)(0.31940,0.61507)(0.32040,0.61506)(0.32140,0.61504)(0.32240,0.61502)(0.32340,0.61500)(0.32441,0.61498)(0.32541,0.61496)(0.32641,0.61495)(0.32741,0.61493)(0.32841,0.61491)(0.32941,0.61489)(0.33041,0.61488)(0.33141,0.61486)(0.33242,0.61484)(0.33342,0.61482)(0.33442,0.61481)(0.33542,0.61479)(0.33642,0.61477)(0.33742,0.61476)(0.33842,0.61474)(0.33942,0.61472)(0.34043,0.61471)(0.34143,0.61469)(0.34243,0.61468)(0.34343,0.61466)(0.34443,0.61465)(0.34543,0.61463)(0.34643,0.61462)(0.34743,0.61460)(0.34844,0.61459)(0.34944,0.61457)(0.35044,0.61456)(0.35144,0.61454)(0.35244,0.61453)(0.35344,0.61451)(0.35444,0.61450)(0.35544,0.61449)(0.35645,0.61447)(0.35745,0.61446)(0.35845,0.61444)(0.35945,0.61443)(0.36045,0.61442)(0.36145,0.61441)(0.36245,0.61439)(0.36345,0.61438)(0.36446,0.61437)(0.36546,0.61435)(0.36646,0.61434)(0.36746,0.61433)(0.36846,0.61432)(0.36946,0.61431)(0.37046,0.61429)(0.37146,0.61428)(0.37247,0.61427)(0.37347,0.61426)(0.37447,0.61425)(0.37547,0.61424)(0.37647,0.61423)(0.37747,0.61422)(0.37847,0.61421)(0.37947,0.61420)(0.38048,0.61419)(0.38148,0.61418)(0.38248,0.61417)(0.38348,0.61416)(0.38448,0.61415)(0.38548,0.61414)(0.38648,0.61413)(0.38748,0.61412)(0.38849,0.61411)(0.38949,0.61410)(0.39049,0.61409)(0.39149,0.61408)(0.39249,0.61408)(0.39349,0.61407)(0.39449,0.61406)(0.39549,0.61405)(0.39650,0.61404)(0.39750,0.61404)(0.39850,0.61403)(0.39950,0.61402)(0.40050,0.61401)(0.40150,0.61401)(0.40250,0.61400)(0.40350,0.61399)(0.40451,0.61399)(0.40551,0.61398)(0.40651,0.61397)(0.40751,0.61397)(0.40851,0.61396)(0.40951,0.61395)(0.41051,0.61395)(0.41151,0.61394)(0.41252,0.61394)(0.41352,0.61393)(0.41452,0.61393)(0.41552,0.61392)(0.41652,0.61392)(0.41752,0.61391)(0.41852,0.61391)(0.41952,0.61390)(0.42053,0.61390)(0.42153,0.61390)(0.42253,0.61389)(0.42353,0.61389)(0.42453,0.61388)(0.42553,0.61388)(0.42653,0.61388)(0.42753,0.61387)(0.42854,0.61387)(0.42954,0.61387)(0.43054,0.61387)(0.43154,0.61386)(0.43254,0.61386)(0.43354,0.61386)(0.43454,0.61386)(0.43554,0.61385)(0.43655,0.61385)(0.43755,0.61385)(0.43855,0.61385)(0.43955,0.61385)(0.44055,0.61385)(0.44155,0.61385)(0.44255,0.61384)(0.44355,0.61384)(0.44456,0.61384)(0.44556,0.61384)(0.44656,0.61384)(0.44756,0.61384)(0.44856,0.61384)(0.44956,0.61384)(0.45056,0.61384)(0.45156,0.61384)(0.45257,0.61384)(0.45357,0.61384)(0.45457,0.61385)(0.45557,0.61385)(0.45657,0.61385)(0.45757,0.61385)(0.45857,0.61385)(0.45957,0.61385)(0.46058,0.61385)(0.46158,0.61386)(0.46258,0.61386)(0.46358,0.61386)(0.46458,0.61386)(0.46558,0.61387)(0.46658,0.61387)(0.46758,0.61387)(0.46859,0.61387)(0.46959,0.61388)(0.47059,0.61388)(0.47159,0.61388)(0.47259,0.61389)(0.47359,0.61389)(0.47459,0.61390)(0.47559,0.61390)(0.47660,0.61390)(0.47760,0.61391)(0.47860,0.61391)(0.47960,0.61392)(0.48060,0.61392)(0.48160,0.61393)(0.48260,0.61393)(0.48360,0.61394)(0.48461,0.61394)(0.48561,0.61395)(0.48661,0.61396)(0.48761,0.61396)(0.48861,0.61397)(0.48961,0.61397)(0.49061,0.61398)(0.49161,0.61399)(0.49262,0.61399)(0.49362,0.61400)(0.49462,0.61401)(0.49562,0.61402)(0.49662,0.61402)(0.49762,0.61403)(0.49862,0.61404)(0.49962,0.61405)(0.50063,0.61405)(0.50163,0.61406)(0.50263,0.61407)(0.50363,0.61408)(0.50463,0.61409)(0.50563,0.61410)(0.50663,0.61410)(0.50763,0.61411)(0.50864,0.61412)(0.50964,0.61413)(0.51064,0.61414)(0.51164,0.61415)(0.51264,0.61416)(0.51364,0.61417)(0.51464,0.61418)(0.51564,0.61419)(0.51665,0.61420)(0.51765,0.61421)(0.51865,0.61422)(0.51965,0.61424)(0.52065,0.61425)(0.52165,0.61426)(0.52265,0.61427)(0.52365,0.61428)(0.52466,0.61429)(0.52566,0.61430)(0.52666,0.61432)(0.52766,0.61433)(0.52866,0.61434)(0.52966,0.61435)(0.53066,0.61437)(0.53166,0.61438)(0.53267,0.61439)(0.53367,0.61441)(0.53467,0.61442)(0.53567,0.61443)(0.53667,0.61445)(0.53767,0.61446)(0.53867,0.61447)(0.53967,0.61449)(0.54068,0.61450)(0.54168,0.61452)(0.54268,0.61453)(0.54368,0.61455)(0.54468,0.61456)(0.54568,0.61458)(0.54668,0.61459)(0.54768,0.61461)(0.54869,0.61462)(0.54969,0.61464)(0.55069,0.61465)(0.55169,0.61467)(0.55269,0.61469)(0.55369,0.61470)(0.55469,0.61472)(0.55569,0.61474)(0.55670,0.61475)(0.55770,0.61477)(0.55870,0.61479)(0.55970,0.61480)(0.56070,0.61482)(0.56170,0.61484)(0.56270,0.61486)(0.56370,0.61488)(0.56471,0.61489)(0.56571,0.61491)(0.56671,0.61493)(0.56771,0.61495)(0.56871,0.61497)(0.56971,0.61499)(0.57071,0.61501)(0.57171,0.61502)(0.57272,0.61504)(0.57372,0.61506)(0.57472,0.61508)(0.57572,0.61510)(0.57672,0.61512)(0.57772,0.61514)(0.57872,0.61516)(0.57972,0.61518)(0.58073,0.61521)(0.58173,0.61523)(0.58273,0.61525)(0.58373,0.61527)(0.58473,0.61529)(0.58573,0.61531)(0.58673,0.61533)(0.58773,0.61536)(0.58874,0.61538)(0.58974,0.61540)(0.59074,0.61542)(0.59174,0.61544)(0.59274,0.61547)(0.59374,0.61549)(0.59474,0.61551)(0.59574,0.61554)(0.59675,0.61556)(0.59775,0.61558)(0.59875,0.61561)(0.59975,0.61563)(0.60075,0.61565)(0.60175,0.61568)(0.60275,0.61570)(0.60375,0.61573)(0.60476,0.61575)(0.60576,0.61578)(0.60676,0.61580)(0.60776,0.61583)(0.60876,0.61585)(0.60976,0.61588)(0.61076,0.61590)(0.61176,0.61593)(0.61277,0.61596)(0.61377,0.61598)(0.61477,0.61601)(0.61577,0.61604)(0.61677,0.61606)(0.61777,0.61609)(0.61877,0.61612)(0.61977,0.61614)(0.62078,0.61617)(0.62178,0.61620)(0.62278,0.61623)(0.62378,0.61625)(0.62478,0.61628)(0.62578,0.61631)(0.62678,0.61634)(0.62778,0.61637)(0.62879,0.61640)(0.62979,0.61642)(0.63079,0.61645)(0.63179,0.61648)(0.63279,0.61651)(0.63379,0.61654)(0.63479,0.61657)(0.63579,0.61660)(0.63680,0.61663)(0.63780,0.61666)(0.63880,0.61669)(0.63980,0.61672)(0.64080,0.61675)(0.64180,0.61678)(0.64280,0.61682)(0.64380,0.61685)(0.64481,0.61688)(0.64581,0.61691)(0.64681,0.61694)(0.64781,0.61697)(0.64881,0.61701)(0.64981,0.61704)(0.65081,0.61707)(0.65181,0.61710)(0.65282,0.61714)(0.65382,0.61717)(0.65482,0.61720)(0.65582,0.61724)(0.65682,0.61727)(0.65782,0.61730)(0.65882,0.61734)(0.65982,0.61737)(0.66083,0.61741)(0.66183,0.61744)(0.66283,0.61748)(0.66383,0.61751)(0.66483,0.61755)(0.66583,0.61758)(0.66683,0.61762)(0.66783,0.61765)(0.66884,0.61769)(0.66984,0.61772)(0.67084,0.61776)(0.67184,0.61780)(0.67284,0.61783)(0.67384,0.61787)(0.67484,0.61791)(0.67584,0.61794)(0.67685,0.61798)(0.67785,0.61802)(0.67885,0.61806)(0.67985,0.61809)(0.68085,0.61813)(0.68185,0.61817)(0.68285,0.61821)(0.68385,0.61825)(0.68486,0.61829)(0.68586,0.61832)(0.68686,0.61836)(0.68786,0.61840)(0.68886,0.61844)(0.68986,0.61848)(0.69086,0.61852)(0.69186,0.61856)(0.69287,0.61860)(0.69387,0.61864)(0.69487,0.61868)(0.69587,0.61872)(0.69687,0.61876)(0.69787,0.61880)(0.69887,0.61885)(0.69987,0.61889)(0.70088,0.61893)(0.70188,0.61897)(0.70288,0.61901)(0.70388,0.61906)(0.70488,0.61910)(0.70588,0.61914)(0.70688,0.61918)(0.70788,0.61923)(0.70889,0.61927)(0.70989,0.61931)(0.71089,0.61936)(0.71189,0.61940)(0.71289,0.61944)(0.71389,0.61949)(0.71489,0.61953)(0.71589,0.61958)(0.71690,0.61962)(0.71790,0.61967)(0.71890,0.61971)(0.71990,0.61976)(0.72090,0.61980)(0.72190,0.61985)(0.72290,0.61989)(0.72390,0.61994)(0.72491,0.61999)(0.72591,0.62003)(0.72691,0.62008)(0.72791,0.62013)(0.72891,0.62017)(0.72991,0.62022)(0.73091,0.62027)(0.73191,0.62032)(0.73292,0.62037)(0.73392,0.62041)(0.73492,0.62046)(0.73592,0.62051)(0.73692,0.62056)(0.73792,0.62061)(0.73892,0.62066)(0.73992,0.62071)(0.74093,0.62076)(0.74193,0.62081)(0.74293,0.62086)(0.74393,0.62091)(0.74493,0.62096)(0.74593,0.62101)(0.74693,0.62106)(0.74793,0.62111)(0.74894,0.62116)(0.74994,0.62121)(0.75094,0.62126)(0.75194,0.62132)(0.75294,0.62137)(0.75394,0.62142)(0.75494,0.62147)(0.75594,0.62152)(0.75695,0.62158)(0.75795,0.62163)(0.75895,0.62168)(0.75995,0.62174)(0.76095,0.62179)(0.76195,0.62185)(0.76295,0.62190)(0.76395,0.62195)(0.76496,0.62201)(0.76596,0.62206)(0.76696,0.62212)(0.76796,0.62217)(0.76896,0.62223)(0.76996,0.62229)(0.77096,0.62234)(0.77196,0.62240)(0.77297,0.62245)(0.77397,0.62251)(0.77497,0.62257)(0.77597,0.62263)(0.77697,0.62268)(0.77797,0.62274)(0.77897,0.62280)(0.77997,0.62286)(0.78098,0.62291)(0.78198,0.62297)(0.78298,0.62303)(0.78398,0.62309)(0.78498,0.62315)(0.78598,0.62321)(0.78698,0.62327)(0.78798,0.62333)(0.78899,0.62339)(0.78999,0.62345)(0.79099,0.62351)(0.79199,0.62357)(0.79299,0.62363)(0.79399,0.62369)(0.79499,0.62375)(0.79599,0.62382)(0.79700,0.62388)(0.79800,0.62394)(0.79900,0.62400)(0.80000,0.62407)
    };
\end{axis}
\end{tikzpicture}
\caption{Plot of $\zeta_{\Theta}(x,14/15,p^*)$ for $x \in [-0.1,0.8]$, $\Theta(x) = \optTheta{}$ and $p^*$ as in \eqref{eqn_bestBoundTriangleFree}. We have $\min_{x \in (-1,1]} \zeta_{\Theta}(x,14/15,p^*) = \zeta_{\Theta}(0,14/15,p^*) = \unroundedApproxOpt{}$, marked by the dot.  }
\label{fig:zetaFunction}
\end{figure}

The approximation ratio in \eqref{eqn_bestBoundTriangleFree} is greater than $\alpha(1)$ (given in \Cref{table_alphaMuVals}), which is a lower bound on the approximation ratio of \Cref{alg:approxLee} if $(h^+_e)_{e \in E}$ is contained in the matching polytope, see \Cref{lemma_approxRatioLee}. The optimal approximation ratio of \Cref{alg:triangleFree}, with respect to $\Theta \in \mathcal{A}$, is given by $\max_{p \in [0,1], \, \Theta \in \mathcal{A}} \zeta^*_{\Theta}(\mu,p)$. Equation \eqref{eqn_bestBoundTriangleFree} proves that $\max_{p \in [0,1], \, \Theta \in \mathcal{A}} \zeta^*_{\Theta}(\mu,p) \geq \unroundedApproxOpt{}$. Let us now provide an upper bound on this maximum.
\begin{theorem}
    \label{lemma_UB_triangleFree}
    The best provable lower bound on the approximation ratio of \Cref{alg:triangleFree}, given by $\max_{p \in [0,1], \, \Theta \in \mathcal{A}} \zeta^*_{\Theta}(14/15,p)$, satisfies
    \begin{align}
        \label{eqn_ub_triangleFreeNoConj}
        \unroundedApproxOpt{} \leq \max_{p \in [0,1], \, \Theta \in \mathcal{A}} \zeta^*_{\Theta}(14/15,p) < \ratioUB{},
    \end{align}
    for $\mathcal{A}$ as in \eqref{eqn_thetaProp} and $\zeta_{\Theta}^*$ as in \eqref{eqn_zetaStarDef}. 
\end{theorem}
\begin{proof}
Let $\mu = 14/15$. The lower bound in \eqref{eqn_ub_triangleFreeNoConj} is due to \eqref{eqn_bestBoundTriangleFree}. For the upper bound, observe that by the definition of $\zeta^*_{\Theta}(\mu,p)$, see \eqref{eqn_zetaStarDef}, we have that
\begin{align}
    \label{eqn_firstIneqZeta}
    \zeta^*_{\Theta}(\mu,p) \leq \min\left\{  \zeta_{\Theta}\left(0,\mu,p\right), \zeta_{\Theta}\left(1/2,\mu,p \right)\right\}.
\end{align}
Since $\max_{p \in [0,1], \, \Theta \in \mathcal{A}} \zeta^*_{\Theta}(\mu,p) \geq \unroundedApproxOpt{}$, we may consider only $\Theta \in \mathcal{A}$ and $p \in [0,1]$ satisfying $\zeta_{\Theta}\left(0,\mu,p\right) \geq \unroundedApproxOpt{}$. Solving  $\zeta_{\Theta}(0,\mu,p) = \left(p q(0)\left(1 - \Theta(1)/2\right)+1-p \right)/2$ for $\Theta(1)$ yields
\begin{align}
    \label{eqn_thetaOneRewrite}
    \Theta(1) =  2 - \frac{2}{q(0)}- \frac{4 \, \zeta_{\Theta}(0,\mu,p)-2}{p q(0)}.
\end{align}
Since $\Theta(1) \geq 0$, \eqref{eqn_thetaOneRewrite} implies that
\begin{align}
    \label{eqn_pIneqZ}
    p \geq \frac{2 \, \zeta_{\Theta}(0,\mu,p)  -1}{q(0)-1}.
\end{align}
We have $\zeta_{\Theta}\left(1/2,\mu,p \right) = f(\Theta(1/2),p)$, where 
\begin{align}
    \label{eqn_fpIneq}
    f(z,p) := \frac{1}{3} \left( p q(1/2) \left(1-\frac{z}{2}+C_1\sqrt{z\left(1-z\right)}\right) + (1-p)C_2 \right),
\end{align}
and $C_1 := (4+\mu (\pi-2))/ (2\pi)$, $C_2 := \left(1+3\mu/2 \right)$.

By the properties of functions in $\mathcal{A}$, see \eqref{eqn_thetaProp}, we have
\begin{align}
    \label{eqn_argumentIneqChain}
     1-\Theta(1) =  \min_{x \in [0,1]} \left( 1-\Theta(x) \right) \left( 1-\Theta(1-x) \right) \leq (1-\Theta(1/2))^2,
\end{align}
from where it follows  that $\Theta(1/2) \leq 1-\sqrt{1-\Theta(1)}$ (the expression $\sqrt{1-\Theta(1)}$ is well-defined since $\Theta(1) \leq 1$). By substituting \eqref{eqn_thetaOneRewrite} for $\Theta(1)$ in this inequality, we find that $\Theta(1/2) \leq 1-\sqrt{1-\Theta(1)}$ is equivalent to $\Theta(1/2) \leq h(\zeta_{\Theta}(0,\mu,p),p)$, where
\begin{align}
    \label{eqn_hFunction}
    h(z,p) := 1 - \sqrt{\frac{4 z-2}{pq(0)} + \frac{2}{q(0)}- 1}.
\end{align}
For future reference, we observe the following two properties of the function $h$:
\begin{align}
    \label{eqn_functionH_properties}
    \frac{\partial h(z,p)}{\partial z} < 0 \, \text{ and } \, \frac{\partial h(z',p)}{\partial p} > 0 \, \text{ if } \, z' > \frac{1}{2}.
\end{align}
We claim that 
\begin{align}
    \label{eqn_fIncreaseClaim2}
    f(\Theta(1/2),p) \leq f\left(h(\zeta_{\Theta}(0,\mu,p),p), p\right).
\end{align}
As $\Theta(1/2) \leq h(\zeta_{\Theta}(0,\mu,p),p)$, claim \eqref{eqn_fIncreaseClaim2} follows by showing that $f(z,p)$ is increasing in $z$ on the interval $0 \leq z \leq h(\zeta_{\Theta}(0,\mu,p),p)$ and $p \in [0,1]$. Indeed, one can verify that 
\begin{align}
    \label{eqn_fDerivative}
    \frac{\partial f(z,p)}{\partial z} =  \frac{p q(1/2)}{6} \left( \frac{C_1(1-2z)}{\sqrt{z(1-z)}} - 1\right) \geq 0 \text{ for all } z \in \left(0, \frac{1}{2}\left(1- \frac{1}{\sqrt{4C_1^2+1}}\right)\right]
\end{align}
and $p \in [0,1]$ (the upper bound on $z$ can be found by solving $\partial f(z,p) / \partial z = 0$). To prove \eqref{eqn_fIncreaseClaim2}, it remains to show that $h(\zeta_{\Theta}(0,\mu,p),p) \leq \frac{1}{2}\left(1- \frac{1}{\sqrt{4C_1^2+1}}\right)$. Recall that $\zeta_{\Theta}(0,\mu,p) \geq \unroundedApproxOpt{}$, which we use to obtain
\begin{align}
  h(\zeta_{\Theta}(0,\mu,p),p) \leq h(\unroundedApproxOpt,p) \leq h(\unroundedApproxOpt,1) \leq \frac{1}{2}\left(1- \frac{1}{\sqrt{4C_1^2+1}}\right),
\end{align}
where the first and second inequality follow from \eqref{eqn_functionH_properties}. The last inequality follows by computing the two numbers. Hence, claim \eqref{eqn_fIncreaseClaim2} follows.

Combining \eqref{eqn_firstIneqZeta}, $\zeta_{\Theta}\left(1/2,\mu,p \right) = f(\Theta(1/2),p)$ and \eqref{eqn_fIncreaseClaim2} yields
\begin{align}
    \zeta^*_{\Theta}(\mu,p) &\leq \min\left\{ \zeta_{\Theta}(0,\mu,p),  f\left(h(\zeta_{\Theta}(0,\mu,p),p), p\right) \right\} \nonumber \\
    &\leq \max_{z \in[\unroundedApproxOpt{},1]} \min\left\{ z, \max_{p \in \left[ \frac{2 \, z  -1}{q(0)-1},1 \right]} f(h(z,p),p) \right\}. \label{eqn_upperBoundRatioZ}
\end{align}
In \eqref{eqn_upperBoundRatioZ}, the constraint on $z$ is due to our  assumption $\unroundedApproxOpt{} \leq \zeta_{\Theta}(0,\mu,p) = z$, whereas the constraint on $p$ is due to \eqref{eqn_pIneqZ}. We define $r := \ratioUB{}$, and show that the above expression is upper bounded by $r$. To solve the maximization problem in $z$ in \eqref{eqn_upperBoundRatioZ}, we first note that when $z < r$ we have $\zeta^*_{\Theta}(\mu,p) \leq z < r$. Thus, we may restrict to $z \in [r,1]$. We now show that
$\max_{z \in [r,1], \, p \in \left[ \frac{2 z  -1}{q(0)-1},1 \right]} f(h(z,p),p) < r$. To simplify this problem, we claim that
\begin{align}
    \label{eqn_claim2}
    f(h(z,p),p) \leq f(h(r,p),p)
\end{align}
for all $z \in [r,1]$ and $p \in \left[ 0,1 \right]$. We showed in \eqref{eqn_fDerivative} that $\partial f(z,p) / \partial z \geq 0$ for all $0 < z \leq \frac{1}{2}\left(1- \frac{1}{\sqrt{4C_1^2+1}}\right)$, so that claim \eqref{eqn_claim2} follows by showing that $0 \leq h(z,p) \leq h(r,p) \leq  \frac{1}{2}\left(1- \frac{1}{\sqrt{4C_1^2+1}}\right)$. To prove these inequalities, we use $z \geq r > 1/2$ and \eqref{eqn_functionH_properties} to obtain $h(z,p) \leq h(r,p) \leq  h(r,1) \leq  \frac{1}{2}\left(1- \frac{1}{\sqrt{4C_1^2+1}}\right)$. Here, the last inequality can be verified by computing the two numbers. Hence, claim \eqref{eqn_claim2} is proven.

Note that the interval $\left[ \frac{2 z  -1}{q(0)-1},1 \right]$  with $z \in [r,1]$, is largest for $z = r$. Combining this observation with \eqref{eqn_claim2} yields that
\begin{align}
    \label{eqn_fzandRinOneIneq}
    \max_{z \in [r,1], \, p \in \left[ \frac{2 z  -1}{q(0)-1},1 \right]} f(h(z,p),p) \leq \max_{ p \in \left[ \frac{2 r  -1}{q(0)-1},1 \right]} f(h(r,p),p).
\end{align}
\Cref{lemma_concaveMax} in \Cref{section_extraResults} proves that $\max_{ p \in \left[ \frac{2 r  -1}{q(0)-1},1 \right]} f(h(r,p),p) < r$. Combined with  \eqref{eqn_upperBoundRatioZ} and \eqref{eqn_fzandRinOneIneq}, we find that $\zeta^*_{\Theta}(\mu,p) < r = \ratioUB{}$, which finishes the proof.
\end{proof}
\section{New approximation algorithm on bipartite graphs}
\label{section_approxBipart}
In this section we consider a QMC approximation algorithm for bipartite graphs that achieves an approximation ratio of $\ratioBipart{}$. The QMC problem, restricted to bipartite graphs, belongs to the complexity class \textsf{StoqMA} 
\cite[Obs.~30]{CM16} and remains \textit{'notoriously difficult to solve'} \cite[Sect.~4]{GP19}. It is known that on bipartite graphs, the QMC problem is equivalent to the \textit{Einstein, Podolsky, Rosen} (EPR) problem, which is another 2-local Hamiltonian problem introduced in \cite[Problem 2]{king2023improved} (specifically, the Hamiltonian corresponding to the EPR problem is unitarily similar to $H_G$). In the same paper, King provides a classical approximation algorithm for the EPR problem that achieves an approximation ratio of $1/\sqrt{2}$ $(\approx 0.707)$.
The algorithm we consider is the following:
\begin{algorithm}[QMC approximation algorithm for bipartite graphs]
\label{alg:bipartiteG2}
Input: bipartite graph $G = (V = V^0 \cup V^1,E, w)$, function $\Theta \in \mathcal{A}$, see \eqref{eqn_thetaProp}, real numbers $h_\mathrm{max} \in \left[ \sqrt{3}/2,1 \right]$ and $\theta^* \in [0,1]$.
\begin{enumerate}
    \item \label{step_algBipartState} For all $i \in V$, set $\ket{z_i} = \ket{0}$ if $i \in V^0$, and $\ket{z_i} = \ket{1}$ if $i \in V^1$.
    \item \label{eqn_thetaHmaxSet} Solve \ref{eqn_sdpRelax} for $k = 2$ to obtain the values $\left(h_e \right)_{e \in E}$, see \Cref{def_sdpValues}, and for all $e \in E$, set \begin{align}
        \label{eqn_thetaBipartSet}
        \theta_{e} := \begin{cases}
            \arcsin{ \sqrt{\Theta\mleft(h^+_{e}\mright) } } &\text{ if } h_e \leq h_{\text{max}} \\
            \arcsin{\sqrt{\theta^*}}&\text{ else.}
        \end{cases}
    \end{align}
    \item Output the state $\ket{\phi}\bra{\phi}$, where
    \begin{align}
        \label{eqn_outputStateBipart}
        \ket{\phi} := \prod_{\edge{i}{j} \in E \, : \, i \in V^0, j \in V^1} \exp{\left( \frac{\imagUnit}{2} \theta_{ij} Y_i X_j \right)} \bigotimes_{i \in V }\ket{z_i}.
    \end{align}
\end{enumerate}
\end{algorithm}
\Cref{alg:bipartiteG2} is inspired by \cite[Alg.~1]{Lee22}. Lee's algorithm is suited for general graphs, achieves an approximation ratio of 0.562, and differs from \Cref{alg:bipartiteG2} in \Cref{step_algBipartState,eqn_thetaHmaxSet}. In \Cref{step_algBipartState}, both algorithms determine the value of $\ket{z_i} \in \left\{ \ket{0}, \ket{1} \right\}$ based on a partition of the vertex set~$V$. For \Cref{alg:bipartiteG2}, this partition is already given as $\left\{V^0,V^1 \right\}$, i.e., the bipartition of $G$. In contrast, \cite[Alg.~1]{Lee22} partitions $V$ as $\left\{ \widetilde{V}, V \setminus \widetilde{V} \right\}$, where $\widetilde{V}$ is constructed as follows: pick a variable $u \in \{x,y,z\}$ uniformly at random. Consider the vectors $\left(\mathbf{v}(u_i) \right)_{i \in V}$, where $\mathbf{v}(u_i)$ is as in \eqref{eqn_gramVectors} (for example, if $u = x$, then $\mathbf{v}(u_i) = \mathbf{v}(x_i)$). Draw a vector $r$ uniformly from the unit sphere of appropriate dimension. Let $\widetilde{V} := \setFunct{i \in V}{ \signFunction{ \mathbf{v}(u_i)^\top r} = 1}$. The random partition $\left\{ \widetilde{V}, V \setminus \widetilde{V} \right\}$ is not guaranteed to recover the bipartition of a bipartite graph. That is, when given a bipartite graph as input, it is not guaranteed that the vertices in $\widetilde{V}$ are pairwise non-adjacent. In \Cref{eqn_thetaHmaxSet}, both algorithms set $\theta_e$ as a function of $h^+_e$. Lee chose $\Theta(x) = 1-e^{-Rx} \in \mathcal{A}$ (see \Cref{lemma_linFunctionInOmega}) and (numerically) optimized the value of $R \geq 0$  to obtain the highest approximation ratio. In contrast, we choose $\Theta(x) = Rx$ in \Cref{thm_LBbipartAlg}, which is also contained in $\mathcal{A}$ by \Cref{lemma_linFunctionInOmega}. Additionally, in \Cref{thm_upperBoundRStar} we provide an upper bound on the approximation ratio attained by the best possible $\Theta \in \mathcal{A}$. Lastly, Lee's algorithm does not use $h_\mathrm{max}$, see \eqref{eqn_thetaBipartSet}, which allows for adjusting the value of $\theta_e$ when $h_e$ is sufficiently large. These differences in \Cref{step_algBipartState,eqn_thetaHmaxSet} ensure that \Cref{alg:bipartiteG2} achieves a higher approximation ratio than \cite[Alg.~1]{Lee22} on bipartite graphs. To compute this approximation ratio in terms of $\Theta$, $h_\mathrm{max}$ and $\theta^*$, we require some preparatory results and definitions. We present first \cite[Lem.~11]{Lee22}. 
\begin{lemma}[\cite{Lee22}]
    \label{lemma_nrgCutEdges}
    Let $\ket{\phi}$ be as in \eqref{eqn_outputStateBipart} and $\ket{z_i}$, $\ket{z_j}$ as in \Cref{step_algBipartState} of \Cref{alg:bipartiteG2}. For any $i,j \in V$, $\bra{\phi} H_{ij} \ket{\phi} \geq 0$. If $\ket{z_i} \neq \ket{z_j}$, then, for $A_{ij}$ and $B_{ij}$ as in \eqref{eqn_edgeTermProds},
    \begin{align}
    \label{eqn_strongBoundEdgeBip}
        \bra{\phi} H_{ij} \ket{\phi} \geq  1 +  A_{ij} B_{ij} + \sin{\mleft(\theta_{ij}\mright)} \left( A_{ij} + B_{ij} \right).
    \end{align}
\end{lemma}
The small discrepancy between our \Cref{lemma_nrgCutEdges} and \cite[Lem.~11]{Lee22} is due to the different scaling of $\theta$. If $\ket{z_i} \neq \ket{z_j}$ for some edge $\{i,j\}$, we say that edge $\{i,j\}$ is cut. If $\{i,j\}$ is a cut edge, $\bra{\phi} H_{ij} \ket{\phi}$ satisfies \eqref{eqn_strongBoundEdgeBip}, which is a stronger lower bound than $\bra{\phi} H_{ij} \ket{\phi} \geq 0$. Observe that \Cref{step_algBipartState} of \Cref{alg:bipartiteG2} ensures that all edges of the bipartite input graph are cut. 

To compute the approximation ratio of \Cref{alg:bipartiteG2}, we also require the  function
\begin{align}
    \label{eqn_deltaFunc}
    \delta_\Theta(x) := \frac{2 - \Theta(1-x^+) + 2 \sqrt{\Theta(x^+)(1-\Theta(1-x^+))}}{2+2x}.
\end{align}
For future reference, observe  that
\begin{align}\label{eqn_simpleThetaObs}
    \delta_\Theta(x) = \frac{2-\Theta(1)}{2+2x} \geq \frac{2-\Theta(1)}{2} = \delta_\Theta(0) \quad \forall x \in (-1,0],
\end{align}
which follows from the fact that $x^+ = \max\{x,0\} = 0$ for all $x \in (-1,0]$, and $\Theta(1) \leq 1$, see \eqref{eqn_thetaProp}.
\begin{lemma}
    \label{lemma_boundHsimplify}
    Let $G$ be a bipartite graph used as input to \Cref{alg:bipartiteG2}, $\ket{\phi}$ be as in \eqref{eqn_outputStateBipart}, $\Theta \in \mathcal{A}$ and $h_\mathrm{max} \in \left[ \sqrt{3}/2, 1 \right]$. Consider the values $\theta_e$ as in \eqref{eqn_thetaBipartSet}. Suppose $\theta_e = \arcsin{\sqrt{ \Theta \mleft( h^+_e\mright)}}$ for some edge $e \in E(G)$ (i.e., $h_e \leq h_\mathrm{max}$), and $h_e > -1$. If for all edges $e'$ adjacent to $e$, $\theta_{e'} = \arcsin{\sqrt{ \Theta \mleft( h^+_{e'}\mright)}}$, then we have
    \begin{align}
        \label{eqn_valueDivideResult}
        \frac{\bra{\phi} H_e \ket{\phi}}{2+2h_e} \geq \delta_{\Theta}(h_e).
    \end{align}
\end{lemma}
\begin{proof}
We substitute $\theta_e = \arcsin{ \sqrt{\Theta \mleft(h^+_e \mright)}}$ in \eqref{eqn_strongBoundEdgeBip} (with $e = \{i,j\}$), use that $\sin \theta_e =  \sqrt{\Theta \mleft(h^+_e \mright)}$, and apply \eqref{eqn_boundProdTheta2}, to obtain
    \begin{equation}
    \label{eqn_rewriteHEq}
    \begin{aligned}
         \bra{\phi} H_{e} \ket{\phi} &\geq 1 + \left(1-\Theta \mleft( 1-h^+_e \mright) \right) + 2 \sqrt{\Theta \mleft(h^+_e \mright)} \sqrt{1-\Theta \mleft( 1-h^+_e \mright)} \\
         &=2 - \Theta \mleft( 1-h^+_e \mright) + 2 \sqrt{\Theta \mleft(h^+_e \mright)\left(1-\Theta \mleft( 1-h^+_e \mright)\right)} = \delta_{\Theta}\mleft( h_e \mright) \left(2 +2 h_e \right).
    \end{aligned}
    \end{equation}
    Since $h_e > -1$, \eqref{eqn_rewriteHEq} is equivalent to \eqref{eqn_valueDivideResult}.
\end{proof}
Lastly, we require the following refinement of monogamy of entanglement, see \eqref{eqn_monogamyStar}.
\begin{lemma}
    \label{lemma_quadHBound}
    Let $e$ and $e'$ be adjacent edges, and let $L \in \feasLass{k}{n}$, see \eqref{eqn_feasLass}, with $k \geq 2$ and $n \geq 3$. Consider the $h$ values as in \Cref{def_sdpValues}, corresponding to $L$. We have
    \begin{align}
        \label{eqn_hSolutionUB}
        h_{e'} \leq \frac{1}{2}\left(\sqrt{3\left(1-h^2_e \right)}-h_e\right).
    \end{align}
  In particular,
    \begin{align}
        \label{eqn_extraCondition}
        h_e \geq \frac{\sqrt{3}}{2} \implies h_{e'} \leq 0.
    \end{align}
\end{lemma} 
\begin{proof}
    \cite[Lem.~3]{lee2024improved} shows that $3 \left( h_{e'}+h_{e} \right)^2+\left( h_{e'}-h_{e} \right)^2 \leq 3$,
    whenever $h_{e'}+h_{e} \geq -1/2$ (recall that the variables in \cite{lee2024improved} are scaled differently compared to this work). Solving this inequality for $h_{e'}$ yields \eqref{eqn_hSolutionUB}. The implication \eqref{eqn_extraCondition} follows directly from \eqref{eqn_hSolutionUB}.
\end{proof}
We are now ready to compute (a lower bound on) the approximation ratio of \Cref{alg:bipartiteG2}.
\begin{theorem}  \label{thm_improvedBipartRatio}
    The approximation ratio of \Cref{alg:bipartiteG2} for bipartite graphs, with parameters $\Theta \in \mathcal{A}$, $h_\mathrm{max} \in \left[ \sqrt{3}/2,1 \right]$ and $\theta^* \in [0,1]$, is at least
    \begin{align}
        \label{eqn_ratioExprBig}
        \min\left\{ \frac{2-\theta^*}{2+\sqrt{3\left( 1-h_\mathrm{max}^2 \right)} - h_\mathrm{max}},  \min_{x \in \left[ 0, h_\mathrm{max} \right]} \delta_{\Theta}(x), \frac{1+\sqrt{\theta^*}}{2} \right\}.
    \end{align}
\end{theorem}
\begin{proof}
Let $E$ be the edge set of the bipartite input graph. Let $(h_e)_{e \in E}$ be the values obtained in \Cref{eqn_thetaHmaxSet} of \Cref{alg:bipartiteG2}, and $\ket{\phi}$ be as in \eqref{eqn_outputStateBipart}. By the definition of $H_G$, see \eqref{eqn_hamiltonDef}, we have that
\begin{equation}
    \label{eqn_lbHGBipart}
\begin{aligned}
    \bra{\phi} H_G \ket{\phi} &= \sum_{e \in E} w_e \bra{\phi} H_e \ket{\phi} \geq \sum_{e \in E : h_e > -1}  w_e \, \frac{\bra{\phi} H_e \ket{\phi}}{ 2+2h_e  } (2+2h_e)  \\
    &\geq \left( \inf_{e \in E: h_e > -1} \frac{\bra{\phi} H_e \ket{\phi}}{ 2+2h_e  } \right) \sum_{e \in E} w_e \, (2+2h_e) \geq \left( \inf_{e \in E: h_e > -1} \frac{\bra{\phi} H_e \ket{\phi}}{ 2+2h_e  } \right) \lambdaMax{H_G}.
\end{aligned}
\end{equation}
For the first inequality, we have used that $w_e \bra{\phi} H_e \ket{\phi} \geq 0$, since $H_e = (1/4) H_e^2 \succeq 0$ and $w_e > 0$. The second inequality follows from the fact that $\sum_{e \in E : h_e > -1} w_e \, (2+2h_e) = \sum_{e \in E} w_e \, (2+2h_e)$. The last inequality follows from the fact that \ref{eqn_sdpRelax} provides an upper bound on $\lambdaMax{H_G}$, as explained in \Cref{section_SDP_defs} (see also \eqref{eqn_sdpUBtoHMax}). 

Considering \eqref{eqn_lbHGBipart}, a lower bound on the approximation ratio of \Cref{alg:bipartiteG2} is given by
    \begin{align}
        \label{eqn_hMinimization}
        \inf_{e \in E: h_e > -1} \frac{\bra{\phi} H_e \ket{\phi}}{ 2+2h_e  }.
    \end{align}
    Note here that $\bra{\phi} H_e \ket{\phi}$ is a function of $h_e$, which follows from the definition of $\ket{\phi}$, see \eqref{eqn_outputStateBipart}. We show later that $\lim_{h_e \downarrow -1} \frac{\bra{\phi} H_e \ket{\phi}}{ 2+2h_e  } = + \infty$, which implies that \eqref{eqn_hMinimization} is well-defined. Let us determine a lower bound on \eqref{eqn_hMinimization} by distinguishing three cases, based on the value of $h_e \in (-1,1]$.
    \begin{mycases}
        \item $h_e < \frac{1}{2}\left( \sqrt{3\left( 1-h_\mathrm{max}^2 \right)} - h_\mathrm{max} \right)$  \\
       Note that $h_\mathrm{max} \geq \sqrt{3}/2$, which implies that $h_e < \frac{1}{2}\left( \sqrt{3\left( 1-h_\mathrm{max}^2 \right)} - h_\mathrm{max} \right) \leq 0$. Now $h_e \leq  0 \implies h^+_e = 0 \implies \theta_e = 0 \implies \sin \theta_e = 0$. Consequently, the bound \eqref{eqn_strongBoundEdgeBip}, with $e = \{i,j\}$, simplifies to
            \begin{align}
                \label{eqn_simpleLocalB}
                \bra{\phi} H_e \ket{\phi} \geq 1 + \prod_{k \in N(i) \setminus \{j\}} \cos \theta_{ik} \prod_{k \in N(j) \setminus \{i\}} \cos \theta_{kj}. 
            \end{align}
            Let $E' := \setFunct{ \{i,k\} }{k \in N(i) \setminus \{j\}}$. By \eqref{eqn_monogamyStar}, at most one $e' \in E'$ can satisfy $h_{e'} > h_\mathrm{max} \geq \sqrt{3}/2$. If there is one such $e'$, then by \eqref{eqn_extraCondition}, all $\widetilde{e} \in E' \setminus \{e'\}$ satisfy $h_{\widetilde{e}} \leq 0 \implies \theta_{\widetilde{e}} = 0 \implies \cos{\theta_{\widetilde{e}}} = 1$. Note also that $h_{e'} > h_\mathrm{max}$ implies that $\theta_{e'} = \arcsin{\sqrt{\theta^*}}$, see \eqref{eqn_thetaBipartSet}.
            Hence, $\prod_{k \in N(i) \setminus \{j\}} \cos \theta_{ik} = \cos{\theta_{e'}} = \cos{\arcsin{\sqrt{\theta^*}}} = \sqrt{1 - \theta^*}$. \\
            Alternatively, if all $e' \in E'$ satisfy $h_{e'} \leq h_\mathrm{max}$, we have that all $\theta_{e'} = \arcsin{\sqrt{\Theta \mleft(h^+_{e'} \mright)}}$. Thus, we can proceed as in \eqref{eqn_boundProdTheta2}, and derive $\prod_{k \in N(i) \setminus \{j\}} \cos \theta_{ik} \geq  \sqrt{1-\Theta\mleft(1-h^+_e \mright)} = \sqrt{1 - \Theta(1)}$.

            To summarize both cases:
            \begin{align}
                \label{eqn_prodLB3}
                \prod_{k \in N(i) \setminus \{j\}} \cos \theta_{ik} \geq \min\left\{ \sqrt{1 - \theta^*}, \sqrt{1 - \Theta(1)} \right\} = \sqrt{1 - \max\{\theta^*, \Theta(1)\}},
            \end{align}
            and similarly, $\prod_{k \in N(j) \setminus \{i\}} \cos \theta_{kj} \geq \sqrt{1 - \max\{\theta^*, \Theta(1)\}}$.
            Combining  \eqref{eqn_simpleLocalB} and \eqref{eqn_prodLB3} yields that $\bra{\phi} H_e \ket{\phi} \geq 2- \max\{\theta^*, \Theta(1) \}$. Consequently, \eqref{eqn_hMinimization} can be lower bounded by
            \begin{align}
                \label{eqn_eHeB1}
                \frac{\bra{\phi} H_e \ket{\phi}}{2+2h_e} \geq \frac{2- \max\{\theta^*, \Theta(1) \}}{2+2h_e} > \frac{2- \max\{\theta^*, \Theta(1) \}}{2+\sqrt{3\left( 1-h_\mathrm{max}^2 \right)} - h_\mathrm{max}}.
            \end{align}
            Here, the last inequality is due to the assumption $h_e < \frac{1}{2}\left( \sqrt{3\left( 1-h_\mathrm{max}^2 \right)} - h_\mathrm{max} \right)$, given by case 1. It follows by \eqref{eqn_eHeB1} that \eqref{eqn_hMinimization} is well-defined. Observe that 
            \begin{align}
                \label{eqn_sqrtHRound}
                \sqrt{3\left( 1-h_\mathrm{max}^2 \right)} - h_\mathrm{max} \leq 0 \implies \frac{2-\Theta(1)}{2+\sqrt{3\left( 1-h_\mathrm{max}^2 \right)} - h_\mathrm{max} } \geq \frac{2-\Theta(1)}{2} = \delta_\Theta(0).
            \end{align}
            For the inequality, we have used that $\Theta(1) \leq 1$, see \eqref{eqn_thetaProp}. The equality follows from the definition of $\delta_\Theta$, see \eqref{eqn_deltaFunc}. By combining \eqref{eqn_eHeB1} and \eqref{eqn_sqrtHRound}, it follows that
            \begin{align*}
                \frac{\bra{\phi} H_e \ket{\phi}}{2+2h_e} \geq \min\left\{ \delta_\Theta(0), \frac{2- \theta^*}{2+\sqrt{3\left( 1-h_\mathrm{max}^2 \right)} - h_\mathrm{max}} \right\}.
            \end{align*}
    \item $h_e \in \mleft[ \frac{1}{2}\left( \sqrt{3\left( 1-h_\mathrm{max}^2 \right)} - h_\mathrm{max} \right), h_\mathrm{max}\mright]$ \\
        Let $e'$ be an edge adjacent to $e$, and consider $h_{e'}$. If $h_{e'} > h_\mathrm{max}$, then by     \eqref{eqn_hSolutionUB}, $h_e < \frac{1}{2}\left( \sqrt{3\left( 1-h_\mathrm{max}^2 \right)} - h_\mathrm{max} \right)$, which contradicts case 2. Thus, it holds that for all edges $e'$ adjacent to $e$, $h_{e'} \leq h_\mathrm{max}$. Now, by \eqref{eqn_thetaBipartSet},
        $h_{e'} \leq h_\mathrm{max} \implies \theta_{e'} = \arcsin{\sqrt{\Theta\mleft(h^+_{e'}\mright)}}$.  Additionally, $\theta_{e} = \arcsin{\sqrt{\Theta\mleft(h^+_e\mright)}}$. Hence, the conditions of \Cref{lemma_boundHsimplify} are satisfied, which implies that 
        \begin{align*}
                 \frac{\bra{\phi} H_e \ket{\phi}}{2+2h_e} \geq \delta_{\Theta}(h_e) \geq \min_{x \in \left[ \frac{1}{2}\left( \sqrt{3\left( 1-h_\mathrm{max}^2 \right)} - h_\mathrm{max} \right), h_\mathrm{max} \right]} \delta_{\Theta}(x) = \min_{x \in \left[ 0, h_\mathrm{max} \right]} \delta_{\Theta}(x).
        \end{align*}
        Here, the equality is due to \eqref{eqn_simpleThetaObs}, and the fact that $\frac{1}{2}\left( \sqrt{3\left( 1-h_\mathrm{max}^2 \right)} - h_\mathrm{max} \right) \leq 0$, since $h_\mathrm{max} \geq \sqrt{3}/2$.
        \item $h_e \in (h_\mathrm{max},1]$ \\
        Note that $h_e > h_\mathrm{max} \geq \sqrt{3}/2$. Hence, it follows from \eqref{eqn_extraCondition} that for all edges $e'$ adjacent to $e$, $h_{e'} \leq 0 \implies h^+_{e'} = 0 \implies \theta_{e'} = 0 \implies \cos{\theta_{e'}} = 1$. Therefore, \eqref{eqn_strongBoundEdgeBip} simplifies to $\bra{\phi} H_e \ket{\phi} \geq 2+2 \sin{\theta_{e}}$. Since $h_e > h_\mathrm{max}$, it follows by \eqref{eqn_thetaBipartSet} that $\theta_e = \arcsin{\sqrt{\theta^*}}$, which implies that $\sin{\theta_{e}} = \sqrt{\theta^*}$. Consequently, \eqref{eqn_hMinimization} can be lower bounded by
        \begin{align*}
            \frac{\bra{\phi} H_e \ket{\phi}}{2+2h_e} \geq \frac{2+2 \sin \theta_e}{2+2h_e} = \frac{2+2\sqrt{\theta^*}}{2+2h_e} \geq \frac{2+2\sqrt{\theta^*}}{4} = \frac{1+ \sqrt{\theta^*}}{2},
        \end{align*}
        where the second inequality is due to $h_e \leq 1$.
    \end{mycases}
   By combining the three cases and noting that $\delta_\Theta(0) \geq  \min_{x \in \left[ 0, h_\mathrm{max} \right]} \delta_{\Theta}(x)$, it follows that the value $\frac{\bra{\phi} H_e \ket{\phi}}{2+2h_e}$ is at least \eqref{eqn_ratioExprBig}, which proves the theorem.
\end{proof}
The following result shows that the optimization problem  $\min_{x \in \left[ 0, h_\mathrm{max} \right]} \delta_{\Theta}(x)$ in \eqref{eqn_ratioExprBig}, for $\delta_\Theta$ as in \eqref{eqn_deltaFunc}, simplifies if $\Theta(x) = Rx$ for some $R \in [0,1/2]$. Recall from \Cref{lemma_linFunctionInOmega} that $Rx \in \mathcal{A}$.
\begin{lemma}
    \label{lemma_simplifyThetaMin}
    Let $\Theta(x) = Rx$, with $R \in [0,1/2]$. Let $h_\mathrm{max} \in \left[ \sqrt{3}/2,1 \right]$. We have that
    \begin{align}
        \label{eqn_desirecConc}
        \min_{x \in \left[0, h_\mathrm{max} \right]} \delta_{\Theta}(x) = \min \left\{ \delta_{\Theta}(0),\delta_{\Theta}\mleft(h_\mathrm{max} \mright) \right\}.
    \end{align}
\end{lemma}
\begin{proof}
In case $R = 0$, $\delta_{\Theta}(x) = 2/(2+2x) \implies  \min_{x \in \left[0, h_\mathrm{max} \right]} \delta_{\Theta}(x) = \delta_{\Theta}\left( h_\mathrm{max} \right)$. 
Let $R \in (0,1/2]$.  The derivative of $\delta_\Theta(x)$ on the interval $(0,1]$, is given by
\begin{align}
    \label{eqn_derivDelta}
        \delta'_\Theta(x) = \frac{f_R(x)}{2\,{\left(x+1\right)}^2\,\sqrt{R\,x\,\left(R\,x-R+1\right)}},
\end{align}
where $f_R(x) := \left(3\,R^2-R\right)\,x+R - R^2+ \left(2\,R-2\right)\,\sqrt{R\,x\,\left(R\,x-R+1\right)}$. A stationary point $x^*\in (0,1]$ of $\delta_\Theta$ satisfies
 $f_R(x^*) = 0$. This is equivalent to
\begin{align}
      \label{eqn_quadraticEquationX}
        \left( \left(3\,R^2-R\right)\,x^*+R - R^2 \right)^2 - \left(2\,R-2\right)^2 \, Rx^* \left(R\,x^*-R+1\right)  = 0.
    \end{align}
The solution of the quadratic equation \eqref{eqn_quadraticEquationX} on the interval $(0,1]$ is:
\begin{align*}
        x^* = \frac{R^{3}+2R^{2}-5R+2 -2\left(1-R\right)^{2}\sqrt{1-R-R^{2}}}{R\left(R+1\right)\left(5R-3\right)}.
\end{align*}
It can be verified that for any $R \in (0,1/2]$, $x^*$ is well-defined. Since the stationary point $x^*$ is unique and $\delta_{\Theta}$ is continuous, we have that
\begin{align}
    \label{eqn_deltaMinAtTP}
     \min_{x \in \left[0, h_\mathrm{max} \right]} \delta_{\Theta}(x) = \min \left\{ \delta_{\Theta}(0), \, \delta_{\Theta}(x^*), \, \delta_{\Theta}\mleft(h_\mathrm{max} \mright) \right\}.
\end{align}
By inspecting \eqref{eqn_derivDelta}, it can be seen that there exists a positive $\varepsilon$, dependent on $R$, such that $\delta'_{\Theta}(x) > 0$ for all $x \in (0,\varepsilon]$, which implies that $\delta_{\Theta}(0) < \delta_{\Theta}(x^*)$. By combining \eqref{eqn_deltaMinAtTP} and $\delta_{\Theta}(0) < \delta_{\Theta}(x^*)$, \eqref{eqn_desirecConc} follows.
\end{proof}
We now provide input parameters for \Cref{alg:bipartiteG2}, such that its approximation ratio is at least $\ratioBipart{}$.
\begin{theorem}
    \label{thm_LBbipartAlg}
    \Cref{alg:bipartiteG2}, with inputs $\Theta(x) = 0.367x$, $h_\mathrm{max} = 0.876$ and $\theta^* = 2/5$, achieves an approximation ratio of at least $\ratioBipartUnrounded{}$.
\end{theorem}
\begin{proof}
The given parameters satisfy the requirements of \Cref{alg:bipartiteG2}. By combining \Cref{thm_improvedBipartRatio} and \Cref{lemma_simplifyThetaMin}, we find that the approximation ratio of  \Cref{alg:bipartiteG2} (with the given inputs) is at least
\begin{align}
\label{eqn_difficultMinP}
        \min\left\{ \frac{2-\theta^*}{2+\sqrt{3\left( 1-h_\mathrm{max}^2 \right)} - h_\mathrm{max}},  \,\, \delta_{\Theta}(0), \, \delta_{\Theta}\mleft( h_\mathrm{max}\mright), \,\, \frac{1+\sqrt{\theta^*}}{2} \right\}.
\end{align}
Using a computer, it can be verified that \eqref{eqn_difficultMinP} is at least $\ratioBipartUnrounded{}$.
\end{proof}

Let
\begin{align}
    \label{eqn_optRatioR}
    r^* := \max_{\Theta \in \mathcal{A}, h_\mathrm{max} \in \left[ \sqrt{3}/2,1 \right], \theta^* \in [0,1]} \min\left\{ \frac{2-\theta^*}{2+\sqrt{3\left( 1-h_\mathrm{max}^2 \right)} - h_\mathrm{max}},  \min_{x \in \left[ 0, h_\mathrm{max} \right]} \delta_{\Theta}(x), \frac{1+\sqrt{\theta^*}}{2} \right\}
\end{align}
denote the optimal approximation ratio of \Cref{alg:bipartiteG2}, with respect to the parameters $\Theta$, $h_\mathrm{max}$ and $\theta^*$. \Cref{thm_LBbipartAlg} proves that $r^* \geq \ratioBipartUnrounded{}$. The following result provides an upper bound on $r^*$ (similar to \Cref{lemma_UB_triangleFree} for \Cref{alg:triangleFree}).

\begin{theorem}
    \label{thm_upperBoundRStar}
    The optimal approximation ratio of \Cref{alg:bipartiteG2}, given by $r^*$ as in \eqref{eqn_optRatioR}, satisfies
    \begin{align}
        \label{eqn_bipartUBproveR}
        \ratioBipartUnrounded{} \leq r^* < \ratioBipartUB{}
    \end{align}
\end{theorem}
\begin{proof}
    The lower bound on $r^*$ is due to \Cref{thm_LBbipartAlg}. For the upper bound, note first that for any $\Theta \in \mathcal{A}$, $h_\mathrm{max} \in \left[ \sqrt{3}/2,1 \right]$ and $\theta^* \in [0,1]$, we have
    \begin{equation}
        \label{eqn_rUBRewrite1}
    \begin{aligned}
        &\min\left\{ \frac{2-\theta^*}{2+\sqrt{3\left( 1-h_\mathrm{max}^2 \right)} - h_\mathrm{max}},  \min_{x \in \left[ 0, h_\mathrm{max} \right]} \delta_{\Theta}(x), \frac{1+\sqrt{\theta^*}}{2} \right\} \\
        &\leq \min_{x \in \left[ 0, h_\mathrm{max} \right]} \delta_{\Theta}(x) \leq \min_{x \in \left[ 0, \sqrt{3}/2 \right]} \delta_{\Theta}(x) \leq \min\left\{ \delta_{\Theta}(0), \, \delta_{\Theta}\mleft(\frac{2-\sqrt{3}}{2} \mright), \, \delta_{\Theta} \mleft( \frac{\sqrt{3}}{2} \mright) \right\}.
    \end{aligned}
    \end{equation}
    Here, the second inequality follows from $h_\mathrm{max} \in \left[ \sqrt{3}/2,1 \right]$, which implies that $\left[ 0, \sqrt{3}/2 \right] \subseteq \left[ 0, h_\mathrm{max} \right]$. The last inequality follows from the fact that $\delta_\Theta(y) \geq \min_{x \in \left[ 0, \sqrt{3}/2 \right]} \delta_{\Theta}(x)$ for any $y \in [0, \sqrt{3}/2]$. By combining \eqref{eqn_rUBRewrite1} with the definition of $r^*$, see \eqref{eqn_optRatioR}, we have that
    \begin{align}
        r^* \leq \max_{\Theta \in \mathcal{A}} \min\left\{ \delta_{\Theta}(0), \, \delta_{\Theta}\mleft(\frac{2-\sqrt{3}}{2} \mright), \, \delta_{\Theta} \mleft( \frac{\sqrt{3}}{2} \mright) \right\}.
        \label{eqn_simpleUBinZ}
    \end{align}
    Observe that \eqref{eqn_simpleUBinZ} is fully determined by $z := \left( z_1,z_2,z_3 \right) = \left(\Theta\mleft( \frac{2-\sqrt{3}}{2} \mright), \Theta\mleft( \frac{\sqrt{3}}{2} \mright), \Theta\mleft(1 \mright) \right)$, i.e.,
    \begin{align}
        \label{eqn_ExtraElabor}
        \left\{ \delta_{\Theta}(0), \, \delta_{\Theta}\mleft(\frac{2-\sqrt{3}}{2} \mright), \, \delta_{\Theta} \mleft( \frac{\sqrt{3}}{2} \mright) \right\} = \left\{ 1- \frac{z_3}{2}, \frac{2-z_2+2\sqrt{z_1(1-z_2)}}{4-\sqrt{3}}, \, \frac{2-z_1+2 \sqrt{z_2(1-z_1)}}{2+\sqrt{3}}\right\},
    \end{align}
    which follows from the definition of $\delta_\Theta$, see \eqref{eqn_deltaFunc}. By the properties of functions in $\mathcal{A}$, see \eqref{eqn_thetaProp}, we have
\begin{align*}
     1-z_3 = 1-\Theta(1) =  \min_{x \in [0,1]} \left( 1-\Theta(x) \right) \left( 1-\Theta(1-x) \right) &\leq \left(1-\Theta\mleft( \frac{2-\sqrt{3}}{2} \mright) \right)\left( 1-\Theta\mleft( \frac{\sqrt{3}}{2} \mright) \right) \\
     &=(1-z_1)(1-z_2).
\end{align*}
Additionally, since $\Theta$ is an increasing function, $0 = \Theta(0) \leq z_1 \leq z_2 \leq z_3 = \Theta(1) \leq 1$. We define
\begin{align*}
    F := \setFunct{ z \in \mathbb{R}^3}{ 1-z_3 \leq (1-z_1)(1-z_2), \, 0 \leq z_1 \leq z_2 \leq z_3 \leq 1}
\end{align*}
as the set of $z$ that satisfy the previously derived constraints. Using \eqref{eqn_simpleUBinZ}, \eqref{eqn_ExtraElabor} and $F$, we derive the following upper bound on $r^*$, in terms of $z$:
\begin{align}
    \label{eqn_rStarRewrite}
    r^* \leq \max_{z \in F} \, \min\left\{ 1- {\frac{z_3}{2}}, \frac{2-z_2+2\sqrt{z_1(1-z_2)}}{4-\sqrt{3}}, \, \frac{2-z_1+2 \sqrt{z_2(1-z_1)}}{2+\sqrt{3}}\right\}.
\end{align}
We define $r := \ratioBipartUB{}$. It follows by \eqref{eqn_rStarRewrite} that $r^* < r$ is implied by the inconsistency of the following system of equations:
\begin{align}
    \label{eqn_systemEqZ}
    z \in F, \,\, 1- \frac{z_3}{2} \geq r, \,\,  \frac{2-z_2+2\sqrt{z_1(1-z_2)}}{4-\sqrt{3}} \geq r, \,\, \frac{2-z_1+2 \sqrt{z_2(1-z_1)}}{2+\sqrt{3}} \geq r.
\end{align}
We will show that \eqref{eqn_systemEqZ} is inconsistent. Observe that $1-z_3/2 \geq r \implies z_3 \leq 2(1-r)$. We may assume without loss of generality that any solution to \eqref{eqn_systemEqZ}, if it exists, satisfies $z_3 = 2(1-r)$. Indeed, if $(z_1,z_2,z_3)$ is a solution to \eqref{eqn_systemEqZ}, also $(z_1,z_2,2(1-r))$ is a solution to \eqref{eqn_systemEqZ}. Thus, the inconsistency of \eqref{eqn_systemEqZ} is equivalent to the inconsistency of the following system of equations:
    \begin{align}
    \label{eqn_equationCases}
    \begin{cases}
        0 \leq z_1 \leq z_2 \leq 2(1-r) \\
        (1-z_1)(1-z_2) \geq 2r-1\\
        2\sqrt{z_1(1-z_2)} - z_2 \geq \left(4-\sqrt{3} \right)r-2 \\
        2\sqrt{z_2(1-z_1)} - z_1 \geq \left( 2+ \sqrt{3} \right) r -2,
        \end{cases}
    \end{align}
    where the first two lines of \eqref{eqn_equationCases} ensure that $z \in F$. \Cref{lemma_systemNoSolutions} in \Cref{section_extraResults} proves that \eqref{eqn_equationCases} is inconsistent. Hence, also \eqref{eqn_systemEqZ} is inconsistent, which implies that $r^* < r$, proving \eqref{eqn_bipartUBproveR}.
\end{proof}

\section{Conclusions and future work}
\label{section_conclusions}
In this paper, we study classical QMC approximation algorithms, for general, triangle-free, and bipartite graphs. For triangle-free and bipartite graphs, we introduce new approximation algorithms. We prove that the algorithms achieve approximation ratios of at least $\ratioGeneral{}$ (general graphs) $\ratioNoTri{}$ (triangle-free graphs) and $\ratioBipart{}$ (bipartite graphs) respectively, which are the current best ratios for their respective graph classes.

The key part of the analysis of the algorithm for general graphs is showing that a particular vector induced by the used QMC SDP relaxation is contained in the matching polytope. We show this by introducing a graph parameter $c(G,k)$, see \eqref{eqn_cFunctionDef}, for SDP relaxation level $k \in \mathbb{N}$, and verifying that $c(G,k) \leq \lfloor s /2 \rfloor$ for all graphs $G$ on $s$ vertices. We establish properties of $c(G,k)$ that greatly reduce the required computation time of verifying this inequality (see \Cref{lemma_vertexCover,lemma_gClassification}), and prove (using a computer) that $c(G,s) \leq \lfloor s/2 \rfloor$ holds for all graphs on $s$ vertices, where $s$ is odd and $s \leq 13$. As future work, it would be interesting to determine if this extends to odd values of $s > 13$, which, if so, results in an improved approximation ratio. A possible starting point in this direction is to consider the 5-cycle, which is the smallest graph for which we have no analytical proof of $c(G,2) \leq \lfloor s /2 \rfloor$ (see \Cref{table_lemmaReductionG}).

The studied QMC approximation algorithms for triangle-free and bipartite graphs both require a function $\Theta \in \mathcal{A}$, see \eqref{eqn_thetaProp}, as parameter. For the triangle-free algorithm, we provide a function $\Theta$ that achieves an approximation ratio that is 0.00009 below the optimal ratio (see \Cref{lemma_UB_triangleFree}). For the bipartite algorithm, the larger gap of 0.0177 (see \Cref{thm_upperBoundRStar}) motivates further study. For both algorithms, an optimal $\Theta \in \mathcal{A}$ does not achieve an approximation ratio of 0.956, which is the optimal QMC approximation ratio under UGC and a related conjecture, see \cite{HNPTW21}. It is therefore interesting to investigate if the constraint $\Theta \in \mathcal{A}$ can be relaxed.
\bibliographystyle{alpha}
\newcommand{\etalchar}[1]{$^{#1}$}

\appendix

\section{Additional lemmas and proofs}
\label{section_extraResults}
We require (a part of) \cite[Lem.~1]{PT22}, to prove \Cref{lemma_triFreeMajorize}, which in turn helps to prove \Cref{lemma_gClassification}. 
\begin{lemma}[\cite{PT22}]
    \label{lemma_ptTriLemma}
    Let $L \in \feasLass{k}{n}$ with $k \geq 2$ and $n \geq 3$. Let $i,j,\ell \in [n]$, and consider the values $h_{ij}$, $h_{i \ell}$, $h_{j \ell}$ associated to $L$ as in \Cref{def_sdpValues}. We have that $h_{ij} + h_{i \ell} + h_{j \ell} \leq 0$.
\end{lemma}
\noindent To compare \Cref{lemma_ptTriLemma} with \cite[Lem.~1]{PT22}, note that the variables $s_{ij}$ in \cite{PT22} satisfy $s_{ij} = -h_{ij}$.
\begin{lemma}
    \label{lemma_triFreeMajorize}
    Let $G \in \mathcal{G}_s$ with $s \geq 3$, and let $k \geq 2$. There exists a triangle-free graph $G' \in \mathcal{G}_s$, satisfying $c(G,k) \leq c\left(G',k \right)$, see \eqref{eqn_cFunctionDef}.
\end{lemma}
\begin{proof}
    If $G = (V,E)$ is triangle-free, the result follows directly by taking $G' = G$. If $G$ is not triangle-free, we may assume without loss of generality that the edges $e_1, e_2, e_3 \in E$ form a triangle in $G$. Consider an optimal solution of the SDP defining $c(G,k)$, with values $(h_e)_{e \in E}$ as in \Cref{def_sdpValues}. By \Cref{lemma_ptTriLemma}, $\sum_{i = 1}^3 h_{e_i} \leq 0$. Thus, there is some $e \in \{e_1, e_2, e_3\}$ for which $h_e \leq 0$. Consider the graph obtained after removing from $G$ this edge $e$, which is given by $G[E \setminus e]$. Observe that $G[E \setminus e]$ contains strictly less triangles than $G$, and satisfies $c(G,k) \leq c(G[E \setminus e],k)$. Repeating the procedure if necessary, this concludes the proof.
\end{proof}
\noindent We now prove \Cref{lemma_gClassification}.
\Gclassification*
\begin{proof} 
The direction $\Rightarrow$ is trivial, since $\mathbf{G} \subseteq \mathcal{G}_s$. For the reverse direction, observe first that $\max_{G \in \mathcal{G}_s} c(G,k) \geq \lfloor s / 2 \rfloor$ (it is straightforward to verify that the graph on $\lfloor s /2 \rfloor$ disjoint edges, denoted by $G'$, satisfies $c(G',k) = \lfloor s /2 \rfloor$). Consequently,
\begin{align*}
    \left\lfloor \frac{s}{2} \right\rfloor \leq \max_{G \in \mathcal{G}_s} c(G,k) = \max\left\{ \max_{G \in \mathbf{G}} c(G,k), \max_{G \in \mathcal{G}_s \setminus \mathbf{G}} c(G,k)\right\} \leq \max\left\{ \left\lfloor \frac{s}{2} \right\rfloor,  \max_{G \in \mathcal{G}_s \setminus \mathbf{G}} c(G,k)\right\}.
\end{align*}
Thus, it remains to prove that all $G \in \mathcal{G}_s \setminus \mathbf{G}$ satisfy $c(G,k) \leq \lfloor s /2 \rfloor$. Due to \Cref{lemma_triFreeMajorize} it suffices to consider $G \in \mathcal G_s \setminus \mathbf G$ that are triangle-free. We distinguish the following cases, corresponding to which of the four properties of \Cref{lemma_gClassification} are not satisfied by $G$.
\begin{mycases}
    \item If $G$ is disconnected, let graphs $(G^{i} = (V^i,E^i))_{i \in [p]}$ form the connected components of $G$ for some $p \in \mathbb{N}$. Item \ref{item_subgraphDecomp} from \Cref{lemma_vertexCover} yields
\begin{align*}
    c\mleft(G,k \mright) \leq \sum_{i=1}^p c\mleft( G[E^i],k \mright) \leq \sum_{i=1}^p \left\lfloor \frac{|V^i | }{ 2} \right\rfloor \leq \left\lfloor \frac{\sum_{i=1}^p |V^i| }{2} \right\rfloor = \left\lfloor \frac{s}{2} \right \rfloor.
\end{align*}
If $G$ is connected but not biconnected, there exists a partition $\{E^1, E^2\}$ of $E(G)$ such that $G[E^i]$, $i \in [2]$, is a graph on $s_i$ vertices, with $s_1+s_2 = s+1$. Since $s+1$ is even, the numbers $s_i$ are either both even or both odd. If they are both odd, we find
\begin{align*}
c\mleft(G,k \mright) \leq \sum_{i=1}^2 c\mleft(G\left[ E^i \right],k \mright) \leq \left\lfloor \frac{s_1}{2}\right\rfloor + \left\lfloor \frac{s_2}{2}\right\rfloor = \frac{s_1-1}{2} + \frac{s_2-1}{2} = \left\lfloor \frac{s}{2}\right\rfloor.
\end{align*}
If the numbers $s_i$ are both even, we proceed as follows: Since $G$ is connected but not biconnected, the graphs $G\left[ E^i \right]$, $i \in [2]$, must have precisely one common vertex $v \in V\left( G \right)$. Let $E_v \subseteq E\left( G \right)$ be the set of edges adjacent to $v$, and consider the partition $\left\{E^1 \setminus E_v, E^2 \setminus E_v, E_v \right\}$ of $E(G)$. Observe that $G[E^i \setminus E_v]$ is a graph on at most $s_i - 1$ vertices. Also note that $G\left[ E_v \right]$ is a star graph, which implies by \eqref{eqn_monogamyStar} that $c(G[E_v],k) \leq 1$. We have
\begin{align*}
    c\mleft(G,k \mright) \leq c\mleft(G[E^1 \setminus E_v],k \mright) + c\mleft(G[E^2 \setminus E_v], k \mright)+ c\mleft( G[E_v],k \mright) \leq \left\lfloor \frac{s}{2}\right\rfloor.
\end{align*}
If $G$ is bipartite, then its vertex cover number satisfies $\tau(G) \leq \lfloor s /2 \rfloor$. Hence, by  Item \ref{item_tauBound} of \Cref{lemma_vertexCover}, $c\mleft(G,k \mright)\leq \lfloor s /2\rfloor$. 
\item If $G$ contains a vertex of degree 1 and the number of  vertices $s = 3$, then $G$ is a star graph on one or two edges. In both cases, by \eqref{eqn_monogamyStar}, $c\mleft(G,k \mright) \leq 1 = \lfloor s / 2\rfloor$. If $G$ contains a vertex of degree 1 and $s \geq 4$, Item \ref{item_vertex1B} of \Cref{lemma_vertexCover} shows that $c\mleft(G,k \mright) \leq 1 + \max_{G \in \mathcal{G}_{s-2}} c(G,k)$. By the assumption that  $\max_{G \in \mathcal{G}_{s^\prime}} c\mleft(G,k \mright) = \lfloor s^\prime /2 \rfloor$ for $2 \leq s' < s$, it follows that $c\mleft(G,k \mright) \leq 1 + \max_{G \in \mathcal{G}_{s-2}} c(G,k) \leq 1 +\lfloor (s-2)/2 \rfloor = \lfloor s / 2 \rfloor$. If $G$ has a vertex $i$ satisfying $\degr{i} > (s-1)/2$ (equivalently, $\degr{i} \geq (s+1)/2$), note that the vertices in $N(i)$ are pairwise non-adjacent since $G$ is triangle-free. This implies that $V( G ) \setminus N(i)$ is a vertex cover of $G$. Then, Item \ref{item_tauBound} from \Cref{lemma_vertexCover} implies that $c\mleft(G,k \mright) \leq \tau \left( G \right) \leq | V( G ) \setminus N(i) | = | V( G )| - |N(i) | = s - \degr{i} \leq s - (s+1)/2 = (s-1)/2 = \lfloor s /2 \rfloor$.

\item If $\left| E(G) \right| < s$, $G$ is either disconnected or a tree (and thus bipartite). In both cases, $c\mleft(G,k \mright)\leq \lfloor s /2 \rfloor$ (as proven in case 1). 

\item If $G = (V,E)$ has a stable set $S \subseteq V$ for which $N(S) := \cup_{i \in S} N(i)$  satisfies $|N(S)| \leq |S|$, we proceed as follows: let $E_{N(S)}$ be the set of edges adjacent to at least one vertex in $N(S)$. The set $N(S)$ is a vertex cover of $G\left[ E_{N(S)} \right]$, so that $c\mleft( G \left[E_{N(S)} \right],k \mright) \leq \tau\left( G \left[E_{N(S)} \right] \right) \leq \left| N(S) \right|$ (Item \ref{item_tauBound} from \Cref{lemma_vertexCover}). Observe that $G \left[E \setminus E_{N(S)} \right]$ is a graph on at most $s-|N(S)| - |S|$ vertices. We find
\begin{align*}
c\mleft( G,k \mright) \leq c\mleft( G \left[E_{N(S)} \right],k \mright) + c\mleft( G \left[E \setminus E_{N(S)} \right],k \mright) \leq |N(S)| + \left\lfloor \frac{s-|N(S)| - |S|}{2} \right\rfloor \leq \left\lfloor \frac{s}{2} \right\rfloor.
\end{align*}
For the last inequality, we used that $|N(S)| \leq |S|$.

\end{mycases}

We have considered all possible cases of $G \not\in \mathbf{G}$, which finishes the proof.
\end{proof}

The following lemma is used in the proof of \Cref{lemma_UB_triangleFree}.
\begin{lemma}
    \label{lemma_concaveMax}
    Let the function $f$ be defined as in \eqref{eqn_fpIneq}, $h$ as in \eqref{eqn_hFunction} and $r = \ratioUB{}$. We have that
    \begin{align}
        \label{eqn_maxValNegative}
        \max_{ p \in \left[ \frac{2 r  -1}{q(0)-1},1 \right]} f(h(r,p),p) < r.
    \end{align}
\end{lemma}
\begin{proof}
To simplify notation, we write $\ell := \frac{2 r  -1}{q(0)-1}$ and $g(p) := f(h(r,p),p)$. We first show that $g$ is concave. Using concavity, we determine a small interval in which the function $g$ attains its maximum. We then apply the Mean Value theorem to this interval, to prove \eqref{eqn_maxValNegative}. For intuition, \Cref{fig:concaveG_wideAxis} provides a plot of $g$.

We now prove that $g$ is concave by showing that $g''(p) \leq 0$. Recall from the proof of \Cref{lemma_UB_triangleFree}, the expressions of the positive numbers $\mu = 14/15$, $C_1 = (4+\mu (\pi-2))/ (2\pi)$ and $C_2 =  \left(1+3\mu/2 \right)$. The function $g'(p)$ can be computed by using the multivariate chain rule. To this end, denote by $D_i f$ the partial derivative of $f$, see \eqref{eqn_fpIneq}, with respect its $i$th argument, $i \in \{1,2\}$. For ease of notation, we write $h(p) = h(r,p)$, for $h$ as in \eqref{eqn_hFunction} and $r = \ratioUB{}$. We have, for $p \in (\ell,1]$,
\begin{align*}
    g'(p) =  h'(p) D_1 f\left( h(p),p \right) + D_2 f( h(p)&,p) = \frac{(2r-1) q(1/2)}{6pq(0) \sqrt{ \frac{2(p+2r-1)}{pq(0)}-1}} \left(\frac{C_{1}(1-2h(p))}{\sqrt{h(p)\left(1-h(p)\right)}}-1\right) \\
     &+ \frac{1}{3} \left[q(1/2) \left(1-\frac{ h(p)}{2}+C_1\sqrt{ h(p)(1-h(p))}\right) - C_2 \right].
\end{align*}
Observe that $g$ is not differentiable for $p = \ell$, since $h(\ell)=0$. However, it can be verified using a computer that $r > g(\ell) \approx 0.587$, so that it suffices to consider only $p \in (\ell,1]$. The function $g''(p)$ is given by
\begin{align*}
    g''(p) = -\frac{1}{3}{\left( \frac{2\,r-1}{q(0)}\right)}^2  q(1/2) \, C_1 \frac{k_1(p)}{ k_2(p)},
\end{align*}
where 
\begin{align*}
    &k_1(p)  :=  -6 u(p) +\frac{2}{C_1} {\left( \sqrt{u(p)} - u(p)\right)}^{3/2}+3\,\,\sqrt{ u(p) } +4\,{\left( u(p)\right)}^{3/2}, \\
    &k_2(p) := 4\,p^3\,{\left(\sqrt{u(p)} - u(p)\right)}^{3/2}\,{\left( u(p) \right)}^{3/2},
\end{align*}
and $u(p) := \frac{2-q(0)}{q(0)}+\frac{4r-2}{pq(0)}$. Since $u(p)$ is decreasing in $p$ for $p \in (\ell,1]$ (and $u(1) \approx 0.909$, as can be verified using a computer), we have
\begin{align}
    \label{eqn_zIneq}
   0 < u(1) \leq u(p) < u(\ell) = 1
\end{align}
for all $p \in (\ell,1]$, which implies that $\sqrt{u(p)} - u(p) > 0$ for $p \in (\ell,1]$. Thus, $k_1(p)$ and $k_2(p)$ are well-defined for all $p \in (\ell,1]$. Furthermore, \eqref{eqn_zIneq} shows that $k_2(p) > 0$ for $p \in (\ell,1]$.

For $k_1(p)$, we find
\begin{align*}
    k_1(p) \geq   -6 u(p) +3\sqrt{ u(p) } +4\,{\left( u(p)\right)}^{3/2}  \geq 0.
\end{align*}
For the first inequality, we have used that $\frac{2}{C_1}\,{\left( \sqrt{u(p)} - u(p)\right)}^{3/2} \geq 0$. The second inequality follows from the fact that 
\begin{align*}
    \min_{z \in [0,1]} \left(-6z + 3 \sqrt{z}+4z^{3/2} \right)= 0,
\end{align*}
since $-6z + 3 \sqrt{z}+4z^{3/2}$ is an increasing function, as can be shown by evaluating the derivative.

Hence, $k_1(p) \geq 0$ and $k_2(p) > 0$ for all $ p \in (\ell,1]$. Since $-\frac{1}{3}{\left( \frac{2\,r-1}{q(0)}\right)}^2 q(1/2) \, C_1 <0$, it follows that $g''(p) \leq 0$, which proves that $g$ is concave on $p \in (\ell,1]$. We will use concavity of $g$, in combination with the Mean Value theorem, to prove \eqref{eqn_maxValNegative}. Let $p_1 = 0.897$ and $p_2 = 0.898$. It can be verified (by computer) that $g'(p_1) > 0$ and $g'(p_2) < 0$. By concavity of $g$, it follows that $\max_{p \in [\ell,1]} g(p) = g(p^*)$, for some $p^* \in (p_1,p_2)$. Observe that $g$ is continuous and continuously differentiable on $(p_1,p_2)$. By the Mean Value theorem, we have
\begin{align}
    \label{eqn_rIneqProof}
    \max_{p \in [\ell,1]} g(p) = g(p^*) = g'(z)(p^*-p_1) + g(p_1) \leq g'(p_1) (p_2 - p_1) + g(p_1) < r,
\end{align}
for some $z \in (p_1,p_2)$. For the first inequality, we have used that $g$ is concave, so that $g'(p)$ is a decreasing function. The second inequality can be verified by computing the number $g'(p_1) (p_2 - p_1) + g(p_1)$ $(\approx 0.61391)$.
\end{proof}

\begin{figure}
    \centering
\begin{tikzpicture}[trim axis left, trim axis right]
\begin{axis}[
    axis lines=left,
    ylabel style={rotate=-90},
    xlabel={$p$},
    ylabel={$g(p)$},
    yticklabel style={/pgf/number format/.cd,
    fixed, fixed zerofill,   /pgf/number format/precision=3},
    xmin=0.7961, xmax=1,
    ymin=0.585, ymax=0.615,
    scaled ticks=false,
    ytick={0.585,0.6,0.615}, 
    legend pos=north west,
    extra y ticks={0.614},
    extra y tick style={grid=none,ticks=major,ticklabel pos=right},
    extra y tick labels={$r = \ratioUB{}$},
]

\addplot[mark=none, black, very thick,dotted] coordinates {(0.7961,0.61392) (1,0.61392)};

\addplot[
    color=black,
    thick
    ]
    coordinates {
        (0.79607,0.58655)(0.79676,0.59104)(0.79744,0.59278)(0.79812,0.59405)(0.79880,0.59510)(0.79948,0.59599)(0.80017,0.59678)(0.80085,0.59749)(0.80153,0.59814)(0.80221,0.59873)(0.80289,0.59928)(0.80358,0.59979)(0.80426,0.60027)(0.80494,0.60073)(0.80562,0.60115)(0.80630,0.60156)(0.80699,0.60195)(0.80767,0.60231)(0.80835,0.60267)(0.80903,0.60300)(0.80972,0.60332)(0.81040,0.60363)(0.81108,0.60393)(0.81176,0.60422)(0.81244,0.60449)(0.81313,0.60476)(0.81381,0.60501)(0.81449,0.60526)(0.81517,0.60550)(0.81585,0.60573)(0.81654,0.60596)(0.81722,0.60618)(0.81790,0.60639)(0.81858,0.60659)(0.81926,0.60679)(0.81995,0.60699)(0.82063,0.60717)(0.82131,0.60736)(0.82199,0.60753)(0.82267,0.60771)(0.82336,0.60788)(0.82404,0.60804)(0.82472,0.60820)(0.82540,0.60835)(0.82608,0.60850)(0.82677,0.60865)(0.82745,0.60880)(0.82813,0.60894)(0.82881,0.60907)(0.82949,0.60921)(0.83018,0.60934)(0.83086,0.60946)(0.83154,0.60959)(0.83222,0.60971)(0.83290,0.60983)(0.83359,0.60994)(0.83427,0.61005)(0.83495,0.61016)(0.83563,0.61027)(0.83631,0.61037)(0.83700,0.61048)(0.83768,0.61058)(0.83836,0.61067)(0.83904,0.61077)(0.83972,0.61086)(0.84041,0.61095)(0.84109,0.61104)(0.84177,0.61113)(0.84245,0.61121)(0.84313,0.61130)(0.84382,0.61138)(0.84450,0.61146)(0.84518,0.61153)(0.84586,0.61161)(0.84654,0.61168)(0.84723,0.61175)(0.84791,0.61182)(0.84859,0.61189)(0.84927,0.61196)(0.84995,0.61202)(0.85064,0.61209)(0.85132,0.61215)(0.85200,0.61221)(0.85268,0.61227)(0.85336,0.61233)(0.85405,0.61239)(0.85473,0.61244)(0.85541,0.61249)(0.85609,0.61255)(0.85677,0.61260)(0.85746,0.61265)(0.85814,0.61270)(0.85882,0.61274)(0.85950,0.61279)(0.86018,0.61284)(0.86087,0.61288)(0.86155,0.61292)(0.86223,0.61296)(0.86291,0.61300)(0.86360,0.61304)(0.86428,0.61308)(0.86496,0.61312)(0.86564,0.61316)(0.86632,0.61319)(0.86701,0.61322)(0.86769,0.61326)(0.86837,0.61329)(0.86905,0.61332)(0.86973,0.61335)(0.87042,0.61338)(0.87110,0.61341)(0.87178,0.61344)(0.87246,0.61346)(0.87314,0.61349)(0.87383,0.61351)(0.87451,0.61354)(0.87519,0.61356)(0.87587,0.61358)(0.87655,0.61360)(0.87724,0.61362)(0.87792,0.61364)(0.87860,0.61366)(0.87928,0.61368)(0.87996,0.61370)(0.88065,0.61372)(0.88133,0.61373)(0.88201,0.61375)(0.88269,0.61376)(0.88337,0.61377)(0.88406,0.61379)(0.88474,0.61380)(0.88542,0.61381)(0.88610,0.61382)(0.88678,0.61383)(0.88747,0.61384)(0.88815,0.61385)(0.88883,0.61386)(0.88951,0.61387)(0.89019,0.61387)(0.89088,0.61388)(0.89156,0.61389)(0.89224,0.61389)(0.89292,0.61390)(0.89360,0.61390)(0.89429,0.61390)(0.89497,0.61391)(0.89565,0.61391)(0.89633,0.61391)(0.89701,0.61391)(0.89770,0.61391)(0.89838,0.61391)(0.89906,0.61391)(0.89974,0.61391)(0.90042,0.61391)(0.90111,0.61390)(0.90179,0.61390)(0.90247,0.61390)(0.90315,0.61389)(0.90383,0.61389)(0.90452,0.61388)(0.90520,0.61388)(0.90588,0.61387)(0.90656,0.61387)(0.90724,0.61386)(0.90793,0.61385)(0.90861,0.61384)(0.90929,0.61384)(0.90997,0.61383)(0.91065,0.61382)(0.91134,0.61381)(0.91202,0.61380)(0.91270,0.61379)(0.91338,0.61378)(0.91406,0.61376)(0.91475,0.61375)(0.91543,0.61374)(0.91611,0.61373)(0.91679,0.61371)(0.91747,0.61370)(0.91816,0.61369)(0.91884,0.61367)(0.91952,0.61366)(0.92020,0.61364)(0.92089,0.61363)(0.92157,0.61361)(0.92225,0.61359)(0.92293,0.61358)(0.92361,0.61356)(0.92430,0.61354)(0.92498,0.61352)(0.92566,0.61350)(0.92634,0.61349)(0.92702,0.61347)(0.92771,0.61345)(0.92839,0.61343)(0.92907,0.61341)(0.92975,0.61339)(0.93043,0.61336)(0.93112,0.61334)(0.93180,0.61332)(0.93248,0.61330)(0.93316,0.61328)(0.93384,0.61325)(0.93453,0.61323)(0.93521,0.61321)(0.93589,0.61318)(0.93657,0.61316)(0.93725,0.61314)(0.93794,0.61311)(0.93862,0.61309)(0.93930,0.61306)(0.93998,0.61303)(0.94066,0.61301)(0.94135,0.61298)(0.94203,0.61296)(0.94271,0.61293)(0.94339,0.61290)(0.94407,0.61287)(0.94476,0.61285)(0.94544,0.61282)(0.94612,0.61279)(0.94680,0.61276)(0.94748,0.61273)(0.94817,0.61270)(0.94885,0.61267)(0.94953,0.61264)(0.95021,0.61261)(0.95089,0.61258)(0.95158,0.61255)(0.95226,0.61252)(0.95294,0.61249)(0.95362,0.61246)(0.95430,0.61243)(0.95499,0.61239)(0.95567,0.61236)(0.95635,0.61233)(0.95703,0.61230)(0.95771,0.61226)(0.95840,0.61223)(0.95908,0.61220)(0.95976,0.61216)(0.96044,0.61213)(0.96112,0.61209)(0.96181,0.61206)(0.96249,0.61202)(0.96317,0.61199)(0.96385,0.61195)(0.96453,0.61192)(0.96522,0.61188)(0.96590,0.61185)(0.96658,0.61181)(0.96726,0.61177)(0.96794,0.61174)(0.96863,0.61170)(0.96931,0.61166)(0.96999,0.61163)(0.97067,0.61159)(0.97135,0.61155)(0.97204,0.61151)(0.97272,0.61147)(0.97340,0.61143)(0.97408,0.61140)(0.97477,0.61136)(0.97545,0.61132)(0.97613,0.61128)(0.97681,0.61124)(0.97749,0.61120)(0.97818,0.61116)(0.97886,0.61112)(0.97954,0.61108)(0.98022,0.61104)(0.98090,0.61100)(0.98159,0.61095)(0.98227,0.61091)(0.98295,0.61087)(0.98363,0.61083)(0.98431,0.61079)(0.98500,0.61075)(0.98568,0.61070)(0.98636,0.61066)(0.98704,0.61062)(0.98772,0.61057)(0.98841,0.61053)(0.98909,0.61049)(0.98977,0.61044)(0.99045,0.61040)(0.99113,0.61036)(0.99182,0.61031)(0.99250,0.61027)(0.99318,0.61022)(0.99386,0.61018)(0.99454,0.61013)(0.99523,0.61009)(0.99591,0.61004)(0.99659,0.61000)(0.99727,0.60995)(0.99795,0.60991)(0.99864,0.60986)(0.99932,0.60982)(1.00000,0.60977)
    };
\end{axis}
\end{tikzpicture}
\caption{Plot of the function $g(p) = f(h(r,p),p)$, see \eqref{eqn_maxValNegative}, for $p \in [\ell,1]$, where $\ell = (2r-1)/ \left( q(0) -1 \right) \approx 0.7961$.} 
\label{fig:concaveG_wideAxis}
\end{figure}

The following lemma is used in the proof of \Cref{thm_upperBoundRStar}.
\begin{lemma} \label{lemma_systemNoSolutions}
    Let $r : = \ratioBipartUB{}$. Then, the system of equations \eqref{eqn_equationCases} is inconsistent.
\end{lemma}
\begin{proof}
    Define 
    \begin{align*}
        &A_1 := \setFunct{ z \in [0,1]^2}{  (1-z_1)(1-z_2) \geq 2r-1, \,  0 \leq z_1 \leq z_2 \leq 2(1-r) } \\
        &A_2 := \setFunct{ z \in [0,1]^2}{2\sqrt{z_1(1-z_2)} - z_2 \geq c_2, \,  0 \leq z_1 \leq z_2 \leq 2(1-r) } \\ &A_3 := \setFunct{z \in [0,1]^2}{ 2\sqrt{z_2(1-z_1)} - z_1 \geq c_3, \,  0 \leq z_1 \leq z_2 \leq 2(1-r) },
    \end{align*}
    where $c_2 := \left(4-\sqrt{3}  \right)r-2 \approx -0.109$, and $c_3 := \left( 2+ \sqrt{3} \right) r -2 \approx 1.112$. If $z$ is a solution to \eqref{eqn_equationCases}, then $z \in A := A_1 \cap A_2 \cap A_3$. We will show that $A = \emptyset$, which implies that \eqref{eqn_equationCases} is inconsistent. Proving $A = \emptyset$ directly is difficult, due to the nonlinearity of $A$. To circumvent this difficulty, we will define half-planes $H_i$, $i \in \{1,2,3\}$, that satisfy $A_i \subseteq H_i$. Therefore, the intersection $H := H_1 \cap H_2 \cap H_3 \cap \setFunct{z \in \mathbb{R}^2}{z \geq 0}$ satisfies $A \subseteq H$, so that $A = \emptyset$ follows from $H = \emptyset$. Since $H$ is a polytope, $H = \emptyset$ can be proven via Farkas' lemma \cite{farkas1902theorie}.

We define the half-planes $H_i$ as
\begin{align*}
    &H_1 := \setFunct{ z \in \mathbb{R}^2}{0.69z_1 + z_2 \leq 0.3330}, \,\, H_2 := \setFunct{z \in \mathbb{R}^2}{z_1 -0.183 z_2 \geq -0.0423}, \\
    &H_3 := \setFunct{z \in \mathbb{R}^2}{ -0.9z_1+z_2 \geq 0.3088}.
\end{align*}
We also define the functions
\begin{align*}
    f_1(x) := 0.69x + \frac{2 r + x - 2}{x - 1}, \, f_2(x) := \frac{\max\left\{ c_2+x,0 \right\}^2}{4(1-x)}-0.183x, \, f_3(x) := -0.9x+ \frac{\left(c_{3}+x\right)^{2}}{4\left(1-x\right)}, 
\end{align*}
that we use to prove the inclusions $A_i \subseteq H_i$.

To prove that $A_1 \subseteq H_1$, observe first that $z \in A_1 \implies z_2 \leq \frac{2 r + z_1 - 2}{z_1 - 1}$ and $z_1 \in [0,2(1-r)]$. For $z \in A_1$, we have
\begin{align*}
    0.69z_1 + z_2 \leq 0.69z_1 + \frac{2 r + z_1 - 2}{z_1 - 1} \leq \max_{x \in [0,2(1-r)]} f_1(x) = f_1(x^*) < 0.3330,
\end{align*}
where $x^*(\approx 0.016)$ is the stationary point of $f_1$ in the interval $[0,2(1-r)]$ (we omit the computation of $x^*$), and $f_1(x^*) \approx 0.332$. Thus, $A_1 \subseteq H_1$. 

For $A_2 \subseteq H_2$, we use that $z \in A_2 \implies z_1 \geq \frac{\max\left\{ c_2+z_2,0 \right\}^2}{4(1-z_2)}$ and $z_2 \in [0,2(1-r)]$. Then
\begin{equation}
\begin{aligned}    
    \label{eqn_minimizeF2}
   z_1 -0.183 z_2 \geq \frac{\max\left\{ c_2+z_2,0 \right\}^2}{4(1-z_2)} -0.183 z_2 \geq \min_{x \in [0,2(1-r)] } f_2(x). 
\end{aligned}
\end{equation}
To solve the minimization problem in \eqref{eqn_minimizeF2}, note that $f_2(x) = -0.183x$ for $x \in [0,-c_2]$. Thus, $\min_{x \in [0,-c_2] } f_2(x) = 0.183c_2 \approx -0.020$. For $x \in [-c_2,2(1-r)]$, we have $\min_{x \in [-c_2,2(1-r)] } f_2(x) = f_2(x^*) \approx -0.0422$. Here, $x^* ( \approx 0.323)$ is the stationary point of $f_2$ in $[-c_2,2(1-r)]$ (we omit the computation of $x^*$). We conclude that $\min_{x \in [0,2(1-r)] } f_2(x) = f_2(x^*) > -0.0423$, which implies by \eqref{eqn_minimizeF2} that $A_2 \subseteq H_2$. 

The proof of $A_3 \subseteq H_3$ is similar to the proof of $A_2 \subseteq H_2$. We have $z \in A_3 \implies z_2 \geq \frac{\left(c_{3}+z_1 \right)^{2}}{4\left(1-z_1\right)}$ and $z_1 \in [0,2(1-r)]$. Then
\begin{align*}
    -0.9z_1+z_2 \geq -0.9z_1 +\frac{\left(c_{3}+z_1 \right)^{2}}{4\left(1-z_1\right)} \geq \min_{x \in [0,2(1-r)] } f_3(x) = f_3(x^*) > 0.3088,
\end{align*}
where $x^* ( \approx 0.015)$ is the stationary point of $f_3$ in $[0,2(1-r)]$, and $f_3(x^*) \approx 0.309$. Thus, $A_i \subseteq H_i$ for all $i \in \{1,2,3\}$. 

Using slack variables $s$, we can write the intersection $H$ as follows: $H$ consists of all the $z \in \mathbb{R}^2$, $z \geq 0$, for which there exists an $s \in \mathbb{R}^3$, $s \geq 0$, such that $M\left[z^\top, s^\top \right]^\top = b$, where
\begin{align*}
    M := \begin{bmatrix}
0.69 & 1 & 1 & 0 & 0 \\
1 & -0.183 & 0 & -1 & 0 \\
-0.9 & 1 & 0 & 0 & -1 
\end{bmatrix} \text{ and } b := \begin{bmatrix}
    0.3330 \\ -0.0423 \\ 0.3088
\end{bmatrix}.
\end{align*}
By Farkas' lemma \cite{farkas1902theorie}, $H$ is empty if and only if there exists a vector $y \in \mathbb{R}^3$ satisfying $M^\top y \geq 0$ and $b^\top y < 0$. It can be verified that $y := [0.824,-1.447,-1.087]^\top$ satisfies these conditions. Thus, $H = \emptyset$, which implies that $A = \emptyset$ because $A \subseteq H$. Since any solution to \eqref{eqn_equationCases} is contained in $A$, it follows that \eqref{eqn_equationCases} is inconsistent.
\end{proof}

\section{Computational details of proof of \texorpdfstring{\Cref{lemma_5VertexBound}}{Lemma 8}}
\label{section_compuDetails}
\newcommand{\fileNameOne}{\texttt{triFreeBiconn\_s13.g6}}
\newcommand{\fileNameTwo}{\texttt{triFreeNonBipartite\_s13.g6}}

We provide computational details on the verification of $c(G,s) \leq \lfloor s /2 \rfloor$ for all $G \in \mathcal{G}_{s}$ for $s \in \{5,7,9,11, 13\}$. We performed the following steps on a laptop (16 GB RAM and Intel i7-1165G7 CPU), which required approximately 16 hours to run. Our code is available at
\begin{center}
    \githubLink.
\end{center}

We first use the software package \texttt{nauty} \cite{mckay2014practical} to generate the graphs in $\mathcal{G}_s$ that satisfy the properties \ref{eqn_graphClassific} to \ref{eqn_edgeMin} of \Cref{lemma_gClassification}. Then, for the case $s \in \{5,7,9\}$, we verify that $c(G,2) \leq \lfloor s /2 \rfloor$ for these graphs in $\mathcal{G}_s$ using SDP. Computing $c(G,2)$ can be done using the Pauli-based \ref{eqn_sdpRelax}, but, in fact, solving a relaxation of \ref{eqn_sdpRelax} based on the SWAP operators already suffices to establish the desired upper bound. (More precisely, we compute the first level of the QMC SDP relaxation based on the SWAP operators, see \cite{KPTTZ23,CEHKW23} and in particular the discussion in \cite[Sect.~5.1.2]{KPTTZ23}.)  

For $s = 11$, we consider the 26360 (see \Cref{table_lemmaReductionG}) remaining graphs $G_j = ([11],E_j)$, $j \in [26360]$ in the sequence as returned by \texttt{nauty}. We construct the corresponding Hamiltonians {recursively as}
\begin{align}
    \label{eqn_HG_construction}
    H_{G_{j+1}} = H_{G_{j}} + \sum_{e \in E_{j+1} \setminus E_j} H_e - \sum_{e \in E_{j} \setminus E_{j+1}} H_e.
\end{align}
Constructing $H_{G_{j+1}}$ using \eqref{eqn_HG_construction} is efficient since \texttt{nauty} returns a sequence of graphs where $E_j \approx E_{j+1}$. Additionally, \eqref{eqn_HG_construction} shows that 
\begin{align}
    \label{eqn_lambdaUB}
    \lambdaMax{H_{G_{j+1}}} \leq \lambdaMax{H_{G_{j}}} + \lambdaMax{ \sum_{e \in E_{j+1} \setminus E_{j}} H_e }.
\end{align}
Here, we have used that $H_e = (1/4)H_e^2 \succeq 0$. Generally, $\lambdaMax{ \sum_{e \in E_{j+1} \setminus E_j} H_e } \leq 4 \left| E_{j+1} \setminus E_j \right|$, and tighter bounds are possible if, for example, the edges $E_{j+1} \setminus E_j$ form a star graph. If the bound \eqref{eqn_lambdaUB} already proves that $c\mleft( G_{j+1},s \mright) = \lambdaMax{H_{G_{j+1}}}/2 - | E_{j+1} | \leq \lfloor s /2 \rfloor$, we do not carry out the computation of $\lambdaMax{H_{G_{j+1}}}$, nor the construction of $H_{G_{j+1}}$.  If the bound \eqref{eqn_lambdaUB} does not prove $c(G_{j+1},s) \leq \lfloor s /2 \rfloor$, we compute $\lambdaMax{H_{G_{j+1}}}$ with MATLABs \texttt{eigs} function. 

The case $s = 13$ proceeds similarly as the case $s = 11$, except we first discard some of the 9035088 graphs in $\mathcal{G}_{13}$ that satisfy properties \ref{eqn_graphClassific} to \ref{eqn_edgeMin} of \Cref{lemma_gClassification}, by arguing as follows: of these 9035088 graphs, 959842 of them do not satisfy property \ref{eqn_stableSetNeighb} with $|S| = 2$. We verify that these 959842 graphs do not satisfy property \ref{eqn_stableSetNeighb} in approximately 5 seconds. Of the now remaining 8075246 graphs, 1622184 of them satisfy $\tau(G) \leq 6$. We verify that these 1622184 graphs satisfy $\tau(G) \leq 6$ in approximately 20 seconds. By Item \ref{item_tauBound} of \Cref{lemma_vertexCover}, these graphs satisfy $c(G,13) \leq \tau(G) \leq 6$, so we may discard them. For the remaining 6453062 graphs $G$, we compute upper bounds on $\lambdaMax{H_G}$ in the manner described for the case $s = 11$.
\end{document}